\documentclass[sigconf]{acmart}

\renewcommand\footnotetextcopyrightpermission[1]{} 
\usepackage{booktabs} 

\usepackage{enumitem}
\usepackage[noend]{algpseudocode}
\usepackage{adjustbox}
\usepackage{subcaption}
\usepackage{multirow}
\usepackage{array}
\usepackage{algorithm}
\usepackage{algorithmicx}
\usepackage{flushend}

\theoremstyle{plain}
\newtheorem{heuristic}{Heuristic}

\theoremstyle{definition}
\newtheorem{proof}{Proof}

\algnewcommand\algorithmicinput{\textbf{INPUT:}}
\algnewcommand\INPUT{\item[\algorithmicinput]}
\algnewcommand\algorithmicoutput{\textbf{OUTPUT:}}
\algnewcommand\OUTPUT{\item[\algorithmicoutput]}

\newcommand{\pluseq}{\mathrel{+}=}

\setcopyright{none}

\begin{document}

\title{The Flexible Group Spatial Keyword Query}

\author{Sabbir Ahmad}
\affiliation{%
  \institution{Dept of Computer Science \& Eng\\Bangladesh Univ of Eng \& Tech}
  \city{Dhaka} 
  \state{Bangladesh} 
}
\email{ahmadsabbir@cse.buet.ac.bd}

\author{Rafi Kamal}
\affiliation{%
  \institution{Dept of Computer Science \& Eng\\Bangladesh Univ of Eng \& Tech}
  \city{Dhaka} 
  \state{Bangladesh} 
}
\email{rafikamalb@gmail.com}

\author{Mohammed Eunus Ali}
\affiliation{%
  \institution{Dept of Computer Science \& Eng\\Bangladesh Univ of Eng \& Tech}
  \city{Dhaka} 
  \state{Bangladesh} 
}
\email{eunus@cse.buet.ac.bd}

\author{Jianzhong Qi}
\affiliation{%
  \institution{University of Melbourne}
  \city{Melbourne} 
  \state{Australia} 
}
\email{jianzhong.qi@unimelb.edu.au}

\author{Peter Scheuermann}
\affiliation{%
  \institution{Northwestern University}
  \city{Illinois} 
  \state{USA} 
}
\email{peters@eecs.northwestern.edu}

\author{Egemen Tanin}
\affiliation{%
  \institution{University of Melbourne}
  \city{Melbourne} 
  \state{Australia} 
}
\email{etanin@unimelb.edu.au}

\renewcommand{\shortauthors}{S. Ahmad et al.}

\begin{abstract}
We present a new class of service for location based social networks, called the Flexible Group Spatial Keyword Query, which enables a group of users to collectively find a point of interest (POI) that optimizes an aggregate cost function combining both spatial distances and keyword  similarities. In addition, our query service allows users to consider the trade-offs between obtaining a sub-optimal solution for the entire group and obtaining an optimimized solution but only for a subgroup.

We propose algorithms to process three variants of the query: (i)  the group nearest neighbor with keywords query, which finds a POI that optimizes the aggregate cost function for the whole group of size $n$, (ii) the subgroup nearest neighbor with keywords query, which finds the optimal subgroup and a POI that optimizes the aggregate cost function for a given subgroup size $m$ ($m \leq n$), and (iii) the multiple subgroup nearest neighbor with keywords query, 
which finds optimal subgroups and corresponding POIs for each of the subgroup sizes in the range [$m$, $n$]. We design query processing algorithms based on branch-and-bound and best-first paradigms. 
Finally, we provide theoretical bounds and conduct extensive experiments with two real datasets which verify the effectiveness and efficiency of the proposed algorithms.
\end{abstract}


\maketitle

\section{Introduction}\label{sec:intro}


The \textit{group nearest neighbor} (GNN) query~\cite{papadias2005aggregate} and its variants, the flexible aggregate nearest neighbor (FANN)~\cite{li2011flexible} query and the consensus query~\cite{ali2015consensus} have been previously studied in the spatial database domain. Given a set $Q$ of $n$ queries and a dataset $D$, a GNN query finds the data object that minimizes the aggregate distance (e.g., sum or max) for the group, whereas an FANN query finds the optimal subgroup of query points and the data object that minimizes the aggregate distance for a subgroup of size $m$, and a consensus query finds optimal subgroups and the data objects for each of the subgroup sizes in the range [$n^\prime$, $n$] for $n^\prime<n$. In all these studies, the aggregate similarity is computed based on only spatial (or Euclidean) distances between a data point and a group of query points. In this paper, we address variants of the above queries in the context of the \emph{spatial textual} domain, where both spatial proximity and keyword similarity for a \textit{group or subgroups of users} to data points need to be considered. We call this class of query the \emph{flexible group spatial keyword query}.

%
%
%
%
%
%

\begin{figure}[!ht]
\centering
\includegraphics[width=1.0\columnwidth]{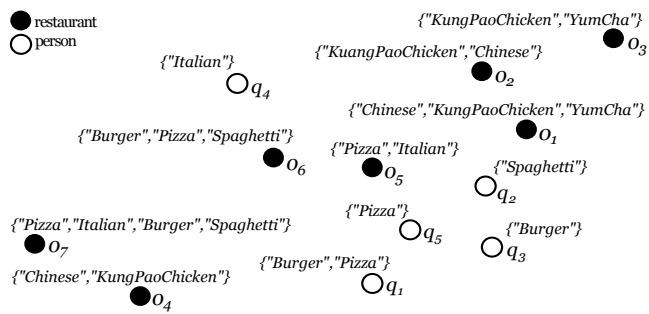}
\caption{A set of user locations $\{q_1,q_2,q_3,q_4,q_5\}$ and a set of restaurants $\{o_1,o_2,...,o_7\}$. Restaurant $o_7$ suits the whole group the best. If size-4 subgroups are considered, then $\{q_1,q_2,q_3,q_5\}$ is optimal with $o_6$ being the best restaurant.}
\label{fig:example}
\end{figure}

The flexible spatial keyword query has many applications in the spatial and multimedia database domains. For example, in a \emph{location-based social networks} (e.g., Foursquare), a group of users residing at their homes or offices can share their locations as spatial coordinates and their preferences as sets of keywords to find a Point of Interest (POI), e.g., restaurant or function venue, that optimizes a cost function composed of aggregate spatial distances and keyword similarities for the group. Since finding a POI that suits all group members might be difficult due to the diverse nature of choices, the group might prefer a result that is not optimal for the entire group, but is optimal for a subset of it. In such cases, we need to find optimal a \textit{subgroup} of users and a POI that minimizes the cost function for the subgroup.

Figure~\ref{fig:example} illustrates the query, where a group of five friends $\{q_1,q_2,q_3,q_4,q_5\}$ is trying to decide on a restaurant for a Sunday brunch. Each person has a location and a preferred type of food, represented by a set of keywords such as \{``Burger'', ``Pizza''\} or \{``Italian''\}, etc. There is a set of restaurants $\{o_1,o_2,...,o_7\}$ to be selected from. Each restaurant also has a location and specializes in a certain type of cuisine which is represented by a set of keywords, e.g., \{``Pizza'', ``Italian''\}. Assume that a cost function $f()$ is used, which considers distance only and aggregates the total travel distance of all the query users in the group to a selected data object.  As can be seen in Figure 1, $o_5$ is the data object closest to the group of query users overall and should be returned by the query. On a different occasion, the group of friends would like to maximize the number of keywords in common between the group query and the POI returned by the query. If we modify the cost function now to stand for the dissimilarity between the respective keyword sets, to be denoted as $g()$, then it turns out that $o_7$ is the one that minimizes this function because it fully covers the keywords of the query users. Both \textit{f} and \textit{g} are extreme cases. In general, it is preferred to find an answer that optimizes both spatial distance and keyword set dissimilarity at the same time, which is the problem studied in this paper. Under such case, neither $o_5$ nor $o_7$ is a good query answer, as they are either not satisfying the query keywords or too far away. However, if we allow leaving out a user, say $q_4$, then more answer candidates become available. In particular, $o_6$ will become the best choice of the subgroup $\{q_1, q_2, q_3, q_5\}$, as it covers all the keywords, and is closer to the group. In fact, leaving any other query user out (e.g., $q_2$) would not obtain a better cost function value. Therefore, $\{q_1, q_2, q_3, q_5\}$ is the optimal subgroup of size 4 and $o_6$ is the corresponding optimal data point.

We observe that in many practical applications relaxing the requirement, i.e., not including all the query objects, has potential benefits in finding good quality answer. Consider a company that wants to find a suitable hotel where to hold the annual shareholder meeting. Each shareholder is identified by his location and a set of keywords describing the type of environment he would like the hotel to be located, like ``metropolitan area'', ``resort'' , ``high altitude'', ``low altitude'', etc. If the cost function to be optimized is an aggregate of the maximum distance traveled and text similarity the hotel selected maybe too far some of the shareholders. On the other hand, by omitting some travelers, the company could accommodate the rest with a shorter travel time. Similarly, in a ride-sharing service, the scheduler may want to find a car for multiple ride-sharers with certain service constraints formulated as keywords. As a third example, in a multimedia domain, one may want to find an image that matches with a subgroup of query images, where an object or query image is represented as a point (in a high-dimensional space) and a set of tag-words. Generally, one may prefer the subgroup size to be maximized, and hence, it benefits to explore the optimal solutions for different subgroup sizes.

The key challenge in processing the flexible group spatial keyword queries is how to utilize both the spatial and keyword preferences  and to efficiently prune the search space. Another major challenge is how to find the optimal subgroups of various sizes in one pass over the data set. We design pruning methods based on branch and bound algorithms to process the queries. We further optimize the algorithms with the best-first search paradigm to minimize the number of data objects visited. Our contributions are as follows:


\begin{itemize}[noitemsep,topsep=0pt]

\item We propose a new class of group queries in the spatial textual domain: (i) the group nearest neighbor with keywords (GNNK) query that finds the best data with respect to our cost function for the whole group, (ii)  the flexible subgroup nearest neighbor with keywords (FSNNK) that finds the optimal subgroup and the corresponding best POI for a given subgroup size of size $m$ (with $m \leq n$, the group size)and (iii) the multiple flexible subgroup nearest neighbor with keywords (MFSNNK) that returns in one pass the optimal subgroups and corresponding POIs for all subgroups of size $m$, where $n^\prime \leq m \leq n$ and $n^\prime$ being the minimum subgroup size. 

\item We propose pruning strategies based on branch and bound as well as best-first strategies for these three queries. The resultant algorithms can process the queries in a single pass over the dataset. 

\item We provide theoretical bounds for our algorithms, and evaluate them through an extensive experimental evaluation on real datasets. The results demonstrate the effectiveness and efficiency of the proposed algorithms. 

\end{itemize}

The rest of the paper is organized as follows. Section~\ref{sec:problemStatement} formulates the queries studied. Section~\ref{sec:related} reviews related work. Section~\ref{sec:processing} describes the proposed algorithms. Sections~\ref{sec:complexitySummary} gives the cost analysis of algorithms. Section~\ref{sec:experiment} reports the experimental results. Section~\ref{sec:conclusion} concludes the paper with a discussion of future work.

\section{Problem Statement}\label{sec:problemStatement}

Let $D$ be a geo-textual dataset. Each object $o \in D$ is defined as a pair $(o.\lambda, o.\psi)$, where $o.\lambda$ is a location point and $o.\psi$ is a set of keywords. A query object $q$ is similarly defined as a pair $(q.\lambda, q.\psi)$. Let $dist(q.\lambda, o.\lambda)$ be the spatial distance between $q$ and $o$, and $similarity\_key(q.\psi, o.\psi)$ be the similarity between their keyword sets. We normalize both $dist(q.\lambda, o.\lambda)$ and $similarity\_key(q.\psi, o.\psi)$ so that their value lie between $0$ and $1$ (inclusive). The cost of $o$ with respect to $q$ is expressed in terms of their spatial distance and keyword set distance:
\begin{equation*}
\begin{split}
cost(q, o) = &\ \alpha \cdot dist(q.\lambda, o.\lambda) \\
			+ &\ (1 - \alpha) \cdot (1 - similarity\_key(q.\psi, o.\psi))
\end{split}
\end{equation*} 

Here, $\alpha$ is a user-defined parameter to control the preference of spatial proximity over keyword set similarity. Using $dist\_key(q.\psi, o.\psi) = 1 - similarity\_key(q.\psi, o.\psi)$, the cost function can be rewritten as:
\begin{equation*}
\begin{split}
cost(q, o) = \ \alpha \cdot dist(q.\lambda, o.\lambda) + (1 - \alpha) \cdot dist\_key(q.\psi, o.\psi)
\end{split}
\end{equation*} 

We formulate the \textit{GNNK}, \textit{FSNNK} and \textit{MFSNNK} queries based on $cost(q, o)$ as follows.

\begin{definition}
$(GNNK)$. Given a set $D$ of spatio-textual objects, a set $Q$ of query objects
$\{q_1 , q_2, ..., q_n\}$, and an aggregate function $f$, the GNNK query finds an object
$o_i \in D$ such that for any $o^\prime \in D \setminus \{o_i\},$ 
$$ f(cost(q_j , o_i) : q_j \in Q) \leq f(cost(q_j , o^\prime) : q_j \in Q) $$
\end{definition}

\begin{definition}
$(FSNNK)$. Given a set $D$ of spatio-textual objects, a set $Q$ of query objects $\{q_1, q_2, ..., q_n\}$, an aggregate function $f$, a subgroup size $m$ $(m \leq n)$, and the set $SG_m$ of all possible subgroups of size $m$, the FSNNK query finds a subgroup $sg_m \in SG_m$ and an object $o_i \in D$ such that for any $o^\prime \in D \setminus \{o_i\}$,
$$f(cost(q_j, o_i) : q_j \in sg_m) \leq f(cost(q_j, o^\prime) : q_j \in sg_m)$$
and for any subgroup $sg_{m}^{\prime} \in SG_m \setminus \{sg_m\}$,
$$f(cost(q_j, o_i) : q_j \in sg_m) \leq f(cost(q^\prime, o^\prime) : q^\prime \in sg_{m}^{\prime})$$
\end{definition}

\begin{definition}
$(MFSNNK)$. Given a set $D$ of spatio-textual objects, a set $Q$ of query objects $\{q_1, q_2, ..., q_n\}$, an aggregate function $f$, and minimum subgroup size $n^\prime$ $(n^\prime \leq n)$, the MFSNNK query returns a set $S$ of $(n-n^\prime+1)$ $\langle subgroup, data\ object \rangle$ pairs such that, each pair $\langle sg_m,o_m \rangle$ is the result of the $FSNNK$ query with subgroup size $m$ ($n^\prime \leq m \leq n$). 
\end{definition}

If the users are interested in the $k$-best POIs then the queries can be generalized as \textit{$k$-GNNK}, \textit{$k$-FSNNK} and \textit{$k$-MFSNNK} queries. These queries are straightforward extensions and the definitions are omitted. In this paper, we focus providing efficient solutions for the above queries for aggregate functions SUM ($\textstyle \sum\nolimits_{q_j \in Q} cost(q_j, o)$) and MAX ($\textstyle \max\nolimits_{q_j \in Q} cost(q_j, o)$). Without loss of generality our solutions work for any aggregate function that is monotonic (e.g., MIN). In our context, a monotonic function means, if we add more elements to the query set $Q$, the aggregate cost will either increase or remain the same.

\section{Related Work} \label{sec:related}

\textbf{Nearest Neighbor Queries.} Nearest neighbor (NN) queries have been well studied in the spatial database community~\cite{hjaltason1995ranking, berchtold1997cost}. The generalization of the nearest neighbor query is known as the $k$NN query. 
The depth-first (DF)~\cite{roussopoulos1995nearest} and the best-first (BF)~\cite{hjaltason1999distance} algorithms are commonly used to process the $k$NN queries. 
They assume the data objects to be indexed in a tree structure, e.g., the R-tree~\cite{guttman1984r}. 
In the DF algorithm, child nodes are recursively visited according to their $min\_dist$ from the query point. Here the $min\_dist$ of a node is defined as the minimum Euclidean distance between its minimum bounding rectangle (MBR) and the query point. It gives a lower bound over the distances of the child nodes, and hence the algorithm can safely prune the nodes with $min\_dist$ greater than the distance of the nearest neighbor already retrieved.
The BF algorithm maintains a priority queue of nodes to be visited. The nodes in the queue are ordered based on the $min\_dist$. Initially the children of the root node are inserted into the priority queue. At each step, the node with the lowest $min\_dist$ is popped from the queue and its children are inserted. The algorithm returns the first $k$ data objects popped from the queue as the $k$NN query answer. 


\textbf{Group Nearest Neighbor Queries.} The group nearest neighbor (GNN) query~\cite{papadias2004group} finds a data point that minimizes the aggregate distance for a group of query locations. SUM, MAX and MIN are commonly used aggregate functions. The generalization of the GNN query is the $k$GNN query, where $k$ best group nearest neighbors are to be found. Several methods for processing GNN queries have been presented in~\cite{papadias2005aggregate}. Among those, the MBM algorithm is the state of the art. It visits the R-tree nodes in the order of their aggregate distance from the set of query points. The distance of the best data object retrieved so far is used as the pruning bound while visiting the nodes.

The flexible aggregate nearest neighbor (FANN) query~\cite{li2011flexible} is a generalization of the GNN query. It returns the data object that minimizes the aggregate distance to any subset of $\phi n$ query points, where $n$ is the size of the query group and $0 < \phi \leq 1$. The query also returns the corresponding subset of query points. Two exact algorithms to process the FANN query have been proposed in~\cite{li2011flexible}. The first uses a branch and bound method to restrict the search space, assuming that the data 
objects are indexed in an R-tree. The second uses the \textit{threshold algorithm}~\cite{fagin2003optimal} to find the answer. 

A query similar to the FANN query called the \textit{consensus query}~\cite{ali2015consensus} is the main motivation of our paper. 
Given a minimum subgroup size $m$ and a set of $n$ query points, the consensus query finds objects that minimize the aggregate distance for all subgroups with sizes in the range $[m, n]$. A BF algorithm was proposed to process the consensus query. 

The above group queries~\cite{papadias2005aggregate, li2011flexible, ali2015consensus} only consider spatial proximity while a selecting data object, whereas, we consider both spatial proximity and textual similarity.

\textbf{Spatial Keyword Queries.} The spatial keyword query consists of a query location and a set of query keywords. A spatio-textual data object is returned based on its spatial proximity to the query location and textual similarity with the query keywords. A number of indexing structures for processing the spatial keyword query have been proposed~\cite{cong2009efficient, li2011ir, zhang2009keyword, cao2010retrieving, rocha2011efficient, ZhengBST15}. Among them, the IR-tree~\cite{cong2009efficient, li2011ir} has been shown to be a highly efficient one. The IR-tree augments each node of the R-tree with an inverted file corresponding to the keyword sets of the child nodes. 
The WAND method~\cite{broder2003efficient,ding2011faster, chakrabarti2011interval} is proposed for document queries. 
This method is mainly designed for document retrieval and uses TF-IDF measures for document ranking. In our study, we consider both spatial and textual similarity, and 
 use the IR-tree to index the data objects, although other spatial keyword indexes may be used as well. The WAND method in particular can be applied 
 in the leaf level of the IR-tree to help compute the textual similarity.

A variant of the spatial keyword query, called \textit{spatial group keyword query} has been introduced \cite{cao2011collective, cao2015efficient}. It finds a group of objects that cover the keywords of a \emph{single query} such that both the aggregate distance of the objects from the query location and the inter-object distances within the group are also minimized. Exact and approximate algorithms for three types of aggregate functions (SUM, MAX and MIN) have been presented in \cite{cao2015efficient}. \cite{chen2012aggregate} studies the aggregate keyword routing (AKR) query (AKR). For a given set of users, an AKR query finds a route through a set of objects $K$ that covers all users' keywords and minimizes the maximum distance travelled by any user to a meeting point $p$ through $K$.

In a study parallel to ours, the group top-$k$ spatial keyword query has been proposed recently~\cite{yao2016efficient}. This paper presents a branch-and-bound technique to retrieve the top-$k$ spatial keyword objects for only one group of queries. This technique is essentially our branch-and-bound method described in Section~\ref{subsection:baseline} for the GNNK queries. As we show in our experimental evaluation (Section~\ref{sec:experiment}), our best-first technique always outperforms the branch-and-bound method substantially even for a single group query.

None of the existing work in the geo-textual domain addresses the problem of finding optimal subgroups and data objects in terms of spatial proximity and textual similarity, which is our main focus in this paper.

\section{Our Approach} \label{sec:processing}
This section presents our algorithms to process the GNNK, FSNNK and MFSNNK queries. The key challenge is to utilize the spatial distance and keyword preference together to 
constrain the search space as much as possible, since the performance of the algorithms is directly proportional to the search space (in both running time and I/O). Another challenge in the FSNNK and MFSNNK queries is to find the optimal subgroup from all possible subgroups.


\subsection{Preliminaries} \label{subsection:ir-tree}
We use the IR-tree \cite{cong2009efficient} to index our geo-textual dataset $D$. Other extensions of the IR-tree, such as the CIR-tree, the DIR-tree or the CDIR-tree~\cite{cong2009efficient} can be used as well. 

The IR-tree is essentially an inverted file augmented R-tree~\cite{guttman1984r}. The leaf nodes of the IR-tree contain references to the objects from dataset $D$. Each leaf node has also a pointer to an inverted file index corresponding to the keyword sets of the objects stored in that node. The inverted file index stores a mapping from the keywords to the objects where the keywords appear. Each node $N$ of the IR-tree has the form ($N.\Lambda$, $N.\Psi$), where $N.\Lambda$ is the minimum bounding rectangle (MBR) that bounds the child node entries, and $N.\Psi$ is the union of the keyword sets in the child node entries. 

\begin{example} \label{example:ir-tree}
Figure~\ref{fig:mbr} shows the locations of seven spatial objects $o_1, o_2, ..., o_7$. Figure~\ref{table:object-keywords} shows their keyword sets. The corresponding IR-tree and inverted files are not shown for space limitation. \qed

\end{example}

\begin{figure}[h!]
    \centering
	\setlength{\belowcaptionskip}{-1pt}
    \begin{subfigure}[b]{0.5\linewidth}
        \includegraphics[width=\linewidth]{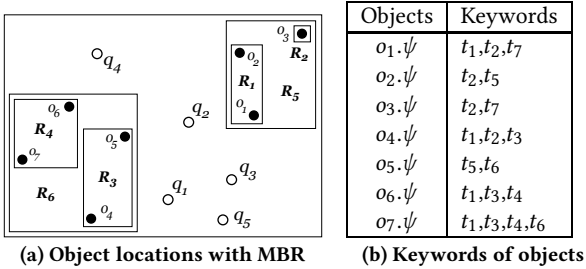}
	    \caption{Object locations with MBR}
        \label{fig:mbr}
	\end{subfigure}\quad
    \begin{subfigure}[b]{0.4\linewidth}
        \begin{tabular}{|c|l|}
    		\hline
            Objects & Keywords             \\ \hline
            $o_1.\psi$   & $t_1$,$t_2$,$t_7$ \\
            $o_2.\psi$   & $t_2$,$t_5$       \\
            $o_3.\psi$   & $t_2$,$t_7$		 \\
            $o_4.\psi$   & $t_1$,$t_2$,$t_3$ \\
            $o_5.\psi$   & $t_5$,$t_6$		 \\
            $o_6.\psi$   & $t_1$,$t_3$,$t_4$ \\
            $o_7.\psi$   & $t_1$,$t_3$,$t_4$,$t_6$ \\ \hline
		\end{tabular}
		\caption{Keywords of objects}
        \label{table:object-keywords}
	\end{subfigure}
    \caption{Locations and keywords of objects and queries}
    \label{fig:mbr-keyword}
\end{figure}

\subsection{Cost Function} \label{subsection:cost-model}
This subsection elaborates the cost function to be optimized. As defined in Section~\ref{sec:problemStatement}, the cost of an object is a combination of spatial distance and keyword dissimilarity:
\begin{equation*}
\begin{split}
cost(q, o) = &\ \alpha \cdot dist(q.\lambda, o.\lambda) \\
			+ &\ (1 - \alpha) \cdot (1 - similarity\_key(q.\psi, o.\psi))
\end{split}
\end{equation*} 

We use the Euclidean distance as the spatial distance metric. The spatial distance is normalized by the maximum spatial distance between any pair of objects in the dataset, $d_{max}$. Thus, $$dist(q.\lambda, o.\lambda) = euclidean\_distance(q.\lambda, o.\lambda) / d_{max}$$

Each keyword in the dataset is associated with a weight. Following a previous study on spatial keyword search~\cite{cong2009efficient}, we use the Language Model~\cite{ponte1998language} to generate the keyword weights. The weight of each keyword is normalized by the maximum keyword weight $w_{max}$ present in the dataset. 
Let $y.w$ be the weight of keyword $y$. Then the text relevance between $q$ and $o$ is the normalized sum of the weights of the keywords shared by $q$ and $o$:
$$similarity\_key(q.\psi, o.\psi) = \frac{1}{|q.\psi|} \sum\limits_{y \in q.\psi \cap o.\psi} \frac{y.w}{w_{max}}$$

Various alternative measures for textual data have been proposed, such as cosine similarity~\cite{rocha2011efficient}, the Extended Jaccard~\cite{lu2014efficient}, etc, but extensive experiments~\cite{lu2014efficient} have shown that not one similarity measure outperforms the others in all cases.

\begin{example}
We continue with the example shown in Figure~\ref{fig:mbr-keyword}. Let the keywords of the query points be: $q_1.\psi=\{t_1,t_2\}$, $q_2.\psi=\{t_4\}$, $q_3.\psi=\{t_3,t_6\}$, $q_4.\psi=\{t_1\}$, and $q_5.\psi=\{t_4,t_6\}$. Let us assume $\alpha=0.5$, the weight of any keyword $w = 1$ ($w_{max}=1$), and $f= $SUM. 

We show the aggregate cost computation for $o_6$. 
Let the distances from $q_1, q_2, q_3, q_4$, and $q_5$ to $o_6$ be 3.5, 5.5, 6.5, 1, and 9.5 units, and $d_{max}$ be 10 units. Then $dist(q_3.\lambda,o_6.\lambda)=\frac{6.5}{10}=0.65$. Meanwhile, $q_3.\psi \cap o_6.\psi=\{t_3\}$. Thus,\\ $similarity\_key(q_3.\psi,o_6.\psi)=\frac{t3.w}{|q_3.\psi|}=0.5$, and overall, 
\begin{equation*}
\begin{split}
cost(q_3, o_6) = &\ \alpha \cdot dist(q_3.\lambda, o_6.\lambda) \\
			+ &\ (1 - \alpha) \cdot (1 - similarity\_key(q_3.\psi, o_6.\psi))\\
			= &\ 0.5*0.65 + (1-0.5)*(1-0.5) = 0.575\\
\end{split}
\end{equation*}
Similarly, we compute the costs for $q_1,q_2,q_4$, and $q_5$, which are 0.175, 0.535, 0.05, and 0.725, respectively. Thus, the aggregate cost is $f(cost(Q, o_6)) =   \sum\limits_{q_j \in Q} cost(q_j, o_6) = 2.05.$
\qed
\label{example:cost-model} 
\end{example}

The cost of an IR-tree node is defined similarly to the cost of a data object:
\begin{equation*} \label{eq:cost_query_to_node}
\begin{split}
cost(q, N) &= \alpha\ min\_dist(q.\lambda, N.\Lambda) \\
	   & + (1 - \alpha)\ (1 - similarity\_key(q.\psi, N.\Psi))
\end{split}
\end{equation*}

Here, $min\_dist(q.\lambda, N.\Lambda)$ is the minimum spatial distance between the query location and the MBR of $N$; $similarity\_key(q.\psi, N.\Psi)$ is the textual similarity between the query keywords and the keywords of the node. The cost of an IR-tree node gives a lower bound over the cost of its children, as formalized by the following lemma:

\begin{lemma} \label{lemma:cost}
Let $N$ be an IR-tree node and $q$ be a query object. If $N_c$ is a child of $N$, then $cost(q, N) \leq cost(q, N_c)$. 
\end{lemma}

\begin{proof}
The child $N_c$ can either be a data object or an IR-tree node. In either case $min\_dist(q.\lambda, N.\Lambda)$ is smaller than or equal to that of $N_c$ according 
to the R-tree structure. Meanwhile, the keyword set of $N_c$ is a subset of the keyword set of $N$. Thus, $N$ will have a higher (or equal) textual similarity value (and hence lower keyword set distance) with the query keywords. 
Overall, we have $cost(q, N) \leq cost(q, N_c)$.
\end{proof}


\subsection{Branch and Bound Algorithms for GNNK and FSNNK} \label{subsection:baseline}
Traditional nearest neighbor algorithms access the data indexed in a spatial index (e.g., R-tree) and restricts its search space by pruning bounds~\cite{roussopoulos1995nearest}. We extend this idea to design two branch and bound algorithms for the GNNK and FSNNK queries. These two algorithms will work as the baseline algorithms 
in the experiments. 

\begin{algorithm} [ht!]
	\caption{GNNK-BB ($R, Q, f$)}
	\label{algo:gnnk-baseline}
	\begin{algorithmic}[1]
	\INPUT IR-tree index $R$ of all data objects, $n$ query points $Q = \{q_1, q_2, ..., q_n\}$, monotonic cost function $f$.
	\OUTPUT A data object $o$ that minimizes the aggregate cost with respect to the query set $Q$
	\State $min\_cost \gets \infty$
	\State $stack \gets \emptyset $
	\State $stack.push(root)$
	\Repeat
		\State $N \gets stack.pop()$
		\If {$N$ is an intermediate node}
			\ForAll {$N_c$ in $N.children$}
				\If {$f(cost(Q, N_c)) < min\_cost$}
					\State $stack.push(N_c)$
				\EndIf
			\EndFor
		\ElsIf{$N$ is a leaf node}
			\ForAll {$o$ in $N.children$}
				\If {$f(cost(Q, o)) < min\_cost$}
					\State $min\_cost \gets f(cost(Q, o))$
					\State $best\_object \gets o$
				\EndIf
			\EndFor
		\EndIf
	\Until{$stack$ is empty}
	\State return $best\_object$
	\end{algorithmic}
\end{algorithm}

\textbf{Branch and Bound Algorithm for \textit{GNNK}.}
We use the following heuristic to prune the unnecessary nodes while searching the IR-tree for the best object with the minimum aggregate cost.
\begin{heuristic}\label{heuristic:gnnk}
A node $N$ can be safely pruned if its aggregate cost with respect to the query set $Q$ is greater than or equal to the smallest cost of any object retrieved so far.
\end{heuristic}

This heuristic is derived from Lemma~\ref{lemma:cost}. As $f$ is a monotonic function and $cost(q, N) \leq cost(q, N_c)$ for any child $N_c$ of $N$, $f(cost(Q, N))$ will be less than or equal to $f(cost(Q, N_c))$. Let $min\_cost$ be the smallest cost of any data object retrieved so far. Then $f(cost(Q, N)) \geq min\_cost$ implies that the cost of any descendant of $N$ is greater than or equal to $min\_cost$, and we can safely prune $N$.

Algorithm~\ref{algo:gnnk-baseline} shows the pseudo-code of the branch and bound algorithm based on the heuristic, denoted by GNNK-BB. 
The algorithm maintains a stack of nodes/objects to be visited. 
The lowest cost object visited so far as well as the lowest cost are maintained in the variables $best\_object$ and $min\_cost$, respectively.
The algorithm starts with inserting the root node of the IR-tree into the stack (Line 3).  At each step, it gets the next node/object from the stack (Line 5) and computes the  aggregate query cost for each of the child nodes (if any) (Lines 7-8 and 11-12). If the child is a data object and its cost is lower than $min\_cost$, then we update $min\_cost$ with the aggregate cost of that child (Lines 12-14). Otherwise the child is an IR-tree node and if its cost is lower than $min\_cost$, we insert it into the stack so that we can visit its children later (Lines 8-9). At the end when the stack becomes empty, the algorithm returns the object corresponding to $min\_cost$ as the result (Line 16).

\begin{table}[ht!]
\centering
\caption{Example of the GNNK-BB algorithm}
\begin{adjustbox}{max width=1.0\columnwidth}
\begin{tabular}{|l|l|l|l|l|l|l|}
\hline
\textbf{Step} & \textbf{S} & \textbf{Elm} & \textbf{$f(cost)$} & \textbf{$best\_obj$}& \textbf{$min\_cost$} & \textbf{S} (updated)\\
\hline
\multirow{2}{*}{1} & \multirow{2}{*}{$root$} & \multirow{2}{*}{$root$} & $R_5$ : 2.475 & $\emptyset$ & $\infty$ & \multirow{2}{*}{$R_5,R_6$} \\
& & & $R_6 : 0.725$ & $\emptyset$ & $\infty$ &\\ \hline
\multirow{2}{*}{2} & \multirow{2}{*}{$R_6,R_5$} & \multirow{2}{*}{$R_6$} & $R_3$ : 1.75 & $\emptyset$ & $\infty$ & \multirow{2}{*}{$R_5,R_3,R_4$} \\
& & & $R_4 : 1.1$ & $\emptyset$ & $\infty$ &\\ \hline
\multirow{2}{*}{3} & \multirow{2}{*}{$R_5,R_3,R_4$} & \multirow{2}{*}{$R_4$} & $o_6$ : 2.05 & $o_6$ & 2.05 & \multirow{2}{*}{$R_5,R_3$} \\
& & & $o_7 : 1.625$ & $o_7$ & $1.625$ &\\ \hline
\multirow{2}{*}{4} & \multirow{2}{*}{$R_5,R_3$} & \multirow{2}{*}{$R_3$} & $o_4$ : 2.75 & $o_7$ & 1.625 & \multirow{2}{*}{$R_5$} \\
& & & $o_5 : 3.0$ & $o_7$ & $1.625$ &\\ \hline

5 & $R_5$ & $R_5$ & \multicolumn{4}{l|}{$f(cost(Q,R_5))>min\_cost\Rightarrow$ prune $R_5$; $S={\emptyset}$, $return$ $o_7$} \\ \hline



\end{tabular}
\end{adjustbox}
\label{table:gnnkbb-calc}
\end{table}

\begin{example}(\textbf{GNNK-BB}).
We continue with Example~\ref{example:cost-model}. Table~\ref{table:gnnkbb-calc} summarizes the MFSNNK-BF steps using aggregate function SUM.
Column $S$ shows the stack; column $Elm$ shows the element popped out; $f(cost)$ shows the aggregate costs of the child nodes; column $best\_obj$ and $min\_cost$ shows current best object and minimum cost, respectively; column $S\ (updated)$ shows the updated stack after processing the popped element. 

At start, the tree root is popped out. The aggregate cost for each children $R_5$ and $R_6$ is less than the initialized cost $\infty$ and so, they are pushed into the stack. In step 3, leaf node $R_4$ is popped. So $best\_obj$ and $min\_cost$ are updated. In Step 4, the cost of each object $o_4$ and $o_5$ is greater than $min\_cost$, so no update occurs. In step 5 cost of $R_5$ is greater than $min\_cost$. So, $R_5$ is pruned, and stack $S$ becomes empty. Then algorithm terminates, and $o_7$ is returned, which is the current best object. \qed 

\label{example:gnnk}
\end{example}

\textbf{Branch and Bound Algorithm for \textit{FSNNK}.}
We design a similar branch and bound algorithm named FSNNK-BB for the FSNNK query. The following heuristic is used for pruning.
\begin{heuristic} \label{heuristic:fsnnk}
Let $N$ be an IR-tree node and $m$ be the required subgroup size. If $sg_m$ is the best subgroup of size $m$, and $min\_cost$ is the smallest cost of any size-$m$ subgroup retrieved so far, we can safely prune $N$ if $f(cost(sg_m, N)) \geq min\_cost$.
\end{heuristic}

This heuristic is derived from Lemma~\ref{lemma:cost}. Let $N_c$ be a child of $N$ and $sg'_m$ be the best subgroup corresponding to $N_c$. 
Then we have $$ f(cost(sg'_m, N)) \leq f(cost(sg'_m, N_c)) $$
Meanwhile $sg_m$ is the best subgroup for $N$ among all possible subgroups of size $m$. Thus,  
$$ f(cost(sg_m, N)) \leq f(cost(sg'_m, N)) $$
The two inequalities imply that $ f(cost(sg_m, N)) \leq f(cost(sg'_m, N_c)) $, i.e., the aggregate cost for the best size-$m$ subgroup of $N$ is lower than or 
equal to that of the best size-$m$ subgroup of any of its children. Therefore, if $f(cost(sg_m, N)) \geq min\_cost$, $f(cost(sg_m, N_c))$ will also be greater than or equal to $min\_cost$, and we should prune $N$. 

The overall tree traversal procedure is similar to that of the GNNK-BB algorithm. The difference is in the calculation of the optimization function, where the optimization function value is computed based on the the top-$m$ queries with the lowest costs. For an intermediate node $N$, we compute the best subgroup and the aggregate cost (bound) in a similar way for all of its child nodes. First, the costs from all the query points to a node are calculated. Then $m$ query points with lowest costs are taken to get the best subgroup $sg_m$. If the aggregate cost for $sg_m$ is lower than $min\_cost$, then we insert the child node into the stack. Otherwise, it is pruned. We omit the details due to space constraints.


\subsection{Best-first Algorithms for GNNK and FSNNK} \label{subsection:pq}

Branch and bound techniques may access unnecessary nodes during query processing. To improve the query efficiency by reducing disk accesses, we propose in this section best-first search techniques that only access the necessary nodes.

\textbf{Best-first algorithm for GNNK.}
The best-first procedure for the GNNK query, denoted by GNNK-BF, is shown in Algorithm~\ref{algorithm:gnnk-pq}. This algorithm uses a minimum priority queue $P$ to maintain the nodes/objects to be visited according to their aggregate costs. At start, the queue $P$ is initialized with the root of the IR-tree (Lines 1-2). At each iteration of the main loop (Lines 3-13), the element with the minimum aggregate cost is popped out from $P$. There are three cases to be considered for a popped element: \textit{(i)} If it is an intermediate node, then all child nodes are pushed into $P$ according to their aggregate costs (Lines 5-7). \textit{(ii)} If it is a leaf node, then all child objects are pushed into $P$ according to their aggregate costs (Lines 8-10). \textit{(iii)} If it is an object, then it is returned as the query result (Lines 11-12), and the algorithm terminates (Line 14). 

\begin{algorithm}[t!]
	\caption{GNNK-BF ($R, Q, f$)}
	\begin{algorithmic}[1]
	\INPUT IR-tree index $R$ of all data objects, $n$ query points $Q = \{q_1, q_2, ..., q_n\}$, monotonic cost function $f$.
	\OUTPUT A data object $o$ that minimizes the aggregate cost with respect to the query set $Q$
	\State Initialize a new min priority queue $P$
	\State $ P.push(root, 0) $
    \Repeat
        \State $E \gets P.pop()$
		\If {$E$ is an intermediate node $N$}
			\ForAll {$N_c$ in $N.children$}
				\State $ P.push(N_c, f(cost(Q, N_c))) $
			\EndFor
		\ElsIf{$E$ is a leaf node $N$}
			\ForAll {$o$ in $N.children$}
				\State $ P.push(o, f(cost(Q, o))) $
			\EndFor
		\ElsIf{$E$ is a data object $o$}
			\State return $o$
		\EndIf
	\Until{$P$ is empty}
	\State return $null$
	\end{algorithmic}
	\label{algorithm:gnnk-pq}
\end{algorithm}

%
\begin{example}(\textbf{GNNK-BF}).
We continue with Example~\ref{example:cost-model}. The algorithm steps are summarized in Table~\ref{table:gnnk-calc}, where SUM is used as the aggregate function. 
Column $P$ shows the current elements in the queue; column $Element$ shows the element popped out in the current step; column $f(cost)$ shows the aggregate costs of the child nodes of the popped element; column $P\ (updated)$ shows the updated queue after processing the popped element. 

\begin{table}[ht!]
\centering
\caption{Example of the GNNK-BF algorithm}
\begin{adjustbox}{max width=.98\columnwidth}
\begin{tabular}{|l|l|l|l|l|}
\hline
\textbf{Step} & \textbf{P} & \textbf{Element} & \textbf{$f(cost)$} & \textbf{P} (updated)\\
\hline
\multirow{2}{*}{1} & \multirow{2}{*}{$root$} & \multirow{2}{*}{$root$} & $R_5$ : 2.475 & \multirow{2}{*}{$R_6,R_5$} \\
& & & $R_6 : 0.725$ & \\ \hline
\multirow{2}{*}{2} & \multirow{2}{*}{$R_6,R_5$} & \multirow{2}{*}{$R_6$} & $R_3$ : 1.75 & \multirow{2}{*}{$R_4,R_3,R_5$} \\
& & & $R_4 : 1.1$ & \\ \hline
\multirow{2}{*}{3} & \multirow{2}{*}{$R_4,R_3,R_5$} & \multirow{2}{*}{$R_4$} & $o_6$ : 2.05 & \multirow{2}{*}{$o_7,R_3,o_6,R_5$} \\
& & & $o_7 : 1.625$ & \\ \hline
4 & $o_7,R_3,o_6,R_5$ & $o_7$ & \multicolumn{2}{l|}{$return \quad o_7$} \\ \hline

\end{tabular}
\end{adjustbox}
\label{table:gnnk-calc}
\end{table}

At start, the tree root is popped out. The aggregate costs for the children $R_5$ and $R_6$ are computed and they are pushed into the queue. The node $R_6$ has the lowest aggregate cost, and hence it is at the front of the queue. In the next step, $R_6$ is popped out and the aggregate costs for its children $R_3$ and $R_4$ are computed. This procedure repeats until Step 4 where $o_7$ is popped out. This is the first data object popped out. According to the algorithm, this object is the best object for the query, and hence it is returned as the query answer. \qed 

\label{example:gnnk}
\end{example}

\begin{lemma} \label{lemma:gnnk}
\textbf{\textit{(Proof of Correctness)}} GNNK-BF returns the object with the minimum aggregate cost w.r.t. the query set $Q$.
\end{lemma}

\begin{proof}
Let $o$ be the data object returned by GNNK-BF, i.e, $o$ is the first data object visited by the algorithm. Assume that a different object $o\prime$ is the data object with 
the minimum aggregate cost. 
Then $f(cost(Q, o\prime)) \leq f(cost(Q, o))$. Let $N$ be the first common ancestor of $o$ and $o\prime$ in the IR-tree. We know from Lemma~\ref{lemma:cost} that the cost of an IR-tree node gives a lower bound over the costs of its children. Thus, any node in the path from $N$ to the parent of $o\prime$ will have a lower aggregate cost than 
that of $o\prime$. This implies that these nodes have lower aggregate costs than that of $o$, and should be visited before $o$. Therefore, when $o$ is visited, $o\prime$ must be in the priority queue as its parent has already been visited. Because $o\prime$ has a lower cost than $o$ has, it should be visited first, which means that $o\prime$ must be the data object returned by GNNK-BF rather than $o$. This is conflict to our assumption, and hence $o\prime$ should not have been existed. Therefore, $o$ must be the data object with the minimum aggregate cost. 
\end{proof}


\begin{algorithm}[ht!]
    \caption{FSNNK-BF ($R, Q, m, f$) [partial]}
    \begin{algorithmic}[1]
    \State ...
        \If {$E$ is an intermediate node $N$}
    	    \ForAll {$N_c$ in $N.children$}
            	\State Compute $cost(q_1, N_c), ... , cost(q_n, N_c)$
    			\State $sg_m \gets $ first $m$ query points with the lowest costs
    		    \State $ P.push(N_c, f(cost(sg_m, N_c))) $
    	    \EndFor
        \ElsIf{$E$ is a leaf node $N$}
          \State ...
	    \ElsIf{$E$ is a data object $o$}
	  	    \State return $(o, o.best\_subgroup)$
        \EndIf
    \State ...
    \end{algorithmic}
    \label{algorithm:fsnnk-pq}
\end{algorithm}

\textbf{Best-first Algorithm for FSNNK.}
The best-first algorithm for the FSNNK query, denoted by FSNNK-BF, is similar to GNNK-BF algorithm. This algorithm also maintains a minimum priority queue to manage the nodes/objects to be visited from the IR-tree, and traverses the tree from the root. Here, optimization function is computed for top-$m$ queries. Best subgroup is chosen from the lowest $m$ query points, and pushed into the priority queue. For an intermediate node, aggregate costs and best subgroup are calculated for all the child nodes of the node. For a leaf node, it is done for all the children objects, and then pushed into the priority queue. When an object is first popped, it is returned as the result. The partial pseudo-code is shown in Algorithm~\ref{algorithm:fsnnk-pq}.


\begin{example}(\textbf{FSNNK-BF}).
We continue with Example~\ref{example:cost-model} for the FSNNK query. Let the subgroup size $m=3$. The algorithm steps for FSNNK-BF are summarized in Table~\ref{table:fsnnk-calc}. Column $Element$ shows the elements popped out from $P$ at every step; column $sg_m$ shows the best subgroup of size $m$ corresponding to the current node or data object, which is also the set of $m$ lowest cost query points corresponding to the current node or data object; column $f_m(cost)$ is the aggregate cost over the query points in $sg_m$; column $P\ (updated)$ shows the updated queue after processing the popped element.

\begin{table}[ht!]
\renewcommand{\arraystretch}{0.8}
\centering
\caption{Example of the FSNNK-BF algorithm}
\begin{adjustbox}{max width=.98\columnwidth}
\begin{tabular}{|l|l|l|l|l|}
\hline
\textbf{Step} & \textbf{Element} & \textbf{$f_m(cost)$} & \textbf{$sg_m$} & \textbf{P }(updated)\\
\hline
\multirow{2}{*}{1} & \multirow{2}{*}{$root$} & $R_5$ : 1.2 & $q_4,q_1,q_2$ &\multirow{2}{*}{$R_6,R_5$} \\
& & $R_6 : 0.225$ & $q_4,q_1,q_2$ & \\ \hline
\multirow{2}{*}{2} & \multirow{2}{*}{$R_6$} & $R_3$ : 0.45 & $q_1,q_4,q_3$ & \multirow{2}{*}{$R_3,R_4,R_5$} \\
& & $R_4 : 0.475$ & $q_4,q_1,q_2$ & \\ \hline
\multirow{2}{*}{3} & \multirow{2}{*}{$R_3$} & $o_4$ : 1.15 & $q_1,q_4,q_3$ &\multirow{2}{*}{$R_4,o_4,R_5,o_5$} \\
& & $o_5 : 1.7$ & $q_3,q_1,q_5$ & \\ \hline
\multirow{2}{*}{4} & \multirow{2}{*}{$R_4$} & $o_6$ : 0.75 & $q_4,q_1,q_2$ &\multirow{2}{*}{$o_6,o_7,o_4,R_5,o_5$} \\
& & $o_7 : 0.8$ & $q_1,q_4,q_3$ & \\ \hline
5 & $o_6$ & \multicolumn{3}{l|}{$return \quad (o_6,\{q_1,q_4,q_3\})$}\\ \hline
\end{tabular}
\end{adjustbox}
\label{table:fsnnk-calc}
\end{table}

At start, the tree root is popped out. The individual  costs for children $R_5$ and $R_6$ are computed. The best subgroups for $R_5$ and $R_6$ are shown in the $sg_m$ column. The aggregate costs for $R_5$ and $R_6$ are also computed. Both nodes are then pushed into $P$. In the next step, $R_6$ is popped out, as it has the minimum cost. The computation for the children of $R_6$ is carried out in the same way. This procedure repeats, and at Step 5, $o_6$ is popped out. It is the first object popped out, which gives the minimum aggregate cost among all data objects. FSNNK-BF returns $o_6$ and the corresponding best subgroup $\{q_1,q_4,q_3\}$ as the query answer. \qed
\end{example}


\subsection{Algorithms for MFSNNK} \label{subsection:fsnnk-extended}

\begin{algorithm}[ht!]
    \caption{MFSNNK-BF ($R, Q, m, f$)}
    \begin{algorithmic}[1]
    \INPUT IR-tree index $R$ of all data objects, $n$ query points $Q = \{q_1, q_2, ..., q_n\}$, minimum subgroup size $m ( m \leq n)$, monotonic cost function $f$.
    \OUTPUT A set of $\langle data\_object, subgroup \rangle$ pairs $\langle o_k^*, sg_k^* \rangle$ for all subgroup sizes between $m$ and $n$ (inclusive), where $ \langle o_k^*, sg_k^* \rangle$ minimizes $f(cost(sg_k, o))$.
    \State Initialize a new min priority queue $P$
    \State $ min\_costs[i] \gets \infty $ for $ m \leq i \leq n $
    \State $ root.query\_costs[i] \gets 0 $ for $ m \leq i \leq n $
    \State $ P.push(root, 0) $
    \Repeat
        \State $E \gets P.pop()$
        \If {$\exists i \in [m, n]$: $E.query\_costs[i]<min\_costs[i]$}
            \If {$E$ is an intermediate node}
    	        \ForAll {$N_c$ in $E.children$}
    		    \State Compute $cost(q_1, N_c), ... , cost(q_n, N_c)$
    		    \State $total\_cost \gets 0$
    		    \For {$i = m \to n$}
    		    	\State $sg_i \gets$ top $i$ lowest cost query points
    		    	\State $total\_cost \pluseq f(cost(sg_i, N_c)) $
    		    	\State $N_c.query\_costs[i] = f(cost(sg_i, N_c)) $
    		    \EndFor
   		    \If {$f(cost(sg_i, N_c))<min\_costs[i]$ for any\par
        			\hskip\algorithmicindent \hspace{1.2cm} subgroup size $i \in [m, n]$}
    				\State $P.push(N_c, total\_cost)$
    			\EndIf
    		  \EndFor
            \ElsIf{$E$ is a leaf node}
    	  	    \ForAll {$o$ in $N.children$}
    		    \State Compute $cost(q_1, o), ... , cost(q_n, o)$
    		    \For {$i = m \to n$}
    		    	\State $sg_i \gets$ top $i$ lowest cost query points
    		    	\If {$f(cost(sg_i, o)) < min\_costs[i]$}
    		    		\State $min\_costs[i] \gets f(cost(sg_i, o)) $
    		    		\State $best\_objects[i] \gets o$
    		    		\State $best\_subgroups[i] \gets sg_i$
    		    	\EndIf
    		    \EndFor
    		  \EndFor
          \EndIf
        \EndIf
    \Until{$P$ is empty}
    \State return $best\_objects, best\_subgroups$
  \end{algorithmic}
  \label{algorithm:fsnnk-ex}
\end{algorithm}

To process the MFSNNK query with a minimum subgroup size $m$, we can run FSNNK-BF $n-m+1$ times (for subgroup sizes $m, m+1, ..., n$) and return the combined results. We call this the MFSNNK-N algorithm. However, MFSNNK-N requires accessing the dataset $n-m+1$ times, which is too expensive. To avoid this repeated data access, we design an algorithm based on best-first method that can find the best data objects for all subgroup sizes between $m$ and $n$ in a single pass over the dataset. 
The algorithm is based on the following heuristic.

\begin{heuristic} \label{heuristic:mfsnnk}
Let $N$ be an IR-tree node and $m$ be the minimum subgroup size. Let $sg_i$ be the best subgroup of size $i$ ($m \leq i \leq n$), and $min\_cost_i$ be the smallest cost for subgroup size $i$ from any object retrieved so far. We can safely prune $N$ if $f(cost(sg_i, N)) \geq min\_cost_i$ for any $i$.
\end{heuristic}

The proof of correctness is straightforward based on Heuristic~\ref{heuristic:gnnk} and Heuristic~\ref{heuristic:fsnnk}, 
and is omitted due to space limit. 

Algorithm~\ref{algorithm:fsnnk-ex} summarizes the proposed procedure, denoted as MFSNNK-BF. The algorithm maintains a minimum priority queue $P$ 
to manage the nodes/objects to be visited from the IR-tree (Line 1). The minimum costs for all subgroup sizes in the range [$m, n$] are set to $\infty$  at the beginning (Line 2). 
Each tree node to be visited is associated with an array $query\_costs$ that keeps track of the aggregate costs for all subgroup sizes in the range [$m$, $n$] (Line 3). 
The algorithm pushes the tree root into the queue $P$ and then the main loop begins (Lines 6-36). 
At each iteration, an element is popped out from $P$. The associated $query\_costs$ (already computed at a previous iteration) is compared with $min\_costs$. 
If $query\_cost$ for any subgroup size is lower than the $min\_cost$ of that subgroup size, the element needs to be considered further. Otherwise the element is pruned according to Heuristic~\ref{heuristic:mfsnnk} (Lines 7-8). There are two cases to be further considered: \textit{(i)} If the element is an intermediate node, then we compute the costs for each child node (Lines 10-11).  We compute the aggregate cost for each subgroup size in the range [$m$, $n$] (Lines 13-17), and store the corresponding best query subgroup in $sg_i$. If the aggregate cost is larger than $min\_cost$ for all subgroup sizes, the child node can be safely pruned. Otherwise we insert the child node into $P$ according to its total cost. (Lines 18-20) \textit{(ii)} If the element is a leaf node, then a similar computation is performed for each child object (Lines 23-26). If the aggregate cost is less than $min\_cost$ for a subgroup size $i$, then $min\_costs[i], best\_objects[i]$, and $best\_subgroups[i]$ are updated (Lines 27-31). 


\begin{table}[ht!]
\renewcommand{\arraystretch}{1.2}
\centering
\caption{Example of the MFSNNK-BF algorithm}
\begin{adjustbox}{max width=\columnwidth}
\begin{tabular}{|l|l|l|l|l|l|l|}
\hline 
\textbf{Step} & \textbf{Elm} & $m=3$, $m=4$, $m=5$ & $best\_obj$ &$min\_costs$ &$total\_cost$&\textbf{P} (updated) \\ \hline
\multirow{2}{*}{1} & \multirow{2}{*}{$root$} & $R_5$ : 1.2, 1.775, 2.475 & \{$\emptyset,\emptyset,\emptyset$\}&\{$\infty, \infty, \infty$\} &5.45&\multirow{2}{*}{$R_6,R_5$} \\
& & $R_6$ : 0.225, 0.45, 0.725 & \{$\emptyset,\emptyset,\emptyset$\}&\{$\infty, \infty, \infty$\} &1.4& \\ \hline

\multirow{2}{*}{2} & \multirow{2}{*}{$R_6$} & $R_3$ : 1.2, 1.075, 1.75 & \{$\emptyset,\emptyset,\emptyset$\}&\{$\infty, \infty, \infty$\} &4.025&\multirow{2}{*}{$R_4,R_3,R_5$} \\
& & $R_4$ : 0.225, 0.775, 1.1 & \{$\emptyset,\emptyset,\emptyset$\}&\{$\infty, \infty, \infty$\} &2.13& \\ \hline

\multirow{2}{*}{3} & \multirow{2}{*}{$R_4$} & $o_6$ : 0.75, 1.325, 2.05 & \{$o_6,o_6,o_6$\}&\{$0.75, 1.325, 2.05$\} &$\emptyset$&\multirow{2}{*}{$R_3,R_5$} \\
& & $o_7$ : 0.8, 1.125, 1.625 & \{$o_6,o_7,o_7$\}&\{$0.75, 1.125, 1.625$\} &$\emptyset$& \\ \hline

\multirow{2}{*}{4} & \multirow{2}{*}{$R_3$} & $o_4$ : 1.15, 1.9, 2.75 & \{$o_6,o_6,o_7$\}&\{$0.75, 1.125, 1.625$\} &$\emptyset$&\multirow{2}{*}{$R_5$} \\
& & $o_5$ : 1.7, 2.325, 3.0 & \{$o_6,o_6,o_7$\}&\{$0.75, 1.125, 1.625$\} &$\emptyset$& \\ \hline

5 & $R_5$ & \multicolumn{5}{l|}{$R_5.query\_costs[i] > min\_costs[i]$ for all $m=i$, and hence $R_5$ can be pruned}\\ \hline

\end{tabular}
\end{adjustbox}
\label{table:fsnnk-ext-calc}
\end{table}


\begin{example}(\textbf{MFSNNK-BF}).
We continue with Example~\ref{example:cost-model} for MFSNNK-BF. Let the minimum subgroup size be 3. Then we need to find the best objects and the corresponding subgroups for $m=3$, $m=4$, and $m=5$. The algorithm steps are shown in Table~\ref{table:fsnnk-ext-calc}. Column $Elm$ shows the elements popped out from the queue $P$. The following column show the aggregate costs ($f_m(cost)$) for different subgroup sizes. $total\_cost$ is the sum of $f_m(cost)$ for all subgroup sizes.

At start, $min\_costs$ is initialized with value $\infty$, and $root$ is pushed into $P$. Then nodes are popped and calculations are performed as shown in Step 1 and Step 2. At Step 3, $R_4$ is popped out, which is a leaf node. The algorithm updates $min\_costs$, $best\_objects$ and the subgroup set $best\_subgroups$ as $R_4$ has objects as children. When a object is popped, $best\_objects$ are updated according to $min\_costs$. When $P$ becomes empty after Step 5, the algorithm returns $(o_6,\{q_4,q_1,q_2\})$, $(o_7,\{q_4,q_1,q_3,q_2\})$, $(o_7,\{q_4,q_1,q_3,q_2,q_5\})$. \qed

\label{example:fsnnk-extended}
\end{example}

\textbf{A relaxed pruning bound.}
Heuristic~\ref{heuristic:mfsnnk} states  that, for an IR-tree node $N$, if $min\_cost_i$ is the smallest cost for subgroup size $i$ found so far, then we can prune $N$ if $f(cost(sg_i, N)) \geq min\_cost_i$ for any $i \in [m..n]$. Here, $sg_i$ denotes the best subgroup of size $i$ corresponding to $N$. The MFSNNK-BF algorithm based on this heuristic has a for-loop to compute  $f(cost(sg_i, N))$ and test if $f(cost(sg_i, N)) \geq min\_cost_i$ holds for any $i$ (Lines 12 to 17 in Algorithm~\ref{algorithm:fsnnk-ex}).

A possible simplification is to only test whether $f(cost\-(sg_m, N)) \geq min\_cost_n$, i.e., whether the best subgroup of size $m$ corresponding to $N$ has a cost lower than the $min\_cost$ for the whole group of size $n$ found so far. If this holds, then $N$ can be safely pruned, as formalized by the following heuristic. 

\begin{heuristic} \label{heuristic:mfsnnk2}
Let $N$ be an IR-tree node and $m$ be the minimum subgroup size. Let $sg_m$ be the best subgroup of size $m$ corresponding to $N$, and $min\_cost_n$ be the smallest cost for the whole group of size $n$ from any object retrieved so far. We can safely prune $N$ if $f(cost(sg_m, N)) \geq min\_cost_n$.
\end{heuristic}

The proof is straightforward. Since we consider a monotonic aggregate cost function, we have:
$$f(cost(sg_m, N)) \le f(cost(sg_{m+1}, N)) \le ... \le f(cost(sg_n, N)).$$
 If 
 $$f(cost\-(sg_m, N)) \geq min\_cost_n,$$
then 
$$min\_cost_n \le f(cost(sg_m, N)) \le ... \le f(cost(sg_n, N)).$$
Thus, we can safely prune $N$. Applying this heuristic, Lines 12 to 17 of Algorithm~\ref{algorithm:fsnnk-ex}
can be replaced by: 


\begin{equation*}
\begin{split}
\text{If }f(cost(sg_m, N_c))<min\_costs[n]\text{ then} \\ P.push(N_c, f(cost(sg_m, N_c)))
\end{split}
\end{equation*}

Note that, while this heuristic simplifies the node pruning computation, it also relaxes the pruning bound, which may cause more nodes to be processed. We will use experiments to study the effectiveness of this heuristic. 

\subsection{Discussion}
All the algorithms presented in the previous subsections can be straightforwardly extended to find the $k$ best objects. Both the GNNK-BF and FSNNK-BF algorithms incrementally output the best objects. The first $k$ objects accessed by these algorithms are the $k$ best objects. Particularly, in the case of FSNNK-BF, we can use a queue to store the 
$k$ best objects and the corresponding best subgroups. For the GNNK-BB, FSNNK-BB and MFSNNK-BF algorithms, we can use a heap of size $k$ to hold $k$ currently found best objects. When the algorithms terminate, the heap contains the $k$ best objects. Same as in FSNNK-BF, for the subgroup queries we can store the 
best objects and the corresponding best subgroups together, so that when the algorithms terminate, we not only obtain the best objects but also the corresponding best subgroups.

Though in our problem formulation, we assume that all users in the group have equal priorities, our proposed cost function can be adapted for users with different priorities. Assume that each individual query $q_i$ has a priority $p_i$ associated with it, where for a group of $n$ queries $p_1 + p_2 + ... + p_n = n$. To incorporate user priorities, we need to modify our definition of aggregate cost function as follows: $ f(cost(Q, o)) = f(\frac{cost(q_j, o)}{p_i} : q_j \in Q) $. Thus, users with higher priorities (i.e., larger priority values $p_i$) would have lower costs, and hence the algorithms will tend to converge more to the objects that are spatially closer and textually more similar to the users with higher priority.


\section{Cost Analysis}\label{sec:complexitySummary}

\begin{table*}[ht!]
\renewcommand{\arraystretch}{1.3}
\centering
\caption{Summary of Costs}
\begin{adjustbox}{max width=0.9\textwidth}

    \begin{tabular}{l|l|l}
    \hline
    Algorithm &  I/O  & CPU \\
    \hline
    GNNK-BB     & $io_i + (1-w_{gb}) (\frac{|D|}{C_e-1}+ \frac{|D|}{C_e} \cdot  io_l)$ &$cpu_i + (1-w_{gb}) (\frac{|D|}{C_e-1} \cdot cpu_g+ \frac{|D|}{C_e} \cdot  cpu_l)$ \\
    \hline
    GNNK-BF    & $io_i + (1-w_{gf}) (\frac{|D|}{C_e-1}+ \frac{|D|}{C_e} \cdot  io_l) $ & $cpu_i + (1-w_{gf}) (\frac{|D|}{C_e-1}\cdot cpu_g+ \frac{|D|}{C_e} \cdot  cpu_l)$ \\
    \hline
    FSNNK-BB    & $io_i + (1-w_{sb}) (\frac{|D|}{C_e-1}+ \frac{|D|}{C_e} \cdot  io_l)$ & $cpu_i +  (1-w_{sb}) (\frac{|D|}{C_e-1}\cdot cpu_s+ \frac{|D|}{C_e} \cdot  cpu_l)$      \\
    \hline
    FSNNK-BF    & $io_i + (1-w_{sf}) (\frac{|D|}{C_e-1}+ \frac{|D|}{C_e} \cdot  io_l)$ & $cpu_i + (1-w_{sf}) (\frac{|D|}{C_e-1}\cdot cpu_s+ \frac{|D|}{C_e} \cdot  cpu_l)$ \\
    \hline
    MFSNNK-N    & $io_i + (n-m+1)(1-w_{sf}) (\frac{|D|}{C_e-1}+ \frac{|D|}{C_e} \cdot  io_l)$ & $cpu_i + (n-m+1)(1-w_{sf}) (\frac{|D|}{C_e-1}\cdot cpu_s+ \frac{|D|}{C_e} \cdot  cpu_l)$     \\
    \hline
    MFSNNK-BF    & $io_i + (1-w_{mb}) (\frac{|D|}{C_e-1}+ \frac{|D|}{C_e} \cdot  io_l)$ & $cpu_i +(1-w_{mb}) (\frac{|D|}{C_e-1}\cdot cpu_m+ \frac{|D|}{C_e} \cdot  cpu_l)$  \\
    \hline

\end{tabular}
\end{adjustbox}
\label{tab:complexitySummary}
\end{table*}

We analytically compare the I/O cost and CPU cost of 
the algorithms including GNNK-BB, GNNK-BF, FSNNK-BB, FSNNK-BF, MFSNNK-N, and MFSNNK-BF.
Table~\ref{tab:complexitySummary} summarizes the analytical results. Note that 
MFSNNK-N calls FSNNK-BF for $n-m+1$ times. Its costs are just a multiplication of those of FSNNK-BF. We will omit it in the discussion and simply list 
its costs in the table.

We use the following notation in the analysis. 
Let $C_m$ be the maximum number of entries in a disk block:
$$C_m = \text{block size} / \text{size of a data entry}$$
Let $C_{e}$ be the effective capacity of the 
IR-tree used to index the dataset $D$, i.e., the average number of entries in an IR-tree node. 
Let $|D|$ be the size of $D$. 
The average height of an IR-tree is
$h=\left\lceil\log_{C_e}{|D|}\right\rceil$. The
expected number of nodes in an IR-tree is the total number of nodes
in all tree levels (leaf nodes being level 1 and the root node
being level $h$), which is:  
\begin{equation*}
\resizebox{.9\hsize}{!}{
$\displaystyle \sum_{i=1}^{h}{\frac{|D|}{C_e^i}} =
|D|\left(\frac{1}{C_e} + \frac{1}{C_e^2}+\cdots+\frac{1}{C_e^h}\right) =
\frac{|D|}{C_e-1}(1-\frac{1}{C_e^h})\approx \frac{|D|}{C_e-1}.$
}
\end{equation*}

We assume that an IR-tree node size equals a disk block.
 
According to the structure of the IR-tree, an inverted index that maps keywords to the inner nodes of the tree
is stored separately from the tree structure.
When a group spatial keyword query is issued this inverted index is preloaded 
for all the query keywords, which will be used to guide the search to tree nodes that contain the query keywords.
The cost of this preloading, which is proportional to the number of keywords in both the data points and 
the queries, is the same for every algorithm studied. We denote the I/O cost and CPU cost 
of the preloading by $io_i$ and $cpu_i$, respectively.

\subsection{I/O Cost} \label{sec:ioCost}
For all the algorithms studied, the I/O costs depend on the number
of IR-tree nodes accessed. Further, when a leaf node is accessed, its corresponding 
inverted index that maps the keywords to the data points in the node is accessed as well. 
Analyzing the I/O cost of accessing an inverted index is beyond the scope of this paper. 
For simplicity, we denote this I/O cost by $io_l$, and the associated CPU cost by $cpu_l$.

GNNK-BB, GNNK-BF, FSNNK-BB, FSNNK-BF, and MFSNNK-BF all traverses the IR-tree 
for only once. In the worst case, all the tree nodes plus the inverted index of all the leaf nodes are accessed. 
Thus, the worst-case I/O costs for these methods are the same:
$\displaystyle \frac{|D|}{C_e-1} + \frac{|D|}{C_e} \cdot  io_l $. 

In the average case, some of the IR-tree nodes 
are pruned during the traversal. We quantify the percentage of
pruned nodes in the tree traversal as the pruning power, denoted by $w$; the number of nodes accessed is
then $\displaystyle (1-w) (\frac{|D|}{C_e-1}+ \frac{|D|}{C_e} \cdot  io_l) $ for all the algorithms except FSNNK-BF, where $w$ should be
replaced by $w_{gb}$, $w_{gf}$, $w_{sb}$,  $w_{sf}$, and $w_{mb}$ for GNNK-BB, GNNK-BF, FSNNK-BB, FSNNK-BF, and MFSNNK-BF, respectively.
Note that we use $w$ to represent the pruning power on both inner nodes and leaf nodes, which might be different in reality. We argue that 
this is still a reasonable simplification since the number of leaf nodes pruned will be proportional to the number of inner nodes pruned. Also we aim to compare 
the costs of the different algorithms, not to compute the exact costs. 

The pruning power of the different algorithms is associated with the metrics 
used to determine whether a tree node needs to be accessed. In the algorithms studied, the same pruning metric (e.g., $min\_cost$)  
is used for different algorithms of the same query variant (GNNK-BB and GNNK-BF for the GNNK query). However, the order that the tree nodes are accessed in the different algorithms of the same query variant (e.g., GNNK-BB and GNNK-BF) are different. This leads to different shrinking rates of the value of the pruning metric. In particular, 
the BF algorithms and MFSNNK-BF use best-first traversals, which always access the node with the smallest (estimated) optimization function value first. 
In comparison, the BB algorithms simply push the tree nodes into a stack, and access the node at the top of the stack regardless of the optimization function value. 
Heuristically, the BF algorithms'  pruning metric values should shrink faster. Additionally, the BF algorithms terminates early once a data entry is popped out from the queue, 
while the BB algorithms need to access every node in the stack anyway. Intuitively, the BF algorithms should have better pruning power than those of the corresponding BB algorithms, i.e, 
 $w_{gf} > w_{gb}$ and $w_{sf} > w_{sb}$. MFSNNK-BF only traverses the tree once, and it has a similar pruning strategy to that of FSNNK-BF. Its I/O cost is smaller than that of MFSNNK-N that calls FSNNK-BF multiple times.

\subsection{CPU Cost}
The CPU cost can be considered as the product of the CPU cost per
block (node) multiplied by the number of blocks (nodes) accessed.
The I/O cost analysis provides the number of nodes accessed. The
CPU cost per block, denoted by $cpu$, involves optimization function computation.

Both GNNK algorithms computes $f(cost(Q,N_c))$ for every child node $N_C$ when an inner node $N$ accessed, and 
$f(cost(Q,o))$ for every data point $o$ if $N$ is a leaf node. The CPU cost is proportional to the size of the query group ($n$), the size of the node $N$ ($C_e$), and 
the size of the keywords involved. Sine this per node CPU cost of both GNNK-BB and GNNK-BF is the same, we simply denote it by $cpu_{g}$.
Note that GNNK-BF still has a lower overall CPU cost as it accesses a smaller number of nodes. 

Similarly, we denote the per node CPU cost of FSNNK-BB and FSNNK-BF by $cpu_{s}$.
This CPU cost involves computing $cost(q_i,N_c)$ $(cost(q_i,o))$ for every query user $q_i$, finding to top-$m$ users, and computing 
$f()$ on the cost of these $m$ users. FSNNK-BF also has a lower overall CPU cost as it accesses a smaller number of nodes. 

MFSNNK-N has the same per node CPU cost $cpu_{s}$. Let the per node CPU cost of  MFSNNK-BF be $cpu_{m}$. This cost will be higher than $cpu_{s}$ 
as MFSNNK-N only computes the optimization function value of a given subgroup size each time it access a node, while MFSNNK-BF computes for $n-m+1$ subgroup sizes 
together. However, $cpu_{m} < (n-m+1) cpu_{s}$. This is because, as shown in lines 11 to 17 of the MFSNNK-BF algorithm, the functions $cost(q_i,N_c)$ are computed for 
only once rather than $n-m+1$ times, and the function $f()$ for the different sub-group size are computed progressively instead of repeatedly. As a result, the overall CPU cost 
of MFSNNK-BF will be lower than that of MFSNNK-N.

\section{Experimental Evaluation} \label{sec:experiment}

\subsection{Experimental Settings}
We evaluate the performance of our algorithms for all three types of queries GNNK, FSNNK, and MFSNNK. The branch and bound algorithms presented in Section~\ref{subsection:baseline} (GNNK-BB and FSNNK-BB) are used as the baseline for the GNNK and FSNNK queries. We compare our BF algorithms, GNNK-BF and FSNNK-BF (Section~\ref{subsection:pq}) with baselines. 
We use the MFSNNK-N algorithm as the baseline algorithm for the MFSNNK queries, and compare it with the MFSNNK-BF algorithm proposed in 
Section~\ref{subsection:fsnnk-extended}.

\begin{table}[ht!]
\renewcommand{\arraystretch}{1.0}
\caption{Dataset properties}
\label{table:dataset-properties}
\centering
\begin{adjustbox}{max width=\columnwidth}
\begin{tabular}{|c|c|c|}
\hline
Parameter               & \textbf{Flickr}   & \textbf{Yelp} \\ \hline
Dataset size                & 1,500,000     & 60,667    \\
Number of unique keywords   & 566,432       & 783       \\
Total number of keywords    & 11,579,622    & 176,697   \\
Avg. number of keywords per object & 7.72             & 2.91         \\  \hline
\end{tabular}
\end{adjustbox}
\end{table}

\begin{table}[ht!]
\caption{Query parameters}
\centering
\begin{adjustbox}{max width=\columnwidth}
\begin{tabular}{|c|c|c|}
\hline
 Parameter name                         		& Values     						& Default Value \\ 
 \hline
 Number of queried data points ($k$)        	& 1, 10, 20, 30, 40, 50       			& 10     \\
 Query group size ($n$)                       	& 10, 20, 40, 60, 80           	& 10    \\
 Subgroup size ($m$, \%$n$)    		& 40\%, 50\%, 60\%, 70\%, 80\%    	& 60\%  \\
 Number of query keywords               		& 1, 2, 4, 6, 8, 10               	& 4     \\
 Size of the query space                		& .001\%, .01\%, .02\%, .03\%, .04\%, .05\%   	& 0.01\%   \\
 Size of the query keyword set              	& 1\%, 2\%, 3\%, 4\%, 5\%  			& 3\%     \\
 Spatial vs. textual preference ($\alpha$)    	& 0.1, 0.3, 0.5, 0.7, 1.0   		& 0.5 \\
 Dataset Size (Flickr)							& 1M, 1.5M, 2M, 2.5M					& 1.5M \\
 \hline
\end{tabular}
\end{adjustbox}
\label{table:query-parameters}
\end{table}

\begin {figure}[!ht]
\setlength{\belowcaptionskip}{4pt}
\setlength{\abovecaptionskip}{4pt}
	\centering
	\begin{subfigure}[b]{0.5\linewidth}
		\centering
		\setlength{\abovecaptionskip}{2pt}
		\setlength{\belowcaptionskip}{2pt}
		\resizebox{\textwidth}{!}{
\begingroup
  \makeatletter
  \providecommand\color[2][]{%
    \GenericError{(gnuplot) \space\space\space\@spaces}{%
      Package color not loaded in conjunction with
      terminal option `colourtext'%
    }{See the gnuplot documentation for explanation.%
    }{Either use 'blacktext' in gnuplot or load the package
      color.sty in LaTeX.}%
    \renewcommand\color[2][]{}%
  }%
  \providecommand\includegraphics[2][]{%
    \GenericError{(gnuplot) \space\space\space\@spaces}{%
      Package graphicx or graphics not loaded%
    }{See the gnuplot documentation for explanation.%
    }{The gnuplot epslatex terminal needs graphicx.sty or graphics.sty.}%
    \renewcommand\includegraphics[2][]{}%
  }%
  \providecommand\rotatebox[2]{#2}%
  \@ifundefined{ifGPcolor}{%
    \newif\ifGPcolor
    \GPcolorfalse
  }{}%
  \@ifundefined{ifGPblacktext}{%
    \newif\ifGPblacktext
    \GPblacktexttrue
  }{}%
  \let\gplgaddtomacro\g@addto@macro
  \gdef\gplbacktext{}%
  \gdef\gplfronttext{}%
  \makeatother
  \ifGPblacktext
    \def\colorrgb#1{}%
    \def\colorgray#1{}%
  \else
    \ifGPcolor
      \def\colorrgb#1{\color[rgb]{#1}}%
      \def\colorgray#1{\color[gray]{#1}}%
      \expandafter\def\csname LTw\endcsname{\color{white}}%
      \expandafter\def\csname LTb\endcsname{\color{black}}%
      \expandafter\def\csname LTa\endcsname{\color{black}}%
      \expandafter\def\csname LT0\endcsname{\color[rgb]{1,0,0}}%
      \expandafter\def\csname LT1\endcsname{\color[rgb]{0,1,0}}%
      \expandafter\def\csname LT2\endcsname{\color[rgb]{0,0,1}}%
      \expandafter\def\csname LT3\endcsname{\color[rgb]{1,0,1}}%
      \expandafter\def\csname LT4\endcsname{\color[rgb]{0,1,1}}%
      \expandafter\def\csname LT5\endcsname{\color[rgb]{1,1,0}}%
      \expandafter\def\csname LT6\endcsname{\color[rgb]{0,0,0}}%
      \expandafter\def\csname LT7\endcsname{\color[rgb]{1,0.3,0}}%
      \expandafter\def\csname LT8\endcsname{\color[rgb]{0.5,0.5,0.5}}%
    \else
      \def\colorrgb#1{\color{black}}%
      \def\colorgray#1{\color[gray]{#1}}%
      \expandafter\def\csname LTw\endcsname{\color{white}}%
      \expandafter\def\csname LTb\endcsname{\color{black}}%
      \expandafter\def\csname LTa\endcsname{\color{black}}%
      \expandafter\def\csname LT0\endcsname{\color{black}}%
      \expandafter\def\csname LT1\endcsname{\color{black}}%
      \expandafter\def\csname LT2\endcsname{\color{black}}%
      \expandafter\def\csname LT3\endcsname{\color{black}}%
      \expandafter\def\csname LT4\endcsname{\color{black}}%
      \expandafter\def\csname LT5\endcsname{\color{black}}%
      \expandafter\def\csname LT6\endcsname{\color{black}}%
      \expandafter\def\csname LT7\endcsname{\color{black}}%
      \expandafter\def\csname LT8\endcsname{\color{black}}%
    \fi
  \fi
    \setlength{\unitlength}{0.0500bp}%
    \ifx\gptboxheight\undefined%
      \newlength{\gptboxheight}%
      \newlength{\gptboxwidth}%
      \newsavebox{\gptboxtext}%
    \fi%
    \setlength{\fboxrule}{0.5pt}%
    \setlength{\fboxsep}{1pt}%
\begin{picture}(4320.00,2880.00)%
    \gplgaddtomacro\gplbacktext{%
      \csname LTb\endcsname%
      \put(688,512){\makebox(0,0)[r]{\strut{}$0$}}%
      \put(688,730){\makebox(0,0)[r]{\strut{}$500$}}%
      \put(688,947){\makebox(0,0)[r]{\strut{}$1000$}}%
      \put(688,1165){\makebox(0,0)[r]{\strut{}$1500$}}%
      \put(688,1382){\makebox(0,0)[r]{\strut{}$2000$}}%
      \put(688,1600){\makebox(0,0)[r]{\strut{}$2500$}}%
      \put(688,1817){\makebox(0,0)[r]{\strut{}$3000$}}%
      \put(688,2035){\makebox(0,0)[r]{\strut{}$3500$}}%
      \put(688,2252){\makebox(0,0)[r]{\strut{}$4000$}}%
      \put(688,2470){\makebox(0,0)[r]{\strut{}$4500$}}%
      \put(688,2687){\makebox(0,0)[r]{\strut{}$5000$}}%
      \put(784,352){\makebox(0,0){\strut{}1}}%
      \put(1380,352){\makebox(0,0){\strut{}10}}%
      \put(2043,352){\makebox(0,0){\strut{}20}}%
      \put(2706,352){\makebox(0,0){\strut{}30}}%
      \put(3368,352){\makebox(0,0){\strut{}40}}%
      \put(4031,352){\makebox(0,0){\strut{}50}}%
    }%
    \gplgaddtomacro\gplfronttext{%
      \csname LTb\endcsname%
      \put(128,1599){\rotatebox{-270}{\makebox(0,0){\strut{}running time (ms)}}}%
      \put(2407,112){\makebox(0,0){\strut{}$k$}}%
      \csname LTb\endcsname%
      \put(3296,2544){\makebox(0,0)[r]{\strut{}GNNK-BB (SUM)}}%
      \csname LTb\endcsname%
      \put(3296,2384){\makebox(0,0)[r]{\strut{}GNNK-BF (SUM)}}%
      \csname LTb\endcsname%
      \put(3296,2224){\makebox(0,0)[r]{\strut{}GNNK-BB (MAX)}}%
      \csname LTb\endcsname%
      \put(3296,2064){\makebox(0,0)[r]{\strut{}GNNK-BF (MAX)}}%
    }%
    \gplbacktext
    \put(0,0){\includegraphics{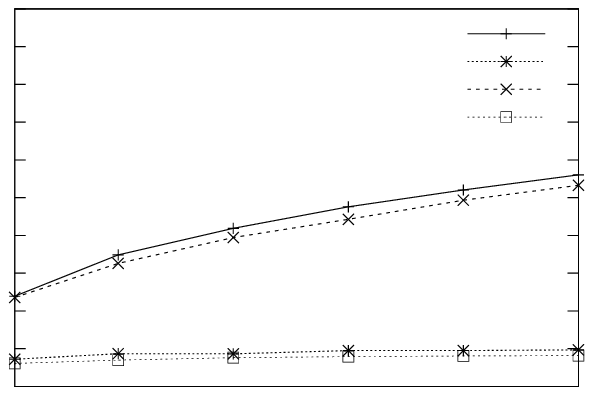}}%
    \gplfronttext
  \end{picture}%
\endgroup

		}
		\caption{}
		\label{graph:k-run}
	\end{subfigure}%
	\begin{subfigure}[b]{0.5\linewidth}
		\centering
		\setlength{\abovecaptionskip}{2pt}
		\setlength{\belowcaptionskip}{2pt}
		\resizebox{\textwidth}{!}{
\begingroup
  \makeatletter
  \providecommand\color[2][]{%
    \GenericError{(gnuplot) \space\space\space\@spaces}{%
      Package color not loaded in conjunction with
      terminal option `colourtext'%
    }{See the gnuplot documentation for explanation.%
    }{Either use 'blacktext' in gnuplot or load the package
      color.sty in LaTeX.}%
    \renewcommand\color[2][]{}%
  }%
  \providecommand\includegraphics[2][]{%
    \GenericError{(gnuplot) \space\space\space\@spaces}{%
      Package graphicx or graphics not loaded%
    }{See the gnuplot documentation for explanation.%
    }{The gnuplot epslatex terminal needs graphicx.sty or graphics.sty.}%
    \renewcommand\includegraphics[2][]{}%
  }%
  \providecommand\rotatebox[2]{#2}%
  \@ifundefined{ifGPcolor}{%
    \newif\ifGPcolor
    \GPcolorfalse
  }{}%
  \@ifundefined{ifGPblacktext}{%
    \newif\ifGPblacktext
    \GPblacktexttrue
  }{}%
  \let\gplgaddtomacro\g@addto@macro
  \gdef\gplbacktext{}%
  \gdef\gplfronttext{}%
  \makeatother
  \ifGPblacktext
    \def\colorrgb#1{}%
    \def\colorgray#1{}%
  \else
    \ifGPcolor
      \def\colorrgb#1{\color[rgb]{#1}}%
      \def\colorgray#1{\color[gray]{#1}}%
      \expandafter\def\csname LTw\endcsname{\color{white}}%
      \expandafter\def\csname LTb\endcsname{\color{black}}%
      \expandafter\def\csname LTa\endcsname{\color{black}}%
      \expandafter\def\csname LT0\endcsname{\color[rgb]{1,0,0}}%
      \expandafter\def\csname LT1\endcsname{\color[rgb]{0,1,0}}%
      \expandafter\def\csname LT2\endcsname{\color[rgb]{0,0,1}}%
      \expandafter\def\csname LT3\endcsname{\color[rgb]{1,0,1}}%
      \expandafter\def\csname LT4\endcsname{\color[rgb]{0,1,1}}%
      \expandafter\def\csname LT5\endcsname{\color[rgb]{1,1,0}}%
      \expandafter\def\csname LT6\endcsname{\color[rgb]{0,0,0}}%
      \expandafter\def\csname LT7\endcsname{\color[rgb]{1,0.3,0}}%
      \expandafter\def\csname LT8\endcsname{\color[rgb]{0.5,0.5,0.5}}%
    \else
      \def\colorrgb#1{\color{black}}%
      \def\colorgray#1{\color[gray]{#1}}%
      \expandafter\def\csname LTw\endcsname{\color{white}}%
      \expandafter\def\csname LTb\endcsname{\color{black}}%
      \expandafter\def\csname LTa\endcsname{\color{black}}%
      \expandafter\def\csname LT0\endcsname{\color{black}}%
      \expandafter\def\csname LT1\endcsname{\color{black}}%
      \expandafter\def\csname LT2\endcsname{\color{black}}%
      \expandafter\def\csname LT3\endcsname{\color{black}}%
      \expandafter\def\csname LT4\endcsname{\color{black}}%
      \expandafter\def\csname LT5\endcsname{\color{black}}%
      \expandafter\def\csname LT6\endcsname{\color{black}}%
      \expandafter\def\csname LT7\endcsname{\color{black}}%
      \expandafter\def\csname LT8\endcsname{\color{black}}%
    \fi
  \fi
    \setlength{\unitlength}{0.0500bp}%
    \ifx\gptboxheight\undefined%
      \newlength{\gptboxheight}%
      \newlength{\gptboxwidth}%
      \newsavebox{\gptboxtext}%
    \fi%
    \setlength{\fboxrule}{0.5pt}%
    \setlength{\fboxsep}{1pt}%
\begin{picture}(4320.00,2880.00)%
    \gplgaddtomacro\gplbacktext{%
      \csname LTb\endcsname%
      \put(592,512){\makebox(0,0)[r]{\strut{}$0$}}%
      \put(592,784){\makebox(0,0)[r]{\strut{}$100$}}%
      \put(592,1056){\makebox(0,0)[r]{\strut{}$200$}}%
      \put(592,1328){\makebox(0,0)[r]{\strut{}$300$}}%
      \put(592,1600){\makebox(0,0)[r]{\strut{}$400$}}%
      \put(592,1871){\makebox(0,0)[r]{\strut{}$500$}}%
      \put(592,2143){\makebox(0,0)[r]{\strut{}$600$}}%
      \put(592,2415){\makebox(0,0)[r]{\strut{}$700$}}%
      \put(592,2687){\makebox(0,0)[r]{\strut{}$800$}}%
      \put(688,352){\makebox(0,0){\strut{}1}}%
      \put(1302,352){\makebox(0,0){\strut{}10}}%
      \put(1984,352){\makebox(0,0){\strut{}20}}%
      \put(2667,352){\makebox(0,0){\strut{}30}}%
      \put(3349,352){\makebox(0,0){\strut{}40}}%
      \put(4031,352){\makebox(0,0){\strut{}50}}%
    }%
    \gplgaddtomacro\gplfronttext{%
      \csname LTb\endcsname%
      \put(128,1599){\rotatebox{-270}{\makebox(0,0){\strut{}\# page accesses}}}%
      \put(2359,112){\makebox(0,0){\strut{}$k$}}%
      \csname LTb\endcsname%
      \put(3296,2544){\makebox(0,0)[r]{\strut{}GNNK-BB (SUM)}}%
      \csname LTb\endcsname%
      \put(3296,2384){\makebox(0,0)[r]{\strut{}GNNK-BF (SUM)}}%
      \csname LTb\endcsname%
      \put(3296,2224){\makebox(0,0)[r]{\strut{}GNNK-BB (MAX)}}%
      \csname LTb\endcsname%
      \put(3296,2064){\makebox(0,0)[r]{\strut{}GNNK-BF (MAX)}}%
    }%
    \gplbacktext
    \put(0,0){\includegraphics{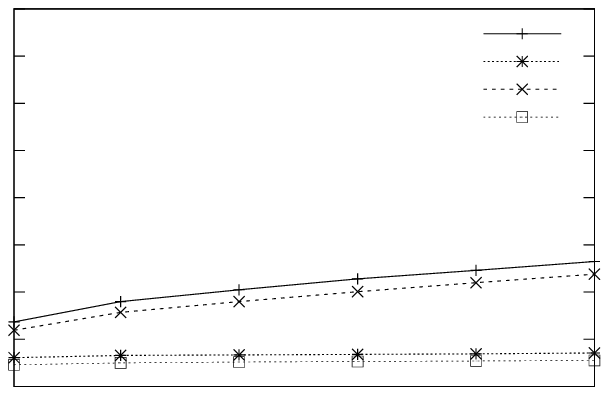}}%
    \gplfronttext
  \end{picture}%
\endgroup

		}
		\caption{}
		\label{graph:k-io}
	\end{subfigure}%

	\begin{subfigure}[b]{0.5\linewidth}
		\centering
		\setlength{\abovecaptionskip}{2pt}
		\setlength{\belowcaptionskip}{2pt}
		\resizebox{\textwidth}{!}{
\begingroup
  \makeatletter
  \providecommand\color[2][]{%
    \GenericError{(gnuplot) \space\space\space\@spaces}{%
      Package color not loaded in conjunction with
      terminal option `colourtext'%
    }{See the gnuplot documentation for explanation.%
    }{Either use 'blacktext' in gnuplot or load the package
      color.sty in LaTeX.}%
    \renewcommand\color[2][]{}%
  }%
  \providecommand\includegraphics[2][]{%
    \GenericError{(gnuplot) \space\space\space\@spaces}{%
      Package graphicx or graphics not loaded%
    }{See the gnuplot documentation for explanation.%
    }{The gnuplot epslatex terminal needs graphicx.sty or graphics.sty.}%
    \renewcommand\includegraphics[2][]{}%
  }%
  \providecommand\rotatebox[2]{#2}%
  \@ifundefined{ifGPcolor}{%
    \newif\ifGPcolor
    \GPcolorfalse
  }{}%
  \@ifundefined{ifGPblacktext}{%
    \newif\ifGPblacktext
    \GPblacktexttrue
  }{}%
  \let\gplgaddtomacro\g@addto@macro
  \gdef\gplbacktext{}%
  \gdef\gplfronttext{}%
  \makeatother
  \ifGPblacktext
    \def\colorrgb#1{}%
    \def\colorgray#1{}%
  \else
    \ifGPcolor
      \def\colorrgb#1{\color[rgb]{#1}}%
      \def\colorgray#1{\color[gray]{#1}}%
      \expandafter\def\csname LTw\endcsname{\color{white}}%
      \expandafter\def\csname LTb\endcsname{\color{black}}%
      \expandafter\def\csname LTa\endcsname{\color{black}}%
      \expandafter\def\csname LT0\endcsname{\color[rgb]{1,0,0}}%
      \expandafter\def\csname LT1\endcsname{\color[rgb]{0,1,0}}%
      \expandafter\def\csname LT2\endcsname{\color[rgb]{0,0,1}}%
      \expandafter\def\csname LT3\endcsname{\color[rgb]{1,0,1}}%
      \expandafter\def\csname LT4\endcsname{\color[rgb]{0,1,1}}%
      \expandafter\def\csname LT5\endcsname{\color[rgb]{1,1,0}}%
      \expandafter\def\csname LT6\endcsname{\color[rgb]{0,0,0}}%
      \expandafter\def\csname LT7\endcsname{\color[rgb]{1,0.3,0}}%
      \expandafter\def\csname LT8\endcsname{\color[rgb]{0.5,0.5,0.5}}%
    \else
      \def\colorrgb#1{\color{black}}%
      \def\colorgray#1{\color[gray]{#1}}%
      \expandafter\def\csname LTw\endcsname{\color{white}}%
      \expandafter\def\csname LTb\endcsname{\color{black}}%
      \expandafter\def\csname LTa\endcsname{\color{black}}%
      \expandafter\def\csname LT0\endcsname{\color{black}}%
      \expandafter\def\csname LT1\endcsname{\color{black}}%
      \expandafter\def\csname LT2\endcsname{\color{black}}%
      \expandafter\def\csname LT3\endcsname{\color{black}}%
      \expandafter\def\csname LT4\endcsname{\color{black}}%
      \expandafter\def\csname LT5\endcsname{\color{black}}%
      \expandafter\def\csname LT6\endcsname{\color{black}}%
      \expandafter\def\csname LT7\endcsname{\color{black}}%
      \expandafter\def\csname LT8\endcsname{\color{black}}%
    \fi
  \fi
    \setlength{\unitlength}{0.0500bp}%
    \ifx\gptboxheight\undefined%
      \newlength{\gptboxheight}%
      \newlength{\gptboxwidth}%
      \newsavebox{\gptboxtext}%
    \fi%
    \setlength{\fboxrule}{0.5pt}%
    \setlength{\fboxsep}{1pt}%
\begin{picture}(4320.00,2880.00)%
    \gplgaddtomacro\gplbacktext{%
      \csname LTb\endcsname%
      \put(688,512){\makebox(0,0)[r]{\strut{}$0$}}%
      \put(688,754){\makebox(0,0)[r]{\strut{}$500$}}%
      \put(688,995){\makebox(0,0)[r]{\strut{}$1000$}}%
      \put(688,1237){\makebox(0,0)[r]{\strut{}$1500$}}%
      \put(688,1479){\makebox(0,0)[r]{\strut{}$2000$}}%
      \put(688,1720){\makebox(0,0)[r]{\strut{}$2500$}}%
      \put(688,1962){\makebox(0,0)[r]{\strut{}$3000$}}%
      \put(688,2204){\makebox(0,0)[r]{\strut{}$3500$}}%
      \put(688,2445){\makebox(0,0)[r]{\strut{}$4000$}}%
      \put(688,2687){\makebox(0,0)[r]{\strut{}$4500$}}%
      \put(784,352){\makebox(0,0){\strut{}10}}%
      \put(1248,352){\makebox(0,0){\strut{}20}}%
      \put(2176,352){\makebox(0,0){\strut{}40}}%
      \put(3103,352){\makebox(0,0){\strut{}60}}%
      \put(4031,352){\makebox(0,0){\strut{}80}}%
    }%
    \gplgaddtomacro\gplfronttext{%
      \csname LTb\endcsname%
      \put(128,1599){\rotatebox{-270}{\makebox(0,0){\strut{}running time (ms)}}}%
      \put(2407,112){\makebox(0,0){\strut{}Group Size}}%
      \csname LTb\endcsname%
      \put(3296,2544){\makebox(0,0)[r]{\strut{}GNNK-BB (SUM)}}%
      \csname LTb\endcsname%
      \put(3296,2384){\makebox(0,0)[r]{\strut{}GNNK-BF (SUM)}}%
      \csname LTb\endcsname%
      \put(3296,2224){\makebox(0,0)[r]{\strut{}GNNK-BB (MAX)}}%
      \csname LTb\endcsname%
      \put(3296,2064){\makebox(0,0)[r]{\strut{}GNNK-BF (MAX)}}%
    }%
    \gplbacktext
    \put(0,0){\includegraphics{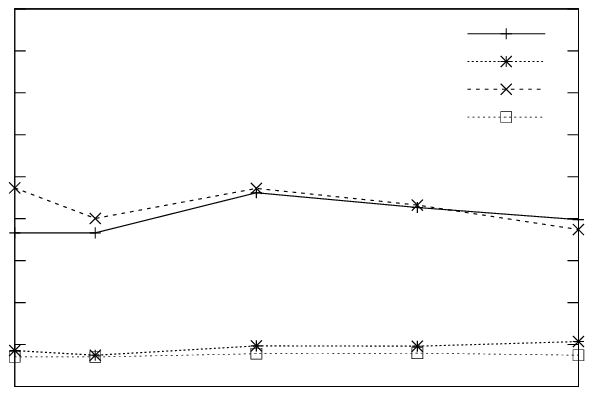}}%
    \gplfronttext
  \end{picture}%
\endgroup

		}
		\caption{}
		\label{graph:group-run}
	\end{subfigure}%
	\begin{subfigure}[b]{0.5\linewidth}
		\centering
		\setlength{\abovecaptionskip}{2pt}
		\setlength{\belowcaptionskip}{2pt}
		\resizebox{\textwidth}{!}{
\begingroup
  \makeatletter
  \providecommand\color[2][]{%
    \GenericError{(gnuplot) \space\space\space\@spaces}{%
      Package color not loaded in conjunction with
      terminal option `colourtext'%
    }{See the gnuplot documentation for explanation.%
    }{Either use 'blacktext' in gnuplot or load the package
      color.sty in LaTeX.}%
    \renewcommand\color[2][]{}%
  }%
  \providecommand\includegraphics[2][]{%
    \GenericError{(gnuplot) \space\space\space\@spaces}{%
      Package graphicx or graphics not loaded%
    }{See the gnuplot documentation for explanation.%
    }{The gnuplot epslatex terminal needs graphicx.sty or graphics.sty.}%
    \renewcommand\includegraphics[2][]{}%
  }%
  \providecommand\rotatebox[2]{#2}%
  \@ifundefined{ifGPcolor}{%
    \newif\ifGPcolor
    \GPcolorfalse
  }{}%
  \@ifundefined{ifGPblacktext}{%
    \newif\ifGPblacktext
    \GPblacktexttrue
  }{}%
  \let\gplgaddtomacro\g@addto@macro
  \gdef\gplbacktext{}%
  \gdef\gplfronttext{}%
  \makeatother
  \ifGPblacktext
    \def\colorrgb#1{}%
    \def\colorgray#1{}%
  \else
    \ifGPcolor
      \def\colorrgb#1{\color[rgb]{#1}}%
      \def\colorgray#1{\color[gray]{#1}}%
      \expandafter\def\csname LTw\endcsname{\color{white}}%
      \expandafter\def\csname LTb\endcsname{\color{black}}%
      \expandafter\def\csname LTa\endcsname{\color{black}}%
      \expandafter\def\csname LT0\endcsname{\color[rgb]{1,0,0}}%
      \expandafter\def\csname LT1\endcsname{\color[rgb]{0,1,0}}%
      \expandafter\def\csname LT2\endcsname{\color[rgb]{0,0,1}}%
      \expandafter\def\csname LT3\endcsname{\color[rgb]{1,0,1}}%
      \expandafter\def\csname LT4\endcsname{\color[rgb]{0,1,1}}%
      \expandafter\def\csname LT5\endcsname{\color[rgb]{1,1,0}}%
      \expandafter\def\csname LT6\endcsname{\color[rgb]{0,0,0}}%
      \expandafter\def\csname LT7\endcsname{\color[rgb]{1,0.3,0}}%
      \expandafter\def\csname LT8\endcsname{\color[rgb]{0.5,0.5,0.5}}%
    \else
      \def\colorrgb#1{\color{black}}%
      \def\colorgray#1{\color[gray]{#1}}%
      \expandafter\def\csname LTw\endcsname{\color{white}}%
      \expandafter\def\csname LTb\endcsname{\color{black}}%
      \expandafter\def\csname LTa\endcsname{\color{black}}%
      \expandafter\def\csname LT0\endcsname{\color{black}}%
      \expandafter\def\csname LT1\endcsname{\color{black}}%
      \expandafter\def\csname LT2\endcsname{\color{black}}%
      \expandafter\def\csname LT3\endcsname{\color{black}}%
      \expandafter\def\csname LT4\endcsname{\color{black}}%
      \expandafter\def\csname LT5\endcsname{\color{black}}%
      \expandafter\def\csname LT6\endcsname{\color{black}}%
      \expandafter\def\csname LT7\endcsname{\color{black}}%
      \expandafter\def\csname LT8\endcsname{\color{black}}%
    \fi
  \fi
    \setlength{\unitlength}{0.0500bp}%
    \ifx\gptboxheight\undefined%
      \newlength{\gptboxheight}%
      \newlength{\gptboxwidth}%
      \newsavebox{\gptboxtext}%
    \fi%
    \setlength{\fboxrule}{0.5pt}%
    \setlength{\fboxsep}{1pt}%
\begin{picture}(4320.00,2880.00)%
    \gplgaddtomacro\gplbacktext{%
      \csname LTb\endcsname%
      \put(688,512){\makebox(0,0)[r]{\strut{}$0$}}%
      \put(688,823){\makebox(0,0)[r]{\strut{}$200$}}%
      \put(688,1133){\makebox(0,0)[r]{\strut{}$400$}}%
      \put(688,1444){\makebox(0,0)[r]{\strut{}$600$}}%
      \put(688,1755){\makebox(0,0)[r]{\strut{}$800$}}%
      \put(688,2066){\makebox(0,0)[r]{\strut{}$1000$}}%
      \put(688,2376){\makebox(0,0)[r]{\strut{}$1200$}}%
      \put(688,2687){\makebox(0,0)[r]{\strut{}$1400$}}%
      \put(784,352){\makebox(0,0){\strut{}10}}%
      \put(1248,352){\makebox(0,0){\strut{}20}}%
      \put(2176,352){\makebox(0,0){\strut{}40}}%
      \put(3103,352){\makebox(0,0){\strut{}60}}%
      \put(4031,352){\makebox(0,0){\strut{}80}}%
    }%
    \gplgaddtomacro\gplfronttext{%
      \csname LTb\endcsname%
      \put(128,1599){\rotatebox{-270}{\makebox(0,0){\strut{}\# page accesses}}}%
      \put(2407,112){\makebox(0,0){\strut{}Group Size}}%
      \csname LTb\endcsname%
      \put(3296,2544){\makebox(0,0)[r]{\strut{}GNNK-BB (SUM)}}%
      \csname LTb\endcsname%
      \put(3296,2384){\makebox(0,0)[r]{\strut{}GNNK-BF (SUM)}}%
      \csname LTb\endcsname%
      \put(3296,2224){\makebox(0,0)[r]{\strut{}GNNK-BB (MAX)}}%
      \csname LTb\endcsname%
      \put(3296,2064){\makebox(0,0)[r]{\strut{}GNNK-BF (MAX)}}%
    }%
    \gplbacktext
    \put(0,0){\includegraphics{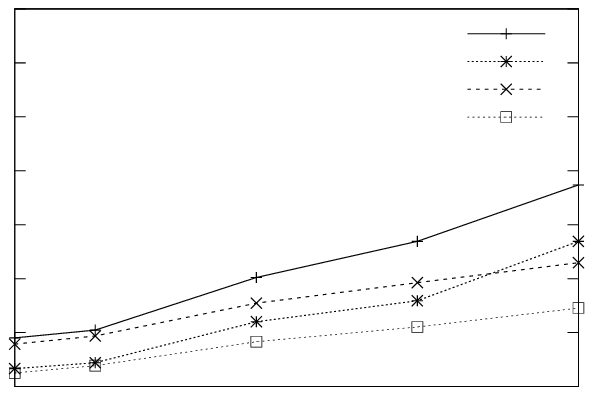}}%
    \gplfronttext
  \end{picture}%
\endgroup

		}
		\caption{}
		\label{graph:group-io}
	\end{subfigure}%

	\begin{subfigure}[b]{0.5\linewidth}
		\centering
		\setlength{\abovecaptionskip}{2pt}
		\setlength{\belowcaptionskip}{2pt}
		\resizebox{\textwidth}{!}{
\begingroup
  \makeatletter
  \providecommand\color[2][]{%
    \GenericError{(gnuplot) \space\space\space\@spaces}{%
      Package color not loaded in conjunction with
      terminal option `colourtext'%
    }{See the gnuplot documentation for explanation.%
    }{Either use 'blacktext' in gnuplot or load the package
      color.sty in LaTeX.}%
    \renewcommand\color[2][]{}%
  }%
  \providecommand\includegraphics[2][]{%
    \GenericError{(gnuplot) \space\space\space\@spaces}{%
      Package graphicx or graphics not loaded%
    }{See the gnuplot documentation for explanation.%
    }{The gnuplot epslatex terminal needs graphicx.sty or graphics.sty.}%
    \renewcommand\includegraphics[2][]{}%
  }%
  \providecommand\rotatebox[2]{#2}%
  \@ifundefined{ifGPcolor}{%
    \newif\ifGPcolor
    \GPcolorfalse
  }{}%
  \@ifundefined{ifGPblacktext}{%
    \newif\ifGPblacktext
    \GPblacktexttrue
  }{}%
  \let\gplgaddtomacro\g@addto@macro
  \gdef\gplbacktext{}%
  \gdef\gplfronttext{}%
  \makeatother
  \ifGPblacktext
    \def\colorrgb#1{}%
    \def\colorgray#1{}%
  \else
    \ifGPcolor
      \def\colorrgb#1{\color[rgb]{#1}}%
      \def\colorgray#1{\color[gray]{#1}}%
      \expandafter\def\csname LTw\endcsname{\color{white}}%
      \expandafter\def\csname LTb\endcsname{\color{black}}%
      \expandafter\def\csname LTa\endcsname{\color{black}}%
      \expandafter\def\csname LT0\endcsname{\color[rgb]{1,0,0}}%
      \expandafter\def\csname LT1\endcsname{\color[rgb]{0,1,0}}%
      \expandafter\def\csname LT2\endcsname{\color[rgb]{0,0,1}}%
      \expandafter\def\csname LT3\endcsname{\color[rgb]{1,0,1}}%
      \expandafter\def\csname LT4\endcsname{\color[rgb]{0,1,1}}%
      \expandafter\def\csname LT5\endcsname{\color[rgb]{1,1,0}}%
      \expandafter\def\csname LT6\endcsname{\color[rgb]{0,0,0}}%
      \expandafter\def\csname LT7\endcsname{\color[rgb]{1,0.3,0}}%
      \expandafter\def\csname LT8\endcsname{\color[rgb]{0.5,0.5,0.5}}%
    \else
      \def\colorrgb#1{\color{black}}%
      \def\colorgray#1{\color[gray]{#1}}%
      \expandafter\def\csname LTw\endcsname{\color{white}}%
      \expandafter\def\csname LTb\endcsname{\color{black}}%
      \expandafter\def\csname LTa\endcsname{\color{black}}%
      \expandafter\def\csname LT0\endcsname{\color{black}}%
      \expandafter\def\csname LT1\endcsname{\color{black}}%
      \expandafter\def\csname LT2\endcsname{\color{black}}%
      \expandafter\def\csname LT3\endcsname{\color{black}}%
      \expandafter\def\csname LT4\endcsname{\color{black}}%
      \expandafter\def\csname LT5\endcsname{\color{black}}%
      \expandafter\def\csname LT6\endcsname{\color{black}}%
      \expandafter\def\csname LT7\endcsname{\color{black}}%
      \expandafter\def\csname LT8\endcsname{\color{black}}%
    \fi
  \fi
    \setlength{\unitlength}{0.0500bp}%
    \ifx\gptboxheight\undefined%
      \newlength{\gptboxheight}%
      \newlength{\gptboxwidth}%
      \newsavebox{\gptboxtext}%
    \fi%
    \setlength{\fboxrule}{0.5pt}%
    \setlength{\fboxsep}{1pt}%
\begin{picture}(4320.00,2880.00)%
    \gplgaddtomacro\gplbacktext{%
      \csname LTb\endcsname%
      \put(688,512){\makebox(0,0)[r]{\strut{}$0$}}%
      \put(688,754){\makebox(0,0)[r]{\strut{}$500$}}%
      \put(688,995){\makebox(0,0)[r]{\strut{}$1000$}}%
      \put(688,1237){\makebox(0,0)[r]{\strut{}$1500$}}%
      \put(688,1479){\makebox(0,0)[r]{\strut{}$2000$}}%
      \put(688,1720){\makebox(0,0)[r]{\strut{}$2500$}}%
      \put(688,1962){\makebox(0,0)[r]{\strut{}$3000$}}%
      \put(688,2204){\makebox(0,0)[r]{\strut{}$3500$}}%
      \put(688,2445){\makebox(0,0)[r]{\strut{}$4000$}}%
      \put(688,2687){\makebox(0,0)[r]{\strut{}$4500$}}%
      \put(784,352){\makebox(0,0){\strut{}1}}%
      \put(1145,352){\makebox(0,0){\strut{}2}}%
      \put(1866,352){\makebox(0,0){\strut{}4}}%
      \put(3309,352){\makebox(0,0){\strut{}8}}%
      \put(4031,352){\makebox(0,0){\strut{}10}}%
    }%
    \gplgaddtomacro\gplfronttext{%
      \csname LTb\endcsname%
      \put(128,1599){\rotatebox{-270}{\makebox(0,0){\strut{}running time (ms)}}}%
      \put(2407,112){\makebox(0,0){\strut{}Number of Keywords}}%
      \csname LTb\endcsname%
      \put(3296,2544){\makebox(0,0)[r]{\strut{}GNNK-BB (SUM)}}%
      \csname LTb\endcsname%
      \put(3296,2384){\makebox(0,0)[r]{\strut{}GNNK-BF (SUM)}}%
      \csname LTb\endcsname%
      \put(3296,2224){\makebox(0,0)[r]{\strut{}GNNK-BB (MAX)}}%
      \csname LTb\endcsname%
      \put(3296,2064){\makebox(0,0)[r]{\strut{}GNNK-BF (MAX)}}%
    }%
    \gplbacktext
    \put(0,0){\includegraphics{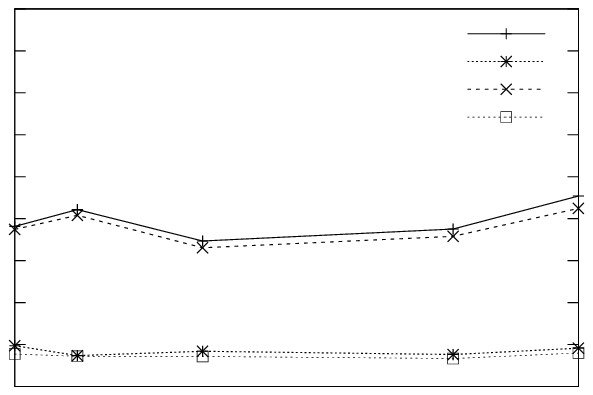}}%
    \gplfronttext
  \end{picture}%
\endgroup

		}
		\caption{}
		\label{graph:keyword-run}
	\end{subfigure}%
	\begin{subfigure}[b]{0.5\linewidth}
		\centering
		\setlength{\abovecaptionskip}{2pt}
		\setlength{\belowcaptionskip}{2pt}
		\resizebox{\textwidth}{!}{
\begingroup
  \makeatletter
  \providecommand\color[2][]{%
    \GenericError{(gnuplot) \space\space\space\@spaces}{%
      Package color not loaded in conjunction with
      terminal option `colourtext'%
    }{See the gnuplot documentation for explanation.%
    }{Either use 'blacktext' in gnuplot or load the package
      color.sty in LaTeX.}%
    \renewcommand\color[2][]{}%
  }%
  \providecommand\includegraphics[2][]{%
    \GenericError{(gnuplot) \space\space\space\@spaces}{%
      Package graphicx or graphics not loaded%
    }{See the gnuplot documentation for explanation.%
    }{The gnuplot epslatex terminal needs graphicx.sty or graphics.sty.}%
    \renewcommand\includegraphics[2][]{}%
  }%
  \providecommand\rotatebox[2]{#2}%
  \@ifundefined{ifGPcolor}{%
    \newif\ifGPcolor
    \GPcolorfalse
  }{}%
  \@ifundefined{ifGPblacktext}{%
    \newif\ifGPblacktext
    \GPblacktexttrue
  }{}%
  \let\gplgaddtomacro\g@addto@macro
  \gdef\gplbacktext{}%
  \gdef\gplfronttext{}%
  \makeatother
  \ifGPblacktext
    \def\colorrgb#1{}%
    \def\colorgray#1{}%
  \else
    \ifGPcolor
      \def\colorrgb#1{\color[rgb]{#1}}%
      \def\colorgray#1{\color[gray]{#1}}%
      \expandafter\def\csname LTw\endcsname{\color{white}}%
      \expandafter\def\csname LTb\endcsname{\color{black}}%
      \expandafter\def\csname LTa\endcsname{\color{black}}%
      \expandafter\def\csname LT0\endcsname{\color[rgb]{1,0,0}}%
      \expandafter\def\csname LT1\endcsname{\color[rgb]{0,1,0}}%
      \expandafter\def\csname LT2\endcsname{\color[rgb]{0,0,1}}%
      \expandafter\def\csname LT3\endcsname{\color[rgb]{1,0,1}}%
      \expandafter\def\csname LT4\endcsname{\color[rgb]{0,1,1}}%
      \expandafter\def\csname LT5\endcsname{\color[rgb]{1,1,0}}%
      \expandafter\def\csname LT6\endcsname{\color[rgb]{0,0,0}}%
      \expandafter\def\csname LT7\endcsname{\color[rgb]{1,0.3,0}}%
      \expandafter\def\csname LT8\endcsname{\color[rgb]{0.5,0.5,0.5}}%
    \else
      \def\colorrgb#1{\color{black}}%
      \def\colorgray#1{\color[gray]{#1}}%
      \expandafter\def\csname LTw\endcsname{\color{white}}%
      \expandafter\def\csname LTb\endcsname{\color{black}}%
      \expandafter\def\csname LTa\endcsname{\color{black}}%
      \expandafter\def\csname LT0\endcsname{\color{black}}%
      \expandafter\def\csname LT1\endcsname{\color{black}}%
      \expandafter\def\csname LT2\endcsname{\color{black}}%
      \expandafter\def\csname LT3\endcsname{\color{black}}%
      \expandafter\def\csname LT4\endcsname{\color{black}}%
      \expandafter\def\csname LT5\endcsname{\color{black}}%
      \expandafter\def\csname LT6\endcsname{\color{black}}%
      \expandafter\def\csname LT7\endcsname{\color{black}}%
      \expandafter\def\csname LT8\endcsname{\color{black}}%
    \fi
  \fi
    \setlength{\unitlength}{0.0500bp}%
    \ifx\gptboxheight\undefined%
      \newlength{\gptboxheight}%
      \newlength{\gptboxwidth}%
      \newsavebox{\gptboxtext}%
    \fi%
    \setlength{\fboxrule}{0.5pt}%
    \setlength{\fboxsep}{1pt}%
\begin{picture}(4320.00,2880.00)%
    \gplgaddtomacro\gplbacktext{%
      \csname LTb\endcsname%
      \put(592,512){\makebox(0,0)[r]{\strut{}$0$}}%
      \put(592,754){\makebox(0,0)[r]{\strut{}$100$}}%
      \put(592,995){\makebox(0,0)[r]{\strut{}$200$}}%
      \put(592,1237){\makebox(0,0)[r]{\strut{}$300$}}%
      \put(592,1479){\makebox(0,0)[r]{\strut{}$400$}}%
      \put(592,1720){\makebox(0,0)[r]{\strut{}$500$}}%
      \put(592,1962){\makebox(0,0)[r]{\strut{}$600$}}%
      \put(592,2204){\makebox(0,0)[r]{\strut{}$700$}}%
      \put(592,2445){\makebox(0,0)[r]{\strut{}$800$}}%
      \put(592,2687){\makebox(0,0)[r]{\strut{}$900$}}%
      \put(688,352){\makebox(0,0){\strut{}1}}%
      \put(1059,352){\makebox(0,0){\strut{}2}}%
      \put(1802,352){\makebox(0,0){\strut{}4}}%
      \put(3288,352){\makebox(0,0){\strut{}8}}%
      \put(4031,352){\makebox(0,0){\strut{}10}}%
    }%
    \gplgaddtomacro\gplfronttext{%
      \csname LTb\endcsname%
      \put(128,1599){\rotatebox{-270}{\makebox(0,0){\strut{}\# page accesses}}}%
      \put(2359,112){\makebox(0,0){\strut{}Number of Keywords}}%
      \csname LTb\endcsname%
      \put(3296,2544){\makebox(0,0)[r]{\strut{}GNNK-BB (SUM)}}%
      \csname LTb\endcsname%
      \put(3296,2384){\makebox(0,0)[r]{\strut{}GNNK-BF (SUM)}}%
      \csname LTb\endcsname%
      \put(3296,2224){\makebox(0,0)[r]{\strut{}GNNK-BB (MAX)}}%
      \csname LTb\endcsname%
      \put(3296,2064){\makebox(0,0)[r]{\strut{}GNNK-BF (MAX)}}%
    }%
    \gplbacktext
    \put(0,0){\includegraphics{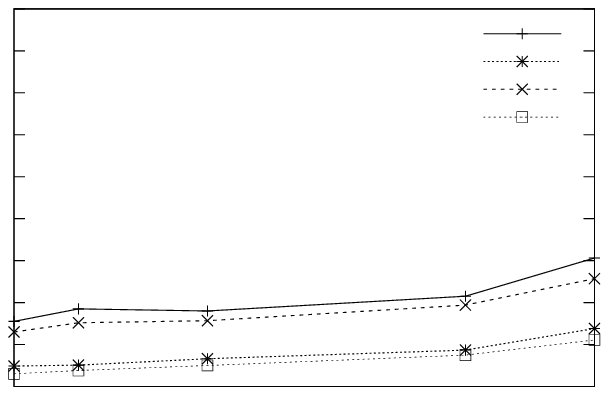}}%
    \gplfronttext
  \end{picture}%
\endgroup

		}
		\caption{}
		\label{graph:keyword-io}
	\end{subfigure}%

	\begin{subfigure}[b]{0.5\linewidth}
		\centering
		\setlength{\abovecaptionskip}{2pt}
		\setlength{\belowcaptionskip}{2pt}
		\resizebox{\textwidth}{!}{
\begingroup
  \makeatletter
  \providecommand\color[2][]{%
    \GenericError{(gnuplot) \space\space\space\@spaces}{%
      Package color not loaded in conjunction with
      terminal option `colourtext'%
    }{See the gnuplot documentation for explanation.%
    }{Either use 'blacktext' in gnuplot or load the package
      color.sty in LaTeX.}%
    \renewcommand\color[2][]{}%
  }%
  \providecommand\includegraphics[2][]{%
    \GenericError{(gnuplot) \space\space\space\@spaces}{%
      Package graphicx or graphics not loaded%
    }{See the gnuplot documentation for explanation.%
    }{The gnuplot epslatex terminal needs graphicx.sty or graphics.sty.}%
    \renewcommand\includegraphics[2][]{}%
  }%
  \providecommand\rotatebox[2]{#2}%
  \@ifundefined{ifGPcolor}{%
    \newif\ifGPcolor
    \GPcolorfalse
  }{}%
  \@ifundefined{ifGPblacktext}{%
    \newif\ifGPblacktext
    \GPblacktexttrue
  }{}%
  \let\gplgaddtomacro\g@addto@macro
  \gdef\gplbacktext{}%
  \gdef\gplfronttext{}%
  \makeatother
  \ifGPblacktext
    \def\colorrgb#1{}%
    \def\colorgray#1{}%
  \else
    \ifGPcolor
      \def\colorrgb#1{\color[rgb]{#1}}%
      \def\colorgray#1{\color[gray]{#1}}%
      \expandafter\def\csname LTw\endcsname{\color{white}}%
      \expandafter\def\csname LTb\endcsname{\color{black}}%
      \expandafter\def\csname LTa\endcsname{\color{black}}%
      \expandafter\def\csname LT0\endcsname{\color[rgb]{1,0,0}}%
      \expandafter\def\csname LT1\endcsname{\color[rgb]{0,1,0}}%
      \expandafter\def\csname LT2\endcsname{\color[rgb]{0,0,1}}%
      \expandafter\def\csname LT3\endcsname{\color[rgb]{1,0,1}}%
      \expandafter\def\csname LT4\endcsname{\color[rgb]{0,1,1}}%
      \expandafter\def\csname LT5\endcsname{\color[rgb]{1,1,0}}%
      \expandafter\def\csname LT6\endcsname{\color[rgb]{0,0,0}}%
      \expandafter\def\csname LT7\endcsname{\color[rgb]{1,0.3,0}}%
      \expandafter\def\csname LT8\endcsname{\color[rgb]{0.5,0.5,0.5}}%
    \else
      \def\colorrgb#1{\color{black}}%
      \def\colorgray#1{\color[gray]{#1}}%
      \expandafter\def\csname LTw\endcsname{\color{white}}%
      \expandafter\def\csname LTb\endcsname{\color{black}}%
      \expandafter\def\csname LTa\endcsname{\color{black}}%
      \expandafter\def\csname LT0\endcsname{\color{black}}%
      \expandafter\def\csname LT1\endcsname{\color{black}}%
      \expandafter\def\csname LT2\endcsname{\color{black}}%
      \expandafter\def\csname LT3\endcsname{\color{black}}%
      \expandafter\def\csname LT4\endcsname{\color{black}}%
      \expandafter\def\csname LT5\endcsname{\color{black}}%
      \expandafter\def\csname LT6\endcsname{\color{black}}%
      \expandafter\def\csname LT7\endcsname{\color{black}}%
      \expandafter\def\csname LT8\endcsname{\color{black}}%
    \fi
  \fi
    \setlength{\unitlength}{0.0500bp}%
    \ifx\gptboxheight\undefined%
      \newlength{\gptboxheight}%
      \newlength{\gptboxwidth}%
      \newsavebox{\gptboxtext}%
    \fi%
    \setlength{\fboxrule}{0.5pt}%
    \setlength{\fboxsep}{1pt}%
\begin{picture}(4320.00,2880.00)%
    \gplgaddtomacro\gplbacktext{%
      \csname LTb\endcsname%
      \put(688,512){\makebox(0,0)[r]{\strut{}$0$}}%
      \put(688,784){\makebox(0,0)[r]{\strut{}$500$}}%
      \put(688,1056){\makebox(0,0)[r]{\strut{}$1000$}}%
      \put(688,1328){\makebox(0,0)[r]{\strut{}$1500$}}%
      \put(688,1600){\makebox(0,0)[r]{\strut{}$2000$}}%
      \put(688,1871){\makebox(0,0)[r]{\strut{}$2500$}}%
      \put(688,2143){\makebox(0,0)[r]{\strut{}$3000$}}%
      \put(688,2415){\makebox(0,0)[r]{\strut{}$3500$}}%
      \put(688,2687){\makebox(0,0)[r]{\strut{}$4000$}}%
      \put(784,352){\makebox(0,0){\strut{}1}}%
      \put(1596,352){\makebox(0,0){\strut{}2}}%
      \put(2408,352){\makebox(0,0){\strut{}3}}%
      \put(3219,352){\makebox(0,0){\strut{}4}}%
      \put(4031,352){\makebox(0,0){\strut{}5}}%
    }%
    \gplgaddtomacro\gplfronttext{%
      \csname LTb\endcsname%
      \put(128,1599){\rotatebox{-270}{\makebox(0,0){\strut{}running time (ms)}}}%
      \put(2407,112){\makebox(0,0){\strut{}Keyword Space Size (\%)}}%
      \csname LTb\endcsname%
      \put(3296,2544){\makebox(0,0)[r]{\strut{}GNNK-BB (SUM)}}%
      \csname LTb\endcsname%
      \put(3296,2384){\makebox(0,0)[r]{\strut{}GNNK-BF (SUM)}}%
      \csname LTb\endcsname%
      \put(3296,2224){\makebox(0,0)[r]{\strut{}GNNK-BB (MAX)}}%
      \csname LTb\endcsname%
      \put(3296,2064){\makebox(0,0)[r]{\strut{}GNNK-BF (MAX)}}%
    }%
    \gplbacktext
    \put(0,0){\includegraphics{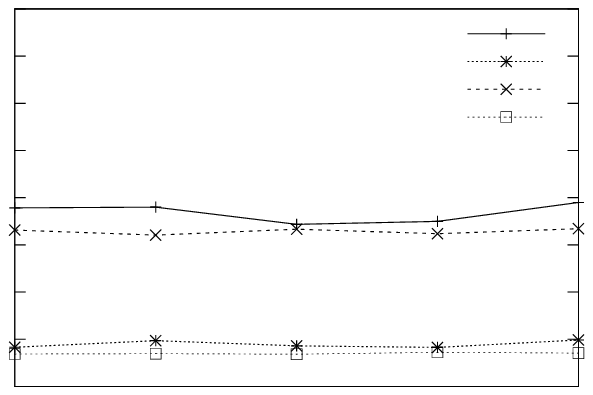}}%
    \gplfronttext
  \end{picture}%
\endgroup

		}
		\caption{}
		\label{graph:keyword-space-run}
	\end{subfigure}%
	\begin{subfigure}[b]{0.5\linewidth}
		\centering
		\setlength{\abovecaptionskip}{2pt}
		\setlength{\belowcaptionskip}{2pt}
		\resizebox{\textwidth}{!}{
\begingroup
  \makeatletter
  \providecommand\color[2][]{%
    \GenericError{(gnuplot) \space\space\space\@spaces}{%
      Package color not loaded in conjunction with
      terminal option `colourtext'%
    }{See the gnuplot documentation for explanation.%
    }{Either use 'blacktext' in gnuplot or load the package
      color.sty in LaTeX.}%
    \renewcommand\color[2][]{}%
  }%
  \providecommand\includegraphics[2][]{%
    \GenericError{(gnuplot) \space\space\space\@spaces}{%
      Package graphicx or graphics not loaded%
    }{See the gnuplot documentation for explanation.%
    }{The gnuplot epslatex terminal needs graphicx.sty or graphics.sty.}%
    \renewcommand\includegraphics[2][]{}%
  }%
  \providecommand\rotatebox[2]{#2}%
  \@ifundefined{ifGPcolor}{%
    \newif\ifGPcolor
    \GPcolorfalse
  }{}%
  \@ifundefined{ifGPblacktext}{%
    \newif\ifGPblacktext
    \GPblacktexttrue
  }{}%
  \let\gplgaddtomacro\g@addto@macro
  \gdef\gplbacktext{}%
  \gdef\gplfronttext{}%
  \makeatother
  \ifGPblacktext
    \def\colorrgb#1{}%
    \def\colorgray#1{}%
  \else
    \ifGPcolor
      \def\colorrgb#1{\color[rgb]{#1}}%
      \def\colorgray#1{\color[gray]{#1}}%
      \expandafter\def\csname LTw\endcsname{\color{white}}%
      \expandafter\def\csname LTb\endcsname{\color{black}}%
      \expandafter\def\csname LTa\endcsname{\color{black}}%
      \expandafter\def\csname LT0\endcsname{\color[rgb]{1,0,0}}%
      \expandafter\def\csname LT1\endcsname{\color[rgb]{0,1,0}}%
      \expandafter\def\csname LT2\endcsname{\color[rgb]{0,0,1}}%
      \expandafter\def\csname LT3\endcsname{\color[rgb]{1,0,1}}%
      \expandafter\def\csname LT4\endcsname{\color[rgb]{0,1,1}}%
      \expandafter\def\csname LT5\endcsname{\color[rgb]{1,1,0}}%
      \expandafter\def\csname LT6\endcsname{\color[rgb]{0,0,0}}%
      \expandafter\def\csname LT7\endcsname{\color[rgb]{1,0.3,0}}%
      \expandafter\def\csname LT8\endcsname{\color[rgb]{0.5,0.5,0.5}}%
    \else
      \def\colorrgb#1{\color{black}}%
      \def\colorgray#1{\color[gray]{#1}}%
      \expandafter\def\csname LTw\endcsname{\color{white}}%
      \expandafter\def\csname LTb\endcsname{\color{black}}%
      \expandafter\def\csname LTa\endcsname{\color{black}}%
      \expandafter\def\csname LT0\endcsname{\color{black}}%
      \expandafter\def\csname LT1\endcsname{\color{black}}%
      \expandafter\def\csname LT2\endcsname{\color{black}}%
      \expandafter\def\csname LT3\endcsname{\color{black}}%
      \expandafter\def\csname LT4\endcsname{\color{black}}%
      \expandafter\def\csname LT5\endcsname{\color{black}}%
      \expandafter\def\csname LT6\endcsname{\color{black}}%
      \expandafter\def\csname LT7\endcsname{\color{black}}%
      \expandafter\def\csname LT8\endcsname{\color{black}}%
    \fi
  \fi
    \setlength{\unitlength}{0.0500bp}%
    \ifx\gptboxheight\undefined%
      \newlength{\gptboxheight}%
      \newlength{\gptboxwidth}%
      \newsavebox{\gptboxtext}%
    \fi%
    \setlength{\fboxrule}{0.5pt}%
    \setlength{\fboxsep}{1pt}%
\begin{picture}(4320.00,2880.00)%
    \gplgaddtomacro\gplbacktext{%
      \csname LTb\endcsname%
      \put(592,512){\makebox(0,0)[r]{\strut{}$0$}}%
      \put(592,784){\makebox(0,0)[r]{\strut{}$100$}}%
      \put(592,1056){\makebox(0,0)[r]{\strut{}$200$}}%
      \put(592,1328){\makebox(0,0)[r]{\strut{}$300$}}%
      \put(592,1600){\makebox(0,0)[r]{\strut{}$400$}}%
      \put(592,1871){\makebox(0,0)[r]{\strut{}$500$}}%
      \put(592,2143){\makebox(0,0)[r]{\strut{}$600$}}%
      \put(592,2415){\makebox(0,0)[r]{\strut{}$700$}}%
      \put(592,2687){\makebox(0,0)[r]{\strut{}$800$}}%
      \put(688,352){\makebox(0,0){\strut{}1}}%
      \put(1524,352){\makebox(0,0){\strut{}2}}%
      \put(2360,352){\makebox(0,0){\strut{}3}}%
      \put(3195,352){\makebox(0,0){\strut{}4}}%
      \put(4031,352){\makebox(0,0){\strut{}5}}%
    }%
    \gplgaddtomacro\gplfronttext{%
      \csname LTb\endcsname%
      \put(128,1599){\rotatebox{-270}{\makebox(0,0){\strut{}\# page accesses}}}%
      \put(2359,112){\makebox(0,0){\strut{}Keyword Space Size (\%)}}%
      \csname LTb\endcsname%
      \put(3296,2544){\makebox(0,0)[r]{\strut{}GNNK-BB (SUM)}}%
      \csname LTb\endcsname%
      \put(3296,2384){\makebox(0,0)[r]{\strut{}GNNK-BF (SUM)}}%
      \csname LTb\endcsname%
      \put(3296,2224){\makebox(0,0)[r]{\strut{}GNNK-BB (MAX)}}%
      \csname LTb\endcsname%
      \put(3296,2064){\makebox(0,0)[r]{\strut{}GNNK-BF (MAX)}}%
    }%
    \gplbacktext
    \put(0,0){\includegraphics{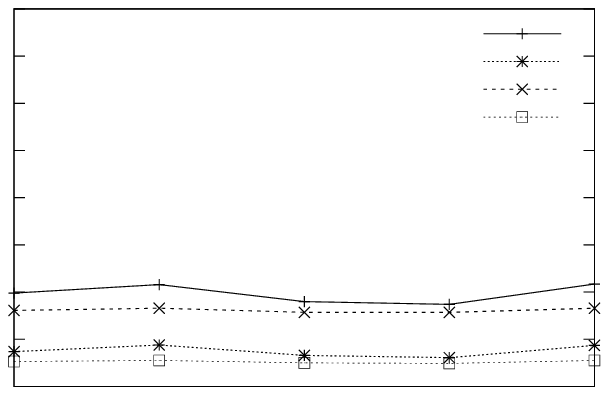}}%
    \gplfronttext
  \end{picture}%
\endgroup

		}
		\caption{}
		\label{graph:keyword-space-io}
	\end{subfigure}%

	\begin{subfigure}[b]{0.5\linewidth}
		\centering
		\setlength{\abovecaptionskip}{2pt}
		\setlength{\belowcaptionskip}{2pt}
		\resizebox{\textwidth}{!}{
\begingroup
  \makeatletter
  \providecommand\color[2][]{%
    \GenericError{(gnuplot) \space\space\space\@spaces}{%
      Package color not loaded in conjunction with
      terminal option `colourtext'%
    }{See the gnuplot documentation for explanation.%
    }{Either use 'blacktext' in gnuplot or load the package
      color.sty in LaTeX.}%
    \renewcommand\color[2][]{}%
  }%
  \providecommand\includegraphics[2][]{%
    \GenericError{(gnuplot) \space\space\space\@spaces}{%
      Package graphicx or graphics not loaded%
    }{See the gnuplot documentation for explanation.%
    }{The gnuplot epslatex terminal needs graphicx.sty or graphics.sty.}%
    \renewcommand\includegraphics[2][]{}%
  }%
  \providecommand\rotatebox[2]{#2}%
  \@ifundefined{ifGPcolor}{%
    \newif\ifGPcolor
    \GPcolorfalse
  }{}%
  \@ifundefined{ifGPblacktext}{%
    \newif\ifGPblacktext
    \GPblacktexttrue
  }{}%
  \let\gplgaddtomacro\g@addto@macro
  \gdef\gplbacktext{}%
  \gdef\gplfronttext{}%
  \makeatother
  \ifGPblacktext
    \def\colorrgb#1{}%
    \def\colorgray#1{}%
  \else
    \ifGPcolor
      \def\colorrgb#1{\color[rgb]{#1}}%
      \def\colorgray#1{\color[gray]{#1}}%
      \expandafter\def\csname LTw\endcsname{\color{white}}%
      \expandafter\def\csname LTb\endcsname{\color{black}}%
      \expandafter\def\csname LTa\endcsname{\color{black}}%
      \expandafter\def\csname LT0\endcsname{\color[rgb]{1,0,0}}%
      \expandafter\def\csname LT1\endcsname{\color[rgb]{0,1,0}}%
      \expandafter\def\csname LT2\endcsname{\color[rgb]{0,0,1}}%
      \expandafter\def\csname LT3\endcsname{\color[rgb]{1,0,1}}%
      \expandafter\def\csname LT4\endcsname{\color[rgb]{0,1,1}}%
      \expandafter\def\csname LT5\endcsname{\color[rgb]{1,1,0}}%
      \expandafter\def\csname LT6\endcsname{\color[rgb]{0,0,0}}%
      \expandafter\def\csname LT7\endcsname{\color[rgb]{1,0.3,0}}%
      \expandafter\def\csname LT8\endcsname{\color[rgb]{0.5,0.5,0.5}}%
    \else
      \def\colorrgb#1{\color{black}}%
      \def\colorgray#1{\color[gray]{#1}}%
      \expandafter\def\csname LTw\endcsname{\color{white}}%
      \expandafter\def\csname LTb\endcsname{\color{black}}%
      \expandafter\def\csname LTa\endcsname{\color{black}}%
      \expandafter\def\csname LT0\endcsname{\color{black}}%
      \expandafter\def\csname LT1\endcsname{\color{black}}%
      \expandafter\def\csname LT2\endcsname{\color{black}}%
      \expandafter\def\csname LT3\endcsname{\color{black}}%
      \expandafter\def\csname LT4\endcsname{\color{black}}%
      \expandafter\def\csname LT5\endcsname{\color{black}}%
      \expandafter\def\csname LT6\endcsname{\color{black}}%
      \expandafter\def\csname LT7\endcsname{\color{black}}%
      \expandafter\def\csname LT8\endcsname{\color{black}}%
    \fi
  \fi
  \setlength{\unitlength}{0.0500bp}%
  \begin{picture}(4320.00,2880.00)%
    \gplgaddtomacro\gplbacktext{%
      \csname LTb\endcsname%
      \put(880,512){\makebox(0,0)[r]{\strut{} 0}}%
      \put(880,875){\makebox(0,0)[r]{\strut{} 2000}}%
      \put(880,1237){\makebox(0,0)[r]{\strut{} 4000}}%
      \put(880,1600){\makebox(0,0)[r]{\strut{} 6000}}%
      \put(880,1962){\makebox(0,0)[r]{\strut{} 8000}}%
      \put(880,2325){\makebox(0,0)[r]{\strut{} 10000}}%
      \put(880,2687){\makebox(0,0)[r]{\strut{} 12000}}%
      \put(976,352){\makebox(0,0){\strut{}1M}}%
      \put(1994,352){\makebox(0,0){\strut{}1.5M}}%
      \put(3013,352){\makebox(0,0){\strut{}2M}}%
      \put(4031,352){\makebox(0,0){\strut{}2.5M}}%
      \large
      \put(128,1599){\rotatebox{-270}{\makebox(0,0){\strut{}running time (ms)}}}%
      \put(2503,112){\makebox(0,0){\strut{}Dataset Size}}%
      \normalsize
    }%
    \gplgaddtomacro\gplfronttext{%
      \csname LTb\endcsname%
      \put(3296,2544){\makebox(0,0)[r]{\strut{}GNNK-BB (SUM)}}%
      \csname LTb\endcsname%
      \put(3296,2384){\makebox(0,0)[r]{\strut{}GNNK-BF (SUM)}}%
      \csname LTb\endcsname%
      \put(3296,2224){\makebox(0,0)[r]{\strut{}GNNK-BB (MAX)}}%
      \csname LTb\endcsname%
      \put(3296,2064){\makebox(0,0)[r]{\strut{}GNNK-BF (MAX)}}%
    }%
    \gplbacktext
    \put(0,0){\includegraphics{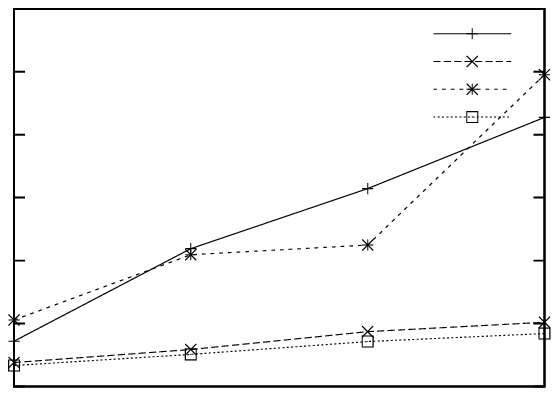}}%
    \gplfronttext
  \end{picture}%
\endgroup

		}
		\caption{}
		\label{graph:dataset-size-run}
	\end{subfigure}%
	\begin{subfigure}[b]{0.5\linewidth}
		\centering
		\setlength{\abovecaptionskip}{2pt}
		\setlength{\belowcaptionskip}{2pt}
		\resizebox{\textwidth}{!}{
\begingroup
  \makeatletter
  \providecommand\color[2][]{%
    \GenericError{(gnuplot) \space\space\space\@spaces}{%
      Package color not loaded in conjunction with
      terminal option `colourtext'%
    }{See the gnuplot documentation for explanation.%
    }{Either use 'blacktext' in gnuplot or load the package
      color.sty in LaTeX.}%
    \renewcommand\color[2][]{}%
  }%
  \providecommand\includegraphics[2][]{%
    \GenericError{(gnuplot) \space\space\space\@spaces}{%
      Package graphicx or graphics not loaded%
    }{See the gnuplot documentation for explanation.%
    }{The gnuplot epslatex terminal needs graphicx.sty or graphics.sty.}%
    \renewcommand\includegraphics[2][]{}%
  }%
  \providecommand\rotatebox[2]{#2}%
  \@ifundefined{ifGPcolor}{%
    \newif\ifGPcolor
    \GPcolorfalse
  }{}%
  \@ifundefined{ifGPblacktext}{%
    \newif\ifGPblacktext
    \GPblacktexttrue
  }{}%
  \let\gplgaddtomacro\g@addto@macro
  \gdef\gplbacktext{}%
  \gdef\gplfronttext{}%
  \makeatother
  \ifGPblacktext
    \def\colorrgb#1{}%
    \def\colorgray#1{}%
  \else
    \ifGPcolor
      \def\colorrgb#1{\color[rgb]{#1}}%
      \def\colorgray#1{\color[gray]{#1}}%
      \expandafter\def\csname LTw\endcsname{\color{white}}%
      \expandafter\def\csname LTb\endcsname{\color{black}}%
      \expandafter\def\csname LTa\endcsname{\color{black}}%
      \expandafter\def\csname LT0\endcsname{\color[rgb]{1,0,0}}%
      \expandafter\def\csname LT1\endcsname{\color[rgb]{0,1,0}}%
      \expandafter\def\csname LT2\endcsname{\color[rgb]{0,0,1}}%
      \expandafter\def\csname LT3\endcsname{\color[rgb]{1,0,1}}%
      \expandafter\def\csname LT4\endcsname{\color[rgb]{0,1,1}}%
      \expandafter\def\csname LT5\endcsname{\color[rgb]{1,1,0}}%
      \expandafter\def\csname LT6\endcsname{\color[rgb]{0,0,0}}%
      \expandafter\def\csname LT7\endcsname{\color[rgb]{1,0.3,0}}%
      \expandafter\def\csname LT8\endcsname{\color[rgb]{0.5,0.5,0.5}}%
    \else
      \def\colorrgb#1{\color{black}}%
      \def\colorgray#1{\color[gray]{#1}}%
      \expandafter\def\csname LTw\endcsname{\color{white}}%
      \expandafter\def\csname LTb\endcsname{\color{black}}%
      \expandafter\def\csname LTa\endcsname{\color{black}}%
      \expandafter\def\csname LT0\endcsname{\color{black}}%
      \expandafter\def\csname LT1\endcsname{\color{black}}%
      \expandafter\def\csname LT2\endcsname{\color{black}}%
      \expandafter\def\csname LT3\endcsname{\color{black}}%
      \expandafter\def\csname LT4\endcsname{\color{black}}%
      \expandafter\def\csname LT5\endcsname{\color{black}}%
      \expandafter\def\csname LT6\endcsname{\color{black}}%
      \expandafter\def\csname LT7\endcsname{\color{black}}%
      \expandafter\def\csname LT8\endcsname{\color{black}}%
    \fi
  \fi
  \setlength{\unitlength}{0.0500bp}%
  \begin{picture}(4320.00,2880.00)%
    \gplgaddtomacro\gplbacktext{%
      \csname LTb\endcsname%
      \put(688,512){\makebox(0,0)[r]{\strut{} 0}}%
      \put(688,823){\makebox(0,0)[r]{\strut{} 100}}%
      \put(688,1133){\makebox(0,0)[r]{\strut{} 200}}%
      \put(688,1444){\makebox(0,0)[r]{\strut{} 300}}%
      \put(688,1755){\makebox(0,0)[r]{\strut{} 400}}%
      \put(688,2066){\makebox(0,0)[r]{\strut{} 500}}%
      \put(688,2376){\makebox(0,0)[r]{\strut{} 600}}%
      \put(688,2687){\makebox(0,0)[r]{\strut{} 700}}%
      \put(784,352){\makebox(0,0){\strut{}1M}}%
      \put(1866,352){\makebox(0,0){\strut{}1.5M}}%
      \put(2949,352){\makebox(0,0){\strut{}2M}}%
      \put(4031,352){\makebox(0,0){\strut{}2.5M}}%
      \large
      \put(128,1599){\rotatebox{-270}{\makebox(0,0){\strut{}\# page accesses}}}%
      \put(2407,112){\makebox(0,0){\strut{}Dataset Size}}%
      \normalsize
    }%
    \gplgaddtomacro\gplfronttext{%
      \csname LTb\endcsname%
      \put(3296,2544){\makebox(0,0)[r]{\strut{}GNNK-BB (SUM)}}%
      \csname LTb\endcsname%
      \put(3296,2384){\makebox(0,0)[r]{\strut{}GNNK-BF (SUM)}}%
      \csname LTb\endcsname%
      \put(3296,2224){\makebox(0,0)[r]{\strut{}GNNK-BB (MAX)}}%
      \csname LTb\endcsname%
      \put(3296,2064){\makebox(0,0)[r]{\strut{}GNNK-BF (MAX)}}%
    }%
    \gplbacktext
    \put(0,0){\includegraphics{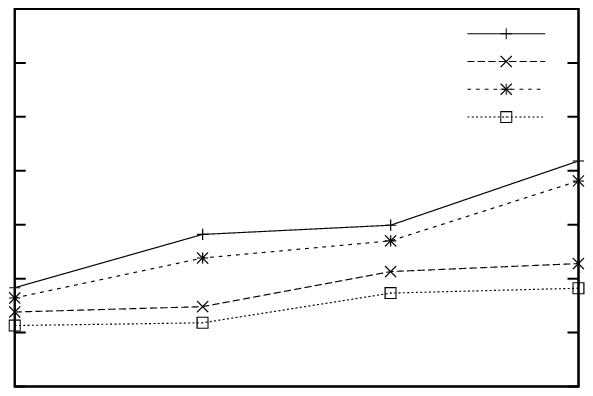}}%
    \gplfronttext
  \end{picture}%
\endgroup

		}
		\caption{}
		\label{graph:dataset-size-io}
	\end{subfigure}%
	\caption{The effect of varying $k$ (a-b), query group size (c-d), number of query keywords (e-f), query keyword set size (g-h) and dataset size (i-j) in running time and I/O}
	\label{graph:flickr}
\end {figure}

\textbf{Dataset.}
We use two real datasets from Yahoo! Flickr\footnote{\url{https://webscope.sandbox.yahoo.com}} and Yelp\footnote{\url{https://www.yelp.com/academic_dataset}} in our experiments. The Flickr dataset is generated from the images from Yahoo! Flickr users that are geo-tagged and contain a set of keyword tags. 
The Yelp dataset contains the basic information about different local businesses. Each data object contains the location of the business along with its categories as the keywords. Only the data locations within US have been used in our experiment. The properties of these two datasets are detailed in Table~\ref{table:dataset-properties}.

\textbf{Query Generation.}
We generate 20 groups of query objects for each experiment and average the results. Each query object contains a location and a set of keywords. To generate the locations in each group of query objects, we first randomly choose a point in the data space. Then we define a square query space centered at the chosen point. 
All the query object locations of the group will then be uniformly generated inside this square query space. The default query space area has been selected to be 0.01\% (~250 sq. miles) of the total query area, which is approximately the size of a medium sized US city. Similarly, for generating the query keywords, a subset of keywords (1\%-5\% of the data objects' keywords) from all keywords inside the query space is first chosen, and then the required of number of keywords are selected from this subset. This ensures the overlapping of query keywords among users.

We also vary the group size ($n$), the minimum subgroup size ($m$), the number of query keywords, the number of queried data points ($k$), dataset size, and $\alpha$. Table~\ref{table:query-parameters} shows ranges and default values of these parameters.

\textbf{Setup.}
We use the IR-tree to index the datasets, which is disk resident. The fanout of the IR-tree is chosen to be 50, and the page size is 4KB. All the algorithms are implemented in Java and the experiments are conducted on a Core i7-4790 CPU @ 3.60 GHz with 4 GB of RAM. The hard drive used is Seagate ST500DM002-1BD142 with 7200 RPM. SUM and MAX are used as the aggregate functions in all the experiments.

We measure the running time and the I/O cost (number of disk page accesses) in the experiments. Note that the running time includes the computation and I/O time.
We use Flickr as our default dataset, unless stated otherwise.

\subsection{The GNNK Query Algorithms}
We conduct seven sets of experiments to evaluate the performance of GNNK-BB and GNNK-BF. In each set of experiments, one parameter (e.g., group size $n$ or $\alpha$) is varied while all other parameters are set to their default values. GNNK-BF outperforms GNNK-BB 
in all experiments both in terms of running time and I/O cost.

\textbf{Varying $k$.}
Figure~\ref{graph:flickr} (a-b) shows that for both GNNNK-BB and GNNK-BF, the processing time and the I/O cost increase with the increase of $k$. For both SUM and MAX, on average GNNK-BF runs 3.5 times faster than GNNK-BB. We also observe that for a larger value of $k$, GNNK-BF algorithm outperforms GNN-BB in a greater margin, which shows the scalability of GNNK-BF. The I/O cost of GNNK-BF is much less than that of GNNK-BB as GNNK-BF only accesses the necessary nodes.

%

\textbf{Varying Query Group Size.}
Figure~\ref{graph:flickr} (c-d) shows the effect of the query group size ($n$). 
The query processing costs of both algorithms increase as the value of $n$ increases. On average, GNNK-BF runs approximately 4 times
faster than GNNK-BB.  

%

\textbf{Varying Number of Query Keywords.}
Figure~\ref{graph:flickr} (e-f) shows the effect of the number of keywords in each query object. 
GNNK-BF again outruns GNNK-BB in all the experiments. 
Also, the query processing costs of both algorithms increase as the number of keywords in each query object increases. This can be explained by that a larger 
set of query keywords takes more time to compute the aggregate cost function. Meanwhile, more data objects' keyword sets would overlap with 
the query keywords, which would reduce the aggregate cost function values and make it more difficult to prune the data objects. 

\textbf{Varying Query Space Size.}
We observe that the running time of our algorithms remains almost constant with the change of the query space area (not shown in graphs). Since varied query space areas are insignificant in compared to the data space, we do not observe any significant change in this experiment.

\textbf{Varying Query Keyword Set Size.}
Figure~\ref{graph:flickr} (g-h) shows the effect of the query keyword set size (the subset of keywords from where the query keywords are generated). We see that the running time of our algorithms do not follow any regular pattern with the change of the query keyword set size and remains relatively stable.

\textbf{Varying $\alpha$.}
We observe that, as $\alpha$ increases, the query costs decrease. A larger $\alpha$ means that spatial proximity is deemed more important than textual similarity. When $\alpha$ increases, the impact of the keyword similarity becomes smaller and algorithms converge faster (not shown in graphs).

\textbf{Varying Dataset Size.}
Figure~\ref{graph:flickr} (i-j) shows the effect of varying number of objects. Both running time and I/O cost of our proposed algorithms increase at a lower rate than the baseline algorithms. When the number of data objects increases from 1M to 2.5M, the running time of GNNK-BB increases 6 times for SUM and 4.7 times for MAX. But the increase in running time of GNNK-BF is only 2.7 times for SUM and 2.5 times for MAX.


\subsection{The FSNNK Query Algorithms}
We performed experiments on FSNNK-BB and FSNNK-BF, by varying query group size, subgroup size, number of query keywords, query space size, query keyword set size, $k$, dataset size, and $\alpha$. FSNNK-BF outperforms FSNNK-BB in all the experiments. For space constraints, we only show the effect of varying the subgroup size (in \% $n$) in Figure~\ref{graph:flickr-fsnnk-mfsnnk} (a-b). On average, FSNNK-BF runs 3.5 times faster and takes 40\% less I/O than FSNNK-BB.

\begin {figure}[!ht]
\setlength{\belowcaptionskip}{4pt}
\setlength{\abovecaptionskip}{4pt}
	\centering
	\begin{subfigure}[b]{0.5\linewidth}
		\centering
		\setlength{\abovecaptionskip}{2pt}
		\setlength{\belowcaptionskip}{2pt}
		\resizebox{\textwidth}{!}{
\begingroup
  \makeatletter
  \providecommand\color[2][]{%
    \GenericError{(gnuplot) \space\space\space\@spaces}{%
      Package color not loaded in conjunction with
      terminal option `colourtext'%
    }{See the gnuplot documentation for explanation.%
    }{Either use 'blacktext' in gnuplot or load the package
      color.sty in LaTeX.}%
    \renewcommand\color[2][]{}%
  }%
  \providecommand\includegraphics[2][]{%
    \GenericError{(gnuplot) \space\space\space\@spaces}{%
      Package graphicx or graphics not loaded%
    }{See the gnuplot documentation for explanation.%
    }{The gnuplot epslatex terminal needs graphicx.sty or graphics.sty.}%
    \renewcommand\includegraphics[2][]{}%
  }%
  \providecommand\rotatebox[2]{#2}%
  \@ifundefined{ifGPcolor}{%
    \newif\ifGPcolor
    \GPcolorfalse
  }{}%
  \@ifundefined{ifGPblacktext}{%
    \newif\ifGPblacktext
    \GPblacktexttrue
  }{}%
  \let\gplgaddtomacro\g@addto@macro
  \gdef\gplbacktext{}%
  \gdef\gplfronttext{}%
  \makeatother
  \ifGPblacktext
    \def\colorrgb#1{}%
    \def\colorgray#1{}%
  \else
    \ifGPcolor
      \def\colorrgb#1{\color[rgb]{#1}}%
      \def\colorgray#1{\color[gray]{#1}}%
      \expandafter\def\csname LTw\endcsname{\color{white}}%
      \expandafter\def\csname LTb\endcsname{\color{black}}%
      \expandafter\def\csname LTa\endcsname{\color{black}}%
      \expandafter\def\csname LT0\endcsname{\color[rgb]{1,0,0}}%
      \expandafter\def\csname LT1\endcsname{\color[rgb]{0,1,0}}%
      \expandafter\def\csname LT2\endcsname{\color[rgb]{0,0,1}}%
      \expandafter\def\csname LT3\endcsname{\color[rgb]{1,0,1}}%
      \expandafter\def\csname LT4\endcsname{\color[rgb]{0,1,1}}%
      \expandafter\def\csname LT5\endcsname{\color[rgb]{1,1,0}}%
      \expandafter\def\csname LT6\endcsname{\color[rgb]{0,0,0}}%
      \expandafter\def\csname LT7\endcsname{\color[rgb]{1,0.3,0}}%
      \expandafter\def\csname LT8\endcsname{\color[rgb]{0.5,0.5,0.5}}%
    \else
      \def\colorrgb#1{\color{black}}%
      \def\colorgray#1{\color[gray]{#1}}%
      \expandafter\def\csname LTw\endcsname{\color{white}}%
      \expandafter\def\csname LTb\endcsname{\color{black}}%
      \expandafter\def\csname LTa\endcsname{\color{black}}%
      \expandafter\def\csname LT0\endcsname{\color{black}}%
      \expandafter\def\csname LT1\endcsname{\color{black}}%
      \expandafter\def\csname LT2\endcsname{\color{black}}%
      \expandafter\def\csname LT3\endcsname{\color{black}}%
      \expandafter\def\csname LT4\endcsname{\color{black}}%
      \expandafter\def\csname LT5\endcsname{\color{black}}%
      \expandafter\def\csname LT6\endcsname{\color{black}}%
      \expandafter\def\csname LT7\endcsname{\color{black}}%
      \expandafter\def\csname LT8\endcsname{\color{black}}%
    \fi
  \fi
    \setlength{\unitlength}{0.0500bp}%
    \ifx\gptboxheight\undefined%
      \newlength{\gptboxheight}%
      \newlength{\gptboxwidth}%
      \newsavebox{\gptboxtext}%
    \fi%
    \setlength{\fboxrule}{0.5pt}%
    \setlength{\fboxsep}{1pt}%
\begin{picture}(4320.00,2880.00)%
    \gplgaddtomacro\gplbacktext{%
      \csname LTb\endcsname%
      \put(688,512){\makebox(0,0)[r]{\strut{}$0$}}%
      \put(688,875){\makebox(0,0)[r]{\strut{}$1000$}}%
      \put(688,1237){\makebox(0,0)[r]{\strut{}$2000$}}%
      \put(688,1600){\makebox(0,0)[r]{\strut{}$3000$}}%
      \put(688,1962){\makebox(0,0)[r]{\strut{}$4000$}}%
      \put(688,2325){\makebox(0,0)[r]{\strut{}$5000$}}%
      \put(688,2687){\makebox(0,0)[r]{\strut{}$6000$}}%
      \put(784,352){\makebox(0,0){\strut{}40}}%
      \put(1596,352){\makebox(0,0){\strut{}50}}%
      \put(2408,352){\makebox(0,0){\strut{}60}}%
      \put(3219,352){\makebox(0,0){\strut{}70}}%
      \put(4031,352){\makebox(0,0){\strut{}80}}%
    }%
    \gplgaddtomacro\gplfronttext{%
      \csname LTb\endcsname%
      \put(128,1599){\rotatebox{-270}{\makebox(0,0){\strut{}running time (ms)}}}%
      \put(2407,112){\makebox(0,0){\strut{}Subgroup Size (\%)}}%
      \csname LTb\endcsname%
      \put(3296,2544){\makebox(0,0)[r]{\strut{}FSNNK-BB (SUM)}}%
      \csname LTb\endcsname%
      \put(3296,2384){\makebox(0,0)[r]{\strut{}FSNNK-BF (SUM)}}%
      \csname LTb\endcsname%
      \put(3296,2224){\makebox(0,0)[r]{\strut{}FSNNK-BB (MAX)}}%
      \csname LTb\endcsname%
      \put(3296,2064){\makebox(0,0)[r]{\strut{}FSNNK-BF (MAX)}}%
    }%
    \gplbacktext
    \put(0,0){\includegraphics{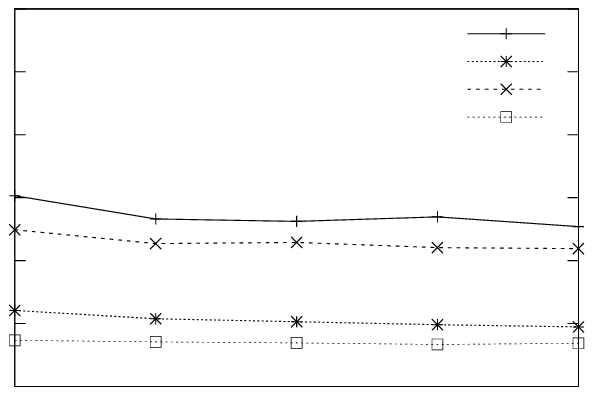}}%
    \gplfronttext
  \end{picture}%
\endgroup

		}
		\caption{}
		\label{graph:subgroup-fsnnk-run}
	\end{subfigure}%
	\begin{subfigure}[b]{0.5\linewidth}
		\centering
		\setlength{\abovecaptionskip}{2pt}
		\setlength{\belowcaptionskip}{2pt}
		\resizebox{\textwidth}{!}{
\begingroup
  \makeatletter
  \providecommand\color[2][]{%
    \GenericError{(gnuplot) \space\space\space\@spaces}{%
      Package color not loaded in conjunction with
      terminal option `colourtext'%
    }{See the gnuplot documentation for explanation.%
    }{Either use 'blacktext' in gnuplot or load the package
      color.sty in LaTeX.}%
    \renewcommand\color[2][]{}%
  }%
  \providecommand\includegraphics[2][]{%
    \GenericError{(gnuplot) \space\space\space\@spaces}{%
      Package graphicx or graphics not loaded%
    }{See the gnuplot documentation for explanation.%
    }{The gnuplot epslatex terminal needs graphicx.sty or graphics.sty.}%
    \renewcommand\includegraphics[2][]{}%
  }%
  \providecommand\rotatebox[2]{#2}%
  \@ifundefined{ifGPcolor}{%
    \newif\ifGPcolor
    \GPcolorfalse
  }{}%
  \@ifundefined{ifGPblacktext}{%
    \newif\ifGPblacktext
    \GPblacktexttrue
  }{}%
  \let\gplgaddtomacro\g@addto@macro
  \gdef\gplbacktext{}%
  \gdef\gplfronttext{}%
  \makeatother
  \ifGPblacktext
    \def\colorrgb#1{}%
    \def\colorgray#1{}%
  \else
    \ifGPcolor
      \def\colorrgb#1{\color[rgb]{#1}}%
      \def\colorgray#1{\color[gray]{#1}}%
      \expandafter\def\csname LTw\endcsname{\color{white}}%
      \expandafter\def\csname LTb\endcsname{\color{black}}%
      \expandafter\def\csname LTa\endcsname{\color{black}}%
      \expandafter\def\csname LT0\endcsname{\color[rgb]{1,0,0}}%
      \expandafter\def\csname LT1\endcsname{\color[rgb]{0,1,0}}%
      \expandafter\def\csname LT2\endcsname{\color[rgb]{0,0,1}}%
      \expandafter\def\csname LT3\endcsname{\color[rgb]{1,0,1}}%
      \expandafter\def\csname LT4\endcsname{\color[rgb]{0,1,1}}%
      \expandafter\def\csname LT5\endcsname{\color[rgb]{1,1,0}}%
      \expandafter\def\csname LT6\endcsname{\color[rgb]{0,0,0}}%
      \expandafter\def\csname LT7\endcsname{\color[rgb]{1,0.3,0}}%
      \expandafter\def\csname LT8\endcsname{\color[rgb]{0.5,0.5,0.5}}%
    \else
      \def\colorrgb#1{\color{black}}%
      \def\colorgray#1{\color[gray]{#1}}%
      \expandafter\def\csname LTw\endcsname{\color{white}}%
      \expandafter\def\csname LTb\endcsname{\color{black}}%
      \expandafter\def\csname LTa\endcsname{\color{black}}%
      \expandafter\def\csname LT0\endcsname{\color{black}}%
      \expandafter\def\csname LT1\endcsname{\color{black}}%
      \expandafter\def\csname LT2\endcsname{\color{black}}%
      \expandafter\def\csname LT3\endcsname{\color{black}}%
      \expandafter\def\csname LT4\endcsname{\color{black}}%
      \expandafter\def\csname LT5\endcsname{\color{black}}%
      \expandafter\def\csname LT6\endcsname{\color{black}}%
      \expandafter\def\csname LT7\endcsname{\color{black}}%
      \expandafter\def\csname LT8\endcsname{\color{black}}%
    \fi
  \fi
    \setlength{\unitlength}{0.0500bp}%
    \ifx\gptboxheight\undefined%
      \newlength{\gptboxheight}%
      \newlength{\gptboxwidth}%
      \newsavebox{\gptboxtext}%
    \fi%
    \setlength{\fboxrule}{0.5pt}%
    \setlength{\fboxsep}{1pt}%
\begin{picture}(4320.00,2880.00)%
    \gplgaddtomacro\gplbacktext{%
      \csname LTb\endcsname%
      \put(688,512){\makebox(0,0)[r]{\strut{}$0$}}%
      \put(688,730){\makebox(0,0)[r]{\strut{}$100$}}%
      \put(688,947){\makebox(0,0)[r]{\strut{}$200$}}%
      \put(688,1165){\makebox(0,0)[r]{\strut{}$300$}}%
      \put(688,1382){\makebox(0,0)[r]{\strut{}$400$}}%
      \put(688,1600){\makebox(0,0)[r]{\strut{}$500$}}%
      \put(688,1817){\makebox(0,0)[r]{\strut{}$600$}}%
      \put(688,2034){\makebox(0,0)[r]{\strut{}$700$}}%
      \put(688,2252){\makebox(0,0)[r]{\strut{}$800$}}%
      \put(688,2470){\makebox(0,0)[r]{\strut{}$900$}}%
      \put(688,2687){\makebox(0,0)[r]{\strut{}$1000$}}%
      \put(784,352){\makebox(0,0){\strut{}40}}%
      \put(1596,352){\makebox(0,0){\strut{}50}}%
      \put(2408,352){\makebox(0,0){\strut{}60}}%
      \put(3219,352){\makebox(0,0){\strut{}70}}%
      \put(4031,352){\makebox(0,0){\strut{}80}}%
    }%
    \gplgaddtomacro\gplfronttext{%
      \csname LTb\endcsname%
      \put(128,1599){\rotatebox{-270}{\makebox(0,0){\strut{}\# page accesses}}}%
      \put(2407,112){\makebox(0,0){\strut{}Subgroup Size (\%)}}%
      \csname LTb\endcsname%
      \put(3296,2544){\makebox(0,0)[r]{\strut{}FSNNK-BB (SUM)}}%
      \csname LTb\endcsname%
      \put(3296,2384){\makebox(0,0)[r]{\strut{}FSNNK-BF (SUM)}}%
      \csname LTb\endcsname%
      \put(3296,2224){\makebox(0,0)[r]{\strut{}FSNNK-BB (MAX)}}%
      \csname LTb\endcsname%
      \put(3296,2064){\makebox(0,0)[r]{\strut{}FSNNK-BF (MAX)}}%
    }%
    \gplbacktext
    \put(0,0){\includegraphics{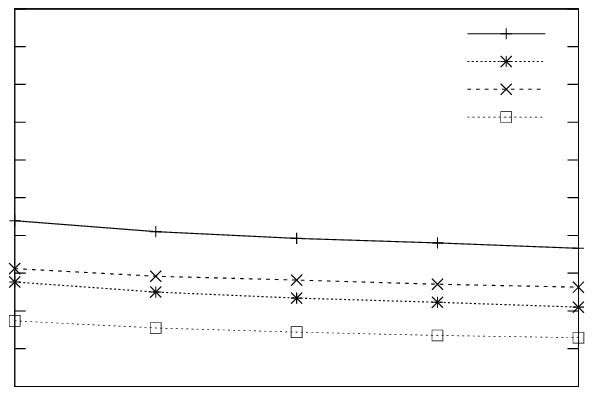}}%
    \gplfronttext
  \end{picture}%
\endgroup

		}
		\caption{}
		\label{graph:subgroup-fsnnk-io}
	\end{subfigure}%

	\begin{subfigure}[b]{0.5\linewidth}
		\centering
		\setlength{\abovecaptionskip}{2pt}
		\setlength{\belowcaptionskip}{2pt}
		\resizebox{\textwidth}{!}{
\begingroup
  \makeatletter
  \providecommand\color[2][]{%
    \GenericError{(gnuplot) \space\space\space\@spaces}{%
      Package color not loaded in conjunction with
      terminal option `colourtext'%
    }{See the gnuplot documentation for explanation.%
    }{Either use 'blacktext' in gnuplot or load the package
      color.sty in LaTeX.}%
    \renewcommand\color[2][]{}%
  }%
  \providecommand\includegraphics[2][]{%
    \GenericError{(gnuplot) \space\space\space\@spaces}{%
      Package graphicx or graphics not loaded%
    }{See the gnuplot documentation for explanation.%
    }{The gnuplot epslatex terminal needs graphicx.sty or graphics.sty.}%
    \renewcommand\includegraphics[2][]{}%
  }%
  \providecommand\rotatebox[2]{#2}%
  \@ifundefined{ifGPcolor}{%
    \newif\ifGPcolor
    \GPcolorfalse
  }{}%
  \@ifundefined{ifGPblacktext}{%
    \newif\ifGPblacktext
    \GPblacktexttrue
  }{}%
  \let\gplgaddtomacro\g@addto@macro
  \gdef\gplbacktext{}%
  \gdef\gplfronttext{}%
  \makeatother
  \ifGPblacktext
    \def\colorrgb#1{}%
    \def\colorgray#1{}%
  \else
    \ifGPcolor
      \def\colorrgb#1{\color[rgb]{#1}}%
      \def\colorgray#1{\color[gray]{#1}}%
      \expandafter\def\csname LTw\endcsname{\color{white}}%
      \expandafter\def\csname LTb\endcsname{\color{black}}%
      \expandafter\def\csname LTa\endcsname{\color{black}}%
      \expandafter\def\csname LT0\endcsname{\color[rgb]{1,0,0}}%
      \expandafter\def\csname LT1\endcsname{\color[rgb]{0,1,0}}%
      \expandafter\def\csname LT2\endcsname{\color[rgb]{0,0,1}}%
      \expandafter\def\csname LT3\endcsname{\color[rgb]{1,0,1}}%
      \expandafter\def\csname LT4\endcsname{\color[rgb]{0,1,1}}%
      \expandafter\def\csname LT5\endcsname{\color[rgb]{1,1,0}}%
      \expandafter\def\csname LT6\endcsname{\color[rgb]{0,0,0}}%
      \expandafter\def\csname LT7\endcsname{\color[rgb]{1,0.3,0}}%
      \expandafter\def\csname LT8\endcsname{\color[rgb]{0.5,0.5,0.5}}%
    \else
      \def\colorrgb#1{\color{black}}%
      \def\colorgray#1{\color[gray]{#1}}%
      \expandafter\def\csname LTw\endcsname{\color{white}}%
      \expandafter\def\csname LTb\endcsname{\color{black}}%
      \expandafter\def\csname LTa\endcsname{\color{black}}%
      \expandafter\def\csname LT0\endcsname{\color{black}}%
      \expandafter\def\csname LT1\endcsname{\color{black}}%
      \expandafter\def\csname LT2\endcsname{\color{black}}%
      \expandafter\def\csname LT3\endcsname{\color{black}}%
      \expandafter\def\csname LT4\endcsname{\color{black}}%
      \expandafter\def\csname LT5\endcsname{\color{black}}%
      \expandafter\def\csname LT6\endcsname{\color{black}}%
      \expandafter\def\csname LT7\endcsname{\color{black}}%
      \expandafter\def\csname LT8\endcsname{\color{black}}%
    \fi
  \fi
    \setlength{\unitlength}{0.0500bp}%
    \ifx\gptboxheight\undefined%
      \newlength{\gptboxheight}%
      \newlength{\gptboxwidth}%
      \newsavebox{\gptboxtext}%
    \fi%
    \setlength{\fboxrule}{0.5pt}%
    \setlength{\fboxsep}{1pt}%
\begin{picture}(4320.00,2880.00)%
    \gplgaddtomacro\gplbacktext{%
      \csname LTb\endcsname%
      \put(688,512){\makebox(0,0)[r]{\strut{}$0$}}%
      \put(688,875){\makebox(0,0)[r]{\strut{}$1000$}}%
      \put(688,1237){\makebox(0,0)[r]{\strut{}$2000$}}%
      \put(688,1600){\makebox(0,0)[r]{\strut{}$3000$}}%
      \put(688,1962){\makebox(0,0)[r]{\strut{}$4000$}}%
      \put(688,2325){\makebox(0,0)[r]{\strut{}$5000$}}%
      \put(688,2687){\makebox(0,0)[r]{\strut{}$6000$}}%
      \put(784,352){\makebox(0,0){\strut{}40}}%
      \put(1596,352){\makebox(0,0){\strut{}50}}%
      \put(2408,352){\makebox(0,0){\strut{}60}}%
      \put(3219,352){\makebox(0,0){\strut{}70}}%
      \put(4031,352){\makebox(0,0){\strut{}80}}%
    }%
    \gplgaddtomacro\gplfronttext{%
      \csname LTb\endcsname%
      \put(128,1599){\rotatebox{-270}{\makebox(0,0){\strut{}running time (ms)}}}%
      \put(2407,112){\makebox(0,0){\strut{}Subgroup Size (\%)}}%
      \csname LTb\endcsname%
      \put(3296,2544){\makebox(0,0)[r]{\strut{}MFSNNK-N (SUM)}}%
      \csname LTb\endcsname%
      \put(3296,2384){\makebox(0,0)[r]{\strut{}MFSNNK-BF (SUM)}}%
      \csname LTb\endcsname%
      \put(3296,2224){\makebox(0,0)[r]{\strut{}MFSNNK-N (MAX)}}%
      \csname LTb\endcsname%
      \put(3296,2064){\makebox(0,0)[r]{\strut{}MFSNNK-BF (MAX)}}%
    }%
    \gplbacktext
    \put(0,0){\includegraphics{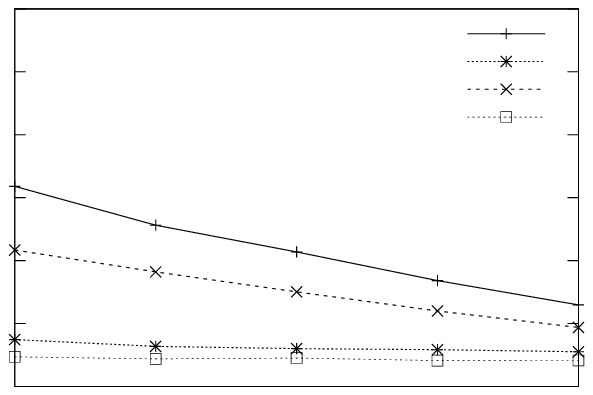}}%
    \gplfronttext
  \end{picture}%
\endgroup

		}
		\caption{}
		\label{graph:subgroup-mfsnnk-run}
	\end{subfigure}%
	\begin{subfigure}[b]{0.5\linewidth}
		\centering
		\setlength{\abovecaptionskip}{2pt}
		\setlength{\belowcaptionskip}{2pt}
		\resizebox{\textwidth}{!}{
\begingroup
  \makeatletter
  \providecommand\color[2][]{%
    \GenericError{(gnuplot) \space\space\space\@spaces}{%
      Package color not loaded in conjunction with
      terminal option `colourtext'%
    }{See the gnuplot documentation for explanation.%
    }{Either use 'blacktext' in gnuplot or load the package
      color.sty in LaTeX.}%
    \renewcommand\color[2][]{}%
  }%
  \providecommand\includegraphics[2][]{%
    \GenericError{(gnuplot) \space\space\space\@spaces}{%
      Package graphicx or graphics not loaded%
    }{See the gnuplot documentation for explanation.%
    }{The gnuplot epslatex terminal needs graphicx.sty or graphics.sty.}%
    \renewcommand\includegraphics[2][]{}%
  }%
  \providecommand\rotatebox[2]{#2}%
  \@ifundefined{ifGPcolor}{%
    \newif\ifGPcolor
    \GPcolorfalse
  }{}%
  \@ifundefined{ifGPblacktext}{%
    \newif\ifGPblacktext
    \GPblacktexttrue
  }{}%
  \let\gplgaddtomacro\g@addto@macro
  \gdef\gplbacktext{}%
  \gdef\gplfronttext{}%
  \makeatother
  \ifGPblacktext
    \def\colorrgb#1{}%
    \def\colorgray#1{}%
  \else
    \ifGPcolor
      \def\colorrgb#1{\color[rgb]{#1}}%
      \def\colorgray#1{\color[gray]{#1}}%
      \expandafter\def\csname LTw\endcsname{\color{white}}%
      \expandafter\def\csname LTb\endcsname{\color{black}}%
      \expandafter\def\csname LTa\endcsname{\color{black}}%
      \expandafter\def\csname LT0\endcsname{\color[rgb]{1,0,0}}%
      \expandafter\def\csname LT1\endcsname{\color[rgb]{0,1,0}}%
      \expandafter\def\csname LT2\endcsname{\color[rgb]{0,0,1}}%
      \expandafter\def\csname LT3\endcsname{\color[rgb]{1,0,1}}%
      \expandafter\def\csname LT4\endcsname{\color[rgb]{0,1,1}}%
      \expandafter\def\csname LT5\endcsname{\color[rgb]{1,1,0}}%
      \expandafter\def\csname LT6\endcsname{\color[rgb]{0,0,0}}%
      \expandafter\def\csname LT7\endcsname{\color[rgb]{1,0.3,0}}%
      \expandafter\def\csname LT8\endcsname{\color[rgb]{0.5,0.5,0.5}}%
    \else
      \def\colorrgb#1{\color{black}}%
      \def\colorgray#1{\color[gray]{#1}}%
      \expandafter\def\csname LTw\endcsname{\color{white}}%
      \expandafter\def\csname LTb\endcsname{\color{black}}%
      \expandafter\def\csname LTa\endcsname{\color{black}}%
      \expandafter\def\csname LT0\endcsname{\color{black}}%
      \expandafter\def\csname LT1\endcsname{\color{black}}%
      \expandafter\def\csname LT2\endcsname{\color{black}}%
      \expandafter\def\csname LT3\endcsname{\color{black}}%
      \expandafter\def\csname LT4\endcsname{\color{black}}%
      \expandafter\def\csname LT5\endcsname{\color{black}}%
      \expandafter\def\csname LT6\endcsname{\color{black}}%
      \expandafter\def\csname LT7\endcsname{\color{black}}%
      \expandafter\def\csname LT8\endcsname{\color{black}}%
    \fi
  \fi
    \setlength{\unitlength}{0.0500bp}%
    \ifx\gptboxheight\undefined%
      \newlength{\gptboxheight}%
      \newlength{\gptboxwidth}%
      \newsavebox{\gptboxtext}%
    \fi%
    \setlength{\fboxrule}{0.5pt}%
    \setlength{\fboxsep}{1pt}%
\begin{picture}(4320.00,2880.00)%
    \gplgaddtomacro\gplbacktext{%
      \csname LTb\endcsname%
      \put(688,512){\makebox(0,0)[r]{\strut{}$0$}}%
      \put(688,823){\makebox(0,0)[r]{\strut{}$200$}}%
      \put(688,1133){\makebox(0,0)[r]{\strut{}$400$}}%
      \put(688,1444){\makebox(0,0)[r]{\strut{}$600$}}%
      \put(688,1755){\makebox(0,0)[r]{\strut{}$800$}}%
      \put(688,2066){\makebox(0,0)[r]{\strut{}$1000$}}%
      \put(688,2376){\makebox(0,0)[r]{\strut{}$1200$}}%
      \put(688,2687){\makebox(0,0)[r]{\strut{}$1400$}}%
      \put(784,352){\makebox(0,0){\strut{}40}}%
      \put(1596,352){\makebox(0,0){\strut{}50}}%
      \put(2408,352){\makebox(0,0){\strut{}60}}%
      \put(3219,352){\makebox(0,0){\strut{}70}}%
      \put(4031,352){\makebox(0,0){\strut{}80}}%
    }%
    \gplgaddtomacro\gplfronttext{%
      \csname LTb\endcsname%
      \put(128,1599){\rotatebox{-270}{\makebox(0,0){\strut{}\# page accesses}}}%
      \put(2407,112){\makebox(0,0){\strut{}Subgroup Size (\%)}}%
      \csname LTb\endcsname%
      \put(3296,2544){\makebox(0,0)[r]{\strut{}MFSNNK-N (SUM)}}%
      \csname LTb\endcsname%
      \put(3296,2384){\makebox(0,0)[r]{\strut{}MFSNNK-BF (SUM)}}%
      \csname LTb\endcsname%
      \put(3296,2224){\makebox(0,0)[r]{\strut{}MFSNNK-N (MAX)}}%
      \csname LTb\endcsname%
      \put(3296,2064){\makebox(0,0)[r]{\strut{}MFSNNK-BF (MAX)}}%
    }%
    \gplbacktext
    \put(0,0){\includegraphics{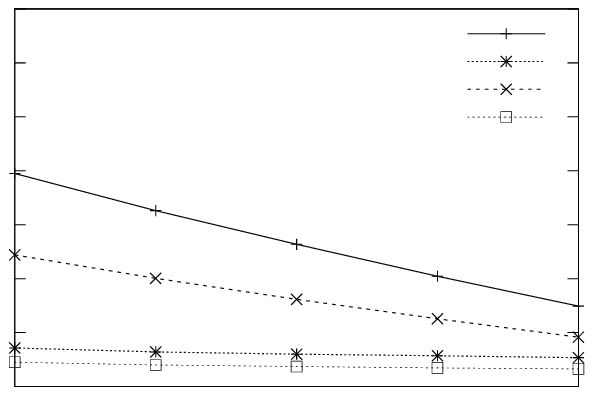}}%
    \gplfronttext
  \end{picture}%
\endgroup

		}
		\caption{}
		\label{graph:subgroup-mfsnnk-io}
	\end{subfigure}%
	\caption{The effect of varying subgroup size $m$ (a-b) and minimum subgroup size (c-d) in running time and I/O}
	\label{graph:flickr-fsnnk-mfsnnk}
\end {figure}

\subsection{The MFSNNK Query Algorithms}
We performed similar experiments on MFSNNK-N and MFSNNK-BF. 
In all the experiments MFSNNK-BF significantly outperforms MFSNNK-N. Due to space constraints, we only show the effect of varying the minimum subgroup size (in percentage of $n$) in Figure~\ref{graph:flickr-fsnnk-mfsnnk} (c-d).
When the minimum subgroup size increases, the running time of both algorithms decrease as expected. Meanwhile, the costs of MFSNNK-BF change in a much smaller scale, which demonstrates the better scalability of MFSNNK-BF. On average, MFSNNK-BF runs about 4 times faster than MFSNNK-N.


%


\begin {figure}[!ht]
\setlength{\belowcaptionskip}{4pt}
\setlength{\abovecaptionskip}{4pt}
	\centering
	\begin{subfigure}[b]{0.5\linewidth}
		\centering
		\setlength{\abovecaptionskip}{2pt}
		\setlength{\belowcaptionskip}{2pt}
		\resizebox{\textwidth}{!}{
\begingroup
  \makeatletter
  \providecommand\color[2][]{%
    \GenericError{(gnuplot) \space\space\space\@spaces}{%
      Package color not loaded in conjunction with
      terminal option `colourtext'%
    }{See the gnuplot documentation for explanation.%
    }{Either use 'blacktext' in gnuplot or load the package
      color.sty in LaTeX.}%
    \renewcommand\color[2][]{}%
  }%
  \providecommand\includegraphics[2][]{%
    \GenericError{(gnuplot) \space\space\space\@spaces}{%
      Package graphicx or graphics not loaded%
    }{See the gnuplot documentation for explanation.%
    }{The gnuplot epslatex terminal needs graphicx.sty or graphics.sty.}%
    \renewcommand\includegraphics[2][]{}%
  }%
  \providecommand\rotatebox[2]{#2}%
  \@ifundefined{ifGPcolor}{%
    \newif\ifGPcolor
    \GPcolorfalse
  }{}%
  \@ifundefined{ifGPblacktext}{%
    \newif\ifGPblacktext
    \GPblacktexttrue
  }{}%
  \let\gplgaddtomacro\g@addto@macro
  \gdef\gplbacktext{}%
  \gdef\gplfronttext{}%
  \makeatother
  \ifGPblacktext
    \def\colorrgb#1{}%
    \def\colorgray#1{}%
  \else
    \ifGPcolor
      \def\colorrgb#1{\color[rgb]{#1}}%
      \def\colorgray#1{\color[gray]{#1}}%
      \expandafter\def\csname LTw\endcsname{\color{white}}%
      \expandafter\def\csname LTb\endcsname{\color{black}}%
      \expandafter\def\csname LTa\endcsname{\color{black}}%
      \expandafter\def\csname LT0\endcsname{\color[rgb]{1,0,0}}%
      \expandafter\def\csname LT1\endcsname{\color[rgb]{0,1,0}}%
      \expandafter\def\csname LT2\endcsname{\color[rgb]{0,0,1}}%
      \expandafter\def\csname LT3\endcsname{\color[rgb]{1,0,1}}%
      \expandafter\def\csname LT4\endcsname{\color[rgb]{0,1,1}}%
      \expandafter\def\csname LT5\endcsname{\color[rgb]{1,1,0}}%
      \expandafter\def\csname LT6\endcsname{\color[rgb]{0,0,0}}%
      \expandafter\def\csname LT7\endcsname{\color[rgb]{1,0.3,0}}%
      \expandafter\def\csname LT8\endcsname{\color[rgb]{0.5,0.5,0.5}}%
    \else
      \def\colorrgb#1{\color{black}}%
      \def\colorgray#1{\color[gray]{#1}}%
      \expandafter\def\csname LTw\endcsname{\color{white}}%
      \expandafter\def\csname LTb\endcsname{\color{black}}%
      \expandafter\def\csname LTa\endcsname{\color{black}}%
      \expandafter\def\csname LT0\endcsname{\color{black}}%
      \expandafter\def\csname LT1\endcsname{\color{black}}%
      \expandafter\def\csname LT2\endcsname{\color{black}}%
      \expandafter\def\csname LT3\endcsname{\color{black}}%
      \expandafter\def\csname LT4\endcsname{\color{black}}%
      \expandafter\def\csname LT5\endcsname{\color{black}}%
      \expandafter\def\csname LT6\endcsname{\color{black}}%
      \expandafter\def\csname LT7\endcsname{\color{black}}%
      \expandafter\def\csname LT8\endcsname{\color{black}}%
    \fi
  \fi
    \setlength{\unitlength}{0.0500bp}%
    \ifx\gptboxheight\undefined%
      \newlength{\gptboxheight}%
      \newlength{\gptboxwidth}%
      \newsavebox{\gptboxtext}%
    \fi%
    \setlength{\fboxrule}{0.5pt}%
    \setlength{\fboxsep}{1pt}%
\begin{picture}(4320.00,2880.00)%
    \gplgaddtomacro\gplbacktext{%
      \csname LTb\endcsname%
      \put(688,512){\makebox(0,0)[r]{\strut{}$0$}}%
      \put(688,875){\makebox(0,0)[r]{\strut{}$500$}}%
      \put(688,1237){\makebox(0,0)[r]{\strut{}$1000$}}%
      \put(688,1600){\makebox(0,0)[r]{\strut{}$1500$}}%
      \put(688,1962){\makebox(0,0)[r]{\strut{}$2000$}}%
      \put(688,2325){\makebox(0,0)[r]{\strut{}$2500$}}%
      \put(688,2687){\makebox(0,0)[r]{\strut{}$3000$}}%
      \put(784,352){\makebox(0,0){\strut{}40}}%
      \put(1596,352){\makebox(0,0){\strut{}50}}%
      \put(2408,352){\makebox(0,0){\strut{}60}}%
      \put(3219,352){\makebox(0,0){\strut{}70}}%
      \put(4031,352){\makebox(0,0){\strut{}80}}%
    }%
    \gplgaddtomacro\gplfronttext{%
      \csname LTb\endcsname%
      \put(128,1599){\rotatebox{-270}{\makebox(0,0){\strut{}running time (ms)}}}%
      \put(2407,112){\makebox(0,0){\strut{}Subgroup Size (\%)}}%
      \csname LTb\endcsname%
      \put(3296,2544){\makebox(0,0)[r]{\strut{}MFSNNK-BF (SUM)}}%
      \csname LTb\endcsname%
      \put(3296,2384){\makebox(0,0)[r]{\strut{}MFSNNK-R (SUM)}}%
      \csname LTb\endcsname%
      \put(3296,2224){\makebox(0,0)[r]{\strut{}MFSNNK-BF (MAX)}}%
      \csname LTb\endcsname%
      \put(3296,2064){\makebox(0,0)[r]{\strut{}MFSNNK-R (MAX)}}%
    }%
    \gplbacktext
    \put(0,0){\includegraphics{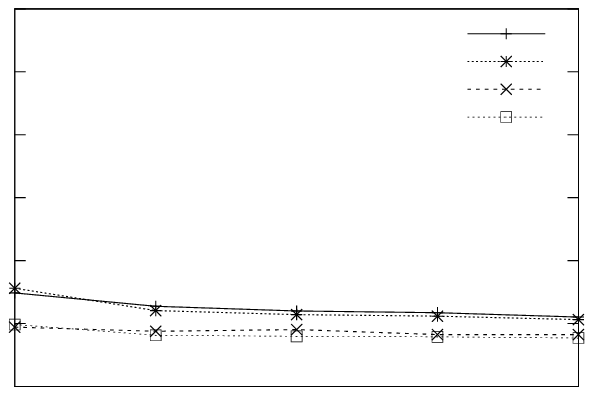}}%
    \gplfronttext
  \end{picture}%
\endgroup

		}
		\caption{}
		\label{graph:subgroup-mfsnnkr-F-run}
	\end{subfigure}%
	\begin{subfigure}[b]{0.5\linewidth}
		\centering
		\setlength{\abovecaptionskip}{2pt}
		\setlength{\belowcaptionskip}{2pt}
		\resizebox{\textwidth}{!}{
\begingroup
  \makeatletter
  \providecommand\color[2][]{%
    \GenericError{(gnuplot) \space\space\space\@spaces}{%
      Package color not loaded in conjunction with
      terminal option `colourtext'%
    }{See the gnuplot documentation for explanation.%
    }{Either use 'blacktext' in gnuplot or load the package
      color.sty in LaTeX.}%
    \renewcommand\color[2][]{}%
  }%
  \providecommand\includegraphics[2][]{%
    \GenericError{(gnuplot) \space\space\space\@spaces}{%
      Package graphicx or graphics not loaded%
    }{See the gnuplot documentation for explanation.%
    }{The gnuplot epslatex terminal needs graphicx.sty or graphics.sty.}%
    \renewcommand\includegraphics[2][]{}%
  }%
  \providecommand\rotatebox[2]{#2}%
  \@ifundefined{ifGPcolor}{%
    \newif\ifGPcolor
    \GPcolorfalse
  }{}%
  \@ifundefined{ifGPblacktext}{%
    \newif\ifGPblacktext
    \GPblacktexttrue
  }{}%
  \let\gplgaddtomacro\g@addto@macro
  \gdef\gplbacktext{}%
  \gdef\gplfronttext{}%
  \makeatother
  \ifGPblacktext
    \def\colorrgb#1{}%
    \def\colorgray#1{}%
  \else
    \ifGPcolor
      \def\colorrgb#1{\color[rgb]{#1}}%
      \def\colorgray#1{\color[gray]{#1}}%
      \expandafter\def\csname LTw\endcsname{\color{white}}%
      \expandafter\def\csname LTb\endcsname{\color{black}}%
      \expandafter\def\csname LTa\endcsname{\color{black}}%
      \expandafter\def\csname LT0\endcsname{\color[rgb]{1,0,0}}%
      \expandafter\def\csname LT1\endcsname{\color[rgb]{0,1,0}}%
      \expandafter\def\csname LT2\endcsname{\color[rgb]{0,0,1}}%
      \expandafter\def\csname LT3\endcsname{\color[rgb]{1,0,1}}%
      \expandafter\def\csname LT4\endcsname{\color[rgb]{0,1,1}}%
      \expandafter\def\csname LT5\endcsname{\color[rgb]{1,1,0}}%
      \expandafter\def\csname LT6\endcsname{\color[rgb]{0,0,0}}%
      \expandafter\def\csname LT7\endcsname{\color[rgb]{1,0.3,0}}%
      \expandafter\def\csname LT8\endcsname{\color[rgb]{0.5,0.5,0.5}}%
    \else
      \def\colorrgb#1{\color{black}}%
      \def\colorgray#1{\color[gray]{#1}}%
      \expandafter\def\csname LTw\endcsname{\color{white}}%
      \expandafter\def\csname LTb\endcsname{\color{black}}%
      \expandafter\def\csname LTa\endcsname{\color{black}}%
      \expandafter\def\csname LT0\endcsname{\color{black}}%
      \expandafter\def\csname LT1\endcsname{\color{black}}%
      \expandafter\def\csname LT2\endcsname{\color{black}}%
      \expandafter\def\csname LT3\endcsname{\color{black}}%
      \expandafter\def\csname LT4\endcsname{\color{black}}%
      \expandafter\def\csname LT5\endcsname{\color{black}}%
      \expandafter\def\csname LT6\endcsname{\color{black}}%
      \expandafter\def\csname LT7\endcsname{\color{black}}%
      \expandafter\def\csname LT8\endcsname{\color{black}}%
    \fi
  \fi
    \setlength{\unitlength}{0.0500bp}%
    \ifx\gptboxheight\undefined%
      \newlength{\gptboxheight}%
      \newlength{\gptboxwidth}%
      \newsavebox{\gptboxtext}%
    \fi%
    \setlength{\fboxrule}{0.5pt}%
    \setlength{\fboxsep}{1pt}%
\begin{picture}(4320.00,2880.00)%
    \gplgaddtomacro\gplbacktext{%
      \csname LTb\endcsname%
      \put(592,512){\makebox(0,0)[r]{\strut{}$0$}}%
      \put(592,823){\makebox(0,0)[r]{\strut{}$100$}}%
      \put(592,1133){\makebox(0,0)[r]{\strut{}$200$}}%
      \put(592,1444){\makebox(0,0)[r]{\strut{}$300$}}%
      \put(592,1755){\makebox(0,0)[r]{\strut{}$400$}}%
      \put(592,2066){\makebox(0,0)[r]{\strut{}$500$}}%
      \put(592,2376){\makebox(0,0)[r]{\strut{}$600$}}%
      \put(592,2687){\makebox(0,0)[r]{\strut{}$700$}}%
      \put(688,352){\makebox(0,0){\strut{}40}}%
      \put(1524,352){\makebox(0,0){\strut{}50}}%
      \put(2360,352){\makebox(0,0){\strut{}60}}%
      \put(3195,352){\makebox(0,0){\strut{}70}}%
      \put(4031,352){\makebox(0,0){\strut{}80}}%
    }%
    \gplgaddtomacro\gplfronttext{%
      \csname LTb\endcsname%
      \put(128,1599){\rotatebox{-270}{\makebox(0,0){\strut{}\# page accesses}}}%
      \put(2359,112){\makebox(0,0){\strut{}Subgroup Size (\%)}}%
      \csname LTb\endcsname%
      \put(3296,2544){\makebox(0,0)[r]{\strut{}MFSNNK-BF (SUM)}}%
      \csname LTb\endcsname%
      \put(3296,2384){\makebox(0,0)[r]{\strut{}MFSNNK-R (SUM)}}%
      \csname LTb\endcsname%
      \put(3296,2224){\makebox(0,0)[r]{\strut{}MFSNNK-BF (MAX)}}%
      \csname LTb\endcsname%
      \put(3296,2064){\makebox(0,0)[r]{\strut{}MFSNNK-R (MAX)}}%
    }%
    \gplbacktext
    \put(0,0){\includegraphics{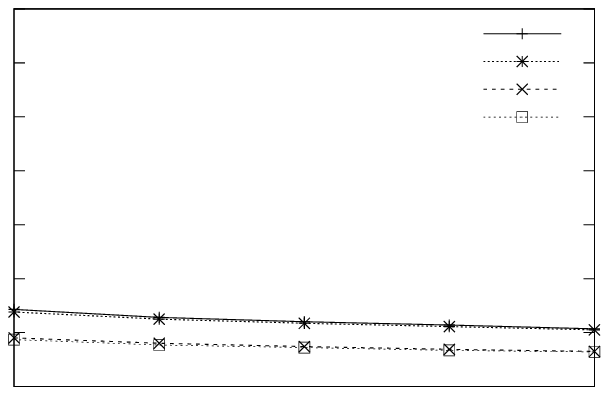}}%
    \gplfronttext
  \end{picture}%
\endgroup

		}
		\caption{}
		\label{graph:subgroup-mfsnnkr-F-io}
	\end{subfigure}%

	\begin{subfigure}[b]{0.5\linewidth}
		\centering
		\setlength{\abovecaptionskip}{2pt}
		\setlength{\belowcaptionskip}{2pt}
		\resizebox{\textwidth}{!}{
\begingroup
  \makeatletter
  \providecommand\color[2][]{%
    \GenericError{(gnuplot) \space\space\space\@spaces}{%
      Package color not loaded in conjunction with
      terminal option `colourtext'%
    }{See the gnuplot documentation for explanation.%
    }{Either use 'blacktext' in gnuplot or load the package
      color.sty in LaTeX.}%
    \renewcommand\color[2][]{}%
  }%
  \providecommand\includegraphics[2][]{%
    \GenericError{(gnuplot) \space\space\space\@spaces}{%
      Package graphicx or graphics not loaded%
    }{See the gnuplot documentation for explanation.%
    }{The gnuplot epslatex terminal needs graphicx.sty or graphics.sty.}%
    \renewcommand\includegraphics[2][]{}%
  }%
  \providecommand\rotatebox[2]{#2}%
  \@ifundefined{ifGPcolor}{%
    \newif\ifGPcolor
    \GPcolorfalse
  }{}%
  \@ifundefined{ifGPblacktext}{%
    \newif\ifGPblacktext
    \GPblacktexttrue
  }{}%
  \let\gplgaddtomacro\g@addto@macro
  \gdef\gplbacktext{}%
  \gdef\gplfronttext{}%
  \makeatother
  \ifGPblacktext
    \def\colorrgb#1{}%
    \def\colorgray#1{}%
  \else
    \ifGPcolor
      \def\colorrgb#1{\color[rgb]{#1}}%
      \def\colorgray#1{\color[gray]{#1}}%
      \expandafter\def\csname LTw\endcsname{\color{white}}%
      \expandafter\def\csname LTb\endcsname{\color{black}}%
      \expandafter\def\csname LTa\endcsname{\color{black}}%
      \expandafter\def\csname LT0\endcsname{\color[rgb]{1,0,0}}%
      \expandafter\def\csname LT1\endcsname{\color[rgb]{0,1,0}}%
      \expandafter\def\csname LT2\endcsname{\color[rgb]{0,0,1}}%
      \expandafter\def\csname LT3\endcsname{\color[rgb]{1,0,1}}%
      \expandafter\def\csname LT4\endcsname{\color[rgb]{0,1,1}}%
      \expandafter\def\csname LT5\endcsname{\color[rgb]{1,1,0}}%
      \expandafter\def\csname LT6\endcsname{\color[rgb]{0,0,0}}%
      \expandafter\def\csname LT7\endcsname{\color[rgb]{1,0.3,0}}%
      \expandafter\def\csname LT8\endcsname{\color[rgb]{0.5,0.5,0.5}}%
    \else
      \def\colorrgb#1{\color{black}}%
      \def\colorgray#1{\color[gray]{#1}}%
      \expandafter\def\csname LTw\endcsname{\color{white}}%
      \expandafter\def\csname LTb\endcsname{\color{black}}%
      \expandafter\def\csname LTa\endcsname{\color{black}}%
      \expandafter\def\csname LT0\endcsname{\color{black}}%
      \expandafter\def\csname LT1\endcsname{\color{black}}%
      \expandafter\def\csname LT2\endcsname{\color{black}}%
      \expandafter\def\csname LT3\endcsname{\color{black}}%
      \expandafter\def\csname LT4\endcsname{\color{black}}%
      \expandafter\def\csname LT5\endcsname{\color{black}}%
      \expandafter\def\csname LT6\endcsname{\color{black}}%
      \expandafter\def\csname LT7\endcsname{\color{black}}%
      \expandafter\def\csname LT8\endcsname{\color{black}}%
    \fi
  \fi
    \setlength{\unitlength}{0.0500bp}%
    \ifx\gptboxheight\undefined%
      \newlength{\gptboxheight}%
      \newlength{\gptboxwidth}%
      \newsavebox{\gptboxtext}%
    \fi%
    \setlength{\fboxrule}{0.5pt}%
    \setlength{\fboxsep}{1pt}%
\begin{picture}(4320.00,2880.00)%
    \gplgaddtomacro\gplbacktext{%
      \csname LTb\endcsname%
      \put(592,512){\makebox(0,0)[r]{\strut{}$0$}}%
      \put(592,823){\makebox(0,0)[r]{\strut{}$20$}}%
      \put(592,1133){\makebox(0,0)[r]{\strut{}$40$}}%
      \put(592,1444){\makebox(0,0)[r]{\strut{}$60$}}%
      \put(592,1755){\makebox(0,0)[r]{\strut{}$80$}}%
      \put(592,2066){\makebox(0,0)[r]{\strut{}$100$}}%
      \put(592,2376){\makebox(0,0)[r]{\strut{}$120$}}%
      \put(592,2687){\makebox(0,0)[r]{\strut{}$140$}}%
      \put(688,352){\makebox(0,0){\strut{}40}}%
      \put(1524,352){\makebox(0,0){\strut{}50}}%
      \put(2360,352){\makebox(0,0){\strut{}60}}%
      \put(3195,352){\makebox(0,0){\strut{}70}}%
      \put(4031,352){\makebox(0,0){\strut{}80}}%
    }%
    \gplgaddtomacro\gplfronttext{%
      \csname LTb\endcsname%
      \put(128,1599){\rotatebox{-270}{\makebox(0,0){\strut{}running time (ms)}}}%
      \put(2359,112){\makebox(0,0){\strut{}Subgroup Size (\%)}}%
      \csname LTb\endcsname%
      \put(3296,2544){\makebox(0,0)[r]{\strut{}MFSNNK-BF (SUM)}}%
      \csname LTb\endcsname%
      \put(3296,2384){\makebox(0,0)[r]{\strut{}MFSNNK-R (SUM)}}%
      \csname LTb\endcsname%
      \put(3296,2224){\makebox(0,0)[r]{\strut{}MFSNNK-BF (MAX)}}%
      \csname LTb\endcsname%
      \put(3296,2064){\makebox(0,0)[r]{\strut{}MFSNNK-R (MAX)}}%
    }%
    \gplbacktext
    \put(0,0){\includegraphics{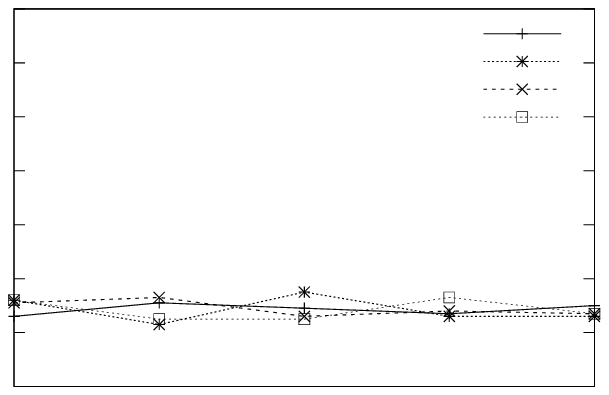}}%
    \gplfronttext
  \end{picture}%
\endgroup

		}
		\caption{}
		\label{graph:subgroup-mfsnnkr-Y-run}
	\end{subfigure}%
	\begin{subfigure}[b]{0.5\linewidth}
		\centering
		\setlength{\abovecaptionskip}{2pt}
		\setlength{\belowcaptionskip}{2pt}
		\resizebox{\textwidth}{!}{
\begingroup
  \makeatletter
  \providecommand\color[2][]{%
    \GenericError{(gnuplot) \space\space\space\@spaces}{%
      Package color not loaded in conjunction with
      terminal option `colourtext'%
    }{See the gnuplot documentation for explanation.%
    }{Either use 'blacktext' in gnuplot or load the package
      color.sty in LaTeX.}%
    \renewcommand\color[2][]{}%
  }%
  \providecommand\includegraphics[2][]{%
    \GenericError{(gnuplot) \space\space\space\@spaces}{%
      Package graphicx or graphics not loaded%
    }{See the gnuplot documentation for explanation.%
    }{The gnuplot epslatex terminal needs graphicx.sty or graphics.sty.}%
    \renewcommand\includegraphics[2][]{}%
  }%
  \providecommand\rotatebox[2]{#2}%
  \@ifundefined{ifGPcolor}{%
    \newif\ifGPcolor
    \GPcolorfalse
  }{}%
  \@ifundefined{ifGPblacktext}{%
    \newif\ifGPblacktext
    \GPblacktexttrue
  }{}%
  \let\gplgaddtomacro\g@addto@macro
  \gdef\gplbacktext{}%
  \gdef\gplfronttext{}%
  \makeatother
  \ifGPblacktext
    \def\colorrgb#1{}%
    \def\colorgray#1{}%
  \else
    \ifGPcolor
      \def\colorrgb#1{\color[rgb]{#1}}%
      \def\colorgray#1{\color[gray]{#1}}%
      \expandafter\def\csname LTw\endcsname{\color{white}}%
      \expandafter\def\csname LTb\endcsname{\color{black}}%
      \expandafter\def\csname LTa\endcsname{\color{black}}%
      \expandafter\def\csname LT0\endcsname{\color[rgb]{1,0,0}}%
      \expandafter\def\csname LT1\endcsname{\color[rgb]{0,1,0}}%
      \expandafter\def\csname LT2\endcsname{\color[rgb]{0,0,1}}%
      \expandafter\def\csname LT3\endcsname{\color[rgb]{1,0,1}}%
      \expandafter\def\csname LT4\endcsname{\color[rgb]{0,1,1}}%
      \expandafter\def\csname LT5\endcsname{\color[rgb]{1,1,0}}%
      \expandafter\def\csname LT6\endcsname{\color[rgb]{0,0,0}}%
      \expandafter\def\csname LT7\endcsname{\color[rgb]{1,0.3,0}}%
      \expandafter\def\csname LT8\endcsname{\color[rgb]{0.5,0.5,0.5}}%
    \else
      \def\colorrgb#1{\color{black}}%
      \def\colorgray#1{\color[gray]{#1}}%
      \expandafter\def\csname LTw\endcsname{\color{white}}%
      \expandafter\def\csname LTb\endcsname{\color{black}}%
      \expandafter\def\csname LTa\endcsname{\color{black}}%
      \expandafter\def\csname LT0\endcsname{\color{black}}%
      \expandafter\def\csname LT1\endcsname{\color{black}}%
      \expandafter\def\csname LT2\endcsname{\color{black}}%
      \expandafter\def\csname LT3\endcsname{\color{black}}%
      \expandafter\def\csname LT4\endcsname{\color{black}}%
      \expandafter\def\csname LT5\endcsname{\color{black}}%
      \expandafter\def\csname LT6\endcsname{\color{black}}%
      \expandafter\def\csname LT7\endcsname{\color{black}}%
      \expandafter\def\csname LT8\endcsname{\color{black}}%
    \fi
  \fi
    \setlength{\unitlength}{0.0500bp}%
    \ifx\gptboxheight\undefined%
      \newlength{\gptboxheight}%
      \newlength{\gptboxwidth}%
      \newsavebox{\gptboxtext}%
    \fi%
    \setlength{\fboxrule}{0.5pt}%
    \setlength{\fboxsep}{1pt}%
\begin{picture}(4320.00,2880.00)%
    \gplgaddtomacro\gplbacktext{%
      \csname LTb\endcsname%
      \put(592,512){\makebox(0,0)[r]{\strut{}$0$}}%
      \put(592,875){\makebox(0,0)[r]{\strut{}$20$}}%
      \put(592,1237){\makebox(0,0)[r]{\strut{}$40$}}%
      \put(592,1600){\makebox(0,0)[r]{\strut{}$60$}}%
      \put(592,1962){\makebox(0,0)[r]{\strut{}$80$}}%
      \put(592,2325){\makebox(0,0)[r]{\strut{}$100$}}%
      \put(592,2687){\makebox(0,0)[r]{\strut{}$120$}}%
      \put(688,352){\makebox(0,0){\strut{}40}}%
      \put(1524,352){\makebox(0,0){\strut{}50}}%
      \put(2360,352){\makebox(0,0){\strut{}60}}%
      \put(3195,352){\makebox(0,0){\strut{}70}}%
      \put(4031,352){\makebox(0,0){\strut{}80}}%
    }%
    \gplgaddtomacro\gplfronttext{%
      \csname LTb\endcsname%
      \put(128,1599){\rotatebox{-270}{\makebox(0,0){\strut{}\# page accesses}}}%
      \put(2359,112){\makebox(0,0){\strut{}Subgroup Size (\%)}}%
      \csname LTb\endcsname%
      \put(3296,2544){\makebox(0,0)[r]{\strut{}MFSNNK-BF (SUM)}}%
      \csname LTb\endcsname%
      \put(3296,2384){\makebox(0,0)[r]{\strut{}MFSNNK-R (SUM)}}%
      \csname LTb\endcsname%
      \put(3296,2224){\makebox(0,0)[r]{\strut{}MFSNNK-BF (MAX)}}%
      \csname LTb\endcsname%
      \put(3296,2064){\makebox(0,0)[r]{\strut{}MFSNNK-R (MAX)}}%
    }%
    \gplbacktext
    \put(0,0){\includegraphics{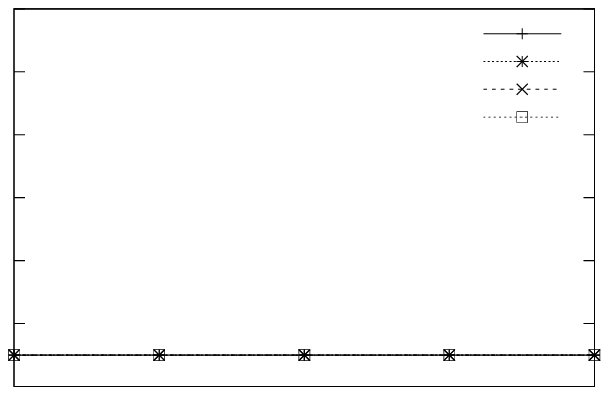}}%
    \gplfronttext
  \end{picture}%
\endgroup

		}
		\caption{}
		\label{graph:subgroup-mfsnnkr-Y-io}
	\end{subfigure}%
		\caption{The effect of varying subgroup size for Flickr (a-b) and Yelp (c-d) with and without using Heuristic~\ref{heuristic:mfsnnk2} in running time and I/O}
	\label{graph:subgroup-mfsnnker}
\end {figure}

\textbf{Effect of Heuristic~\ref{heuristic:mfsnnk2}.}
We have also implemented Heuristic~\ref{heuristic:mfsnnk2} for the MFSNNK query. In Figure~\ref{graph:subgroup-mfsnnker}, we show the performance of MFSNNK with and without using Heuristic~\ref{heuristic:mfsnnk2}, denoted by MFSNNK-R and MFSNNK-BF, respectively. We can see that, when different values of $m$ or different data sets are used, the algorithm may perform better or worse with Heuristic~\ref{heuristic:mfsnnk2}. Particularly, on the Flikr data set, the algorithm with the pruning heuristic works better when $40\%n < m < 80\%n$, and worse when $m \le 40\%n$ or $m \ge 80\%n$ (shown in Figure~\ref{graph:subgroup-mfsnnker} (a-b)). On the Yelp data set, the algorithm performance fluctuates more but the algorithm still performs better with the pruning heuristic in about half of the cases tested (shown in Figure~\ref{graph:subgroup-mfsnnker} (c-d)). This is expected since the heuristic sacrifices the tightness of the pruning bound for a more efficient computation of the pruning bound, as discussed in Section~\ref{subsection:fsnnk-extended}. Depending on different data sets and/or different values of $m$, this sacrifice may or may not be worthy. We would like to argue that, however, since the data sets are usually pre-known, we may empirically pre-test the heuristic performance under a set of different values of $m$ and $n$, and only activate the heuristic at query time if the queried group size falls in a pre-test range where the heuristic shows better performance.

\begin {figure}[!ht]
\setlength{\belowcaptionskip}{4pt}
\setlength{\abovecaptionskip}{4pt}
	\centering
	\begin{subfigure}[b]{0.5\linewidth}
		\centering
		\setlength{\abovecaptionskip}{2pt}
		\setlength{\belowcaptionskip}{2pt}
		\resizebox{\textwidth}{!}{
\begingroup
  \makeatletter
  \providecommand\color[2][]{%
    \GenericError{(gnuplot) \space\space\space\@spaces}{%
      Package color not loaded in conjunction with
      terminal option `colourtext'%
    }{See the gnuplot documentation for explanation.%
    }{Either use 'blacktext' in gnuplot or load the package
      color.sty in LaTeX.}%
    \renewcommand\color[2][]{}%
  }%
  \providecommand\includegraphics[2][]{%
    \GenericError{(gnuplot) \space\space\space\@spaces}{%
      Package graphicx or graphics not loaded%
    }{See the gnuplot documentation for explanation.%
    }{The gnuplot epslatex terminal needs graphicx.sty or graphics.sty.}%
    \renewcommand\includegraphics[2][]{}%
  }%
  \providecommand\rotatebox[2]{#2}%
  \@ifundefined{ifGPcolor}{%
    \newif\ifGPcolor
    \GPcolorfalse
  }{}%
  \@ifundefined{ifGPblacktext}{%
    \newif\ifGPblacktext
    \GPblacktexttrue
  }{}%
  \let\gplgaddtomacro\g@addto@macro
  \gdef\gplbacktext{}%
  \gdef\gplfronttext{}%
  \makeatother
  \ifGPblacktext
    \def\colorrgb#1{}%
    \def\colorgray#1{}%
  \else
    \ifGPcolor
      \def\colorrgb#1{\color[rgb]{#1}}%
      \def\colorgray#1{\color[gray]{#1}}%
      \expandafter\def\csname LTw\endcsname{\color{white}}%
      \expandafter\def\csname LTb\endcsname{\color{black}}%
      \expandafter\def\csname LTa\endcsname{\color{black}}%
      \expandafter\def\csname LT0\endcsname{\color[rgb]{1,0,0}}%
      \expandafter\def\csname LT1\endcsname{\color[rgb]{0,1,0}}%
      \expandafter\def\csname LT2\endcsname{\color[rgb]{0,0,1}}%
      \expandafter\def\csname LT3\endcsname{\color[rgb]{1,0,1}}%
      \expandafter\def\csname LT4\endcsname{\color[rgb]{0,1,1}}%
      \expandafter\def\csname LT5\endcsname{\color[rgb]{1,1,0}}%
      \expandafter\def\csname LT6\endcsname{\color[rgb]{0,0,0}}%
      \expandafter\def\csname LT7\endcsname{\color[rgb]{1,0.3,0}}%
      \expandafter\def\csname LT8\endcsname{\color[rgb]{0.5,0.5,0.5}}%
    \else
      \def\colorrgb#1{\color{black}}%
      \def\colorgray#1{\color[gray]{#1}}%
      \expandafter\def\csname LTw\endcsname{\color{white}}%
      \expandafter\def\csname LTb\endcsname{\color{black}}%
      \expandafter\def\csname LTa\endcsname{\color{black}}%
      \expandafter\def\csname LT0\endcsname{\color{black}}%
      \expandafter\def\csname LT1\endcsname{\color{black}}%
      \expandafter\def\csname LT2\endcsname{\color{black}}%
      \expandafter\def\csname LT3\endcsname{\color{black}}%
      \expandafter\def\csname LT4\endcsname{\color{black}}%
      \expandafter\def\csname LT5\endcsname{\color{black}}%
      \expandafter\def\csname LT6\endcsname{\color{black}}%
      \expandafter\def\csname LT7\endcsname{\color{black}}%
      \expandafter\def\csname LT8\endcsname{\color{black}}%
    \fi
  \fi
    \setlength{\unitlength}{0.0500bp}%
    \ifx\gptboxheight\undefined%
      \newlength{\gptboxheight}%
      \newlength{\gptboxwidth}%
      \newsavebox{\gptboxtext}%
    \fi%
    \setlength{\fboxrule}{0.5pt}%
    \setlength{\fboxsep}{1pt}%
\begin{picture}(4320.00,2880.00)%
    \gplgaddtomacro\gplbacktext{%
      \csname LTb\endcsname%
      \put(592,512){\makebox(0,0)[r]{\strut{}$0$}}%
      \put(592,754){\makebox(0,0)[r]{\strut{}$20$}}%
      \put(592,995){\makebox(0,0)[r]{\strut{}$40$}}%
      \put(592,1237){\makebox(0,0)[r]{\strut{}$60$}}%
      \put(592,1479){\makebox(0,0)[r]{\strut{}$80$}}%
      \put(592,1720){\makebox(0,0)[r]{\strut{}$100$}}%
      \put(592,1962){\makebox(0,0)[r]{\strut{}$120$}}%
      \put(592,2204){\makebox(0,0)[r]{\strut{}$140$}}%
      \put(592,2445){\makebox(0,0)[r]{\strut{}$160$}}%
      \put(592,2687){\makebox(0,0)[r]{\strut{}$180$}}%
      \put(688,352){\makebox(0,0){\strut{}10}}%
      \put(1166,352){\makebox(0,0){\strut{}20}}%
      \put(2121,352){\makebox(0,0){\strut{}40}}%
      \put(3076,352){\makebox(0,0){\strut{}60}}%
      \put(4031,352){\makebox(0,0){\strut{}80}}%
    }%
    \gplgaddtomacro\gplfronttext{%
      \csname LTb\endcsname%
      \put(128,1599){\rotatebox{-270}{\makebox(0,0){\strut{}running time (ms)}}}%
      \put(2359,112){\makebox(0,0){\strut{}Group Size}}%
      \csname LTb\endcsname%
      \put(3296,2544){\makebox(0,0)[r]{\strut{}GNNK-BB (SUM)}}%
      \csname LTb\endcsname%
      \put(3296,2384){\makebox(0,0)[r]{\strut{}GNNK-BF (SUM)}}%
      \csname LTb\endcsname%
      \put(3296,2224){\makebox(0,0)[r]{\strut{}GNNK-BB (MAX)}}%
      \csname LTb\endcsname%
      \put(3296,2064){\makebox(0,0)[r]{\strut{}GNNK-BF (MAX)}}%
    }%
    \gplbacktext
    \put(0,0){\includegraphics{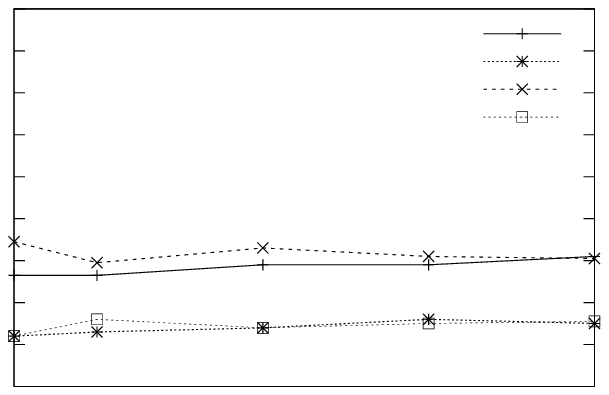}}%
    \gplfronttext
  \end{picture}%
\endgroup

		}
		\caption{}
		\label{graph:group-Y-run}
	\end{subfigure}%
	\begin{subfigure}[b]{0.5\linewidth}
		\centering
		\setlength{\abovecaptionskip}{2pt}
		\setlength{\belowcaptionskip}{2pt}
		\resizebox{\textwidth}{!}{
\begingroup
  \makeatletter
  \providecommand\color[2][]{%
    \GenericError{(gnuplot) \space\space\space\@spaces}{%
      Package color not loaded in conjunction with
      terminal option `colourtext'%
    }{See the gnuplot documentation for explanation.%
    }{Either use 'blacktext' in gnuplot or load the package
      color.sty in LaTeX.}%
    \renewcommand\color[2][]{}%
  }%
  \providecommand\includegraphics[2][]{%
    \GenericError{(gnuplot) \space\space\space\@spaces}{%
      Package graphicx or graphics not loaded%
    }{See the gnuplot documentation for explanation.%
    }{The gnuplot epslatex terminal needs graphicx.sty or graphics.sty.}%
    \renewcommand\includegraphics[2][]{}%
  }%
  \providecommand\rotatebox[2]{#2}%
  \@ifundefined{ifGPcolor}{%
    \newif\ifGPcolor
    \GPcolorfalse
  }{}%
  \@ifundefined{ifGPblacktext}{%
    \newif\ifGPblacktext
    \GPblacktexttrue
  }{}%
  \let\gplgaddtomacro\g@addto@macro
  \gdef\gplbacktext{}%
  \gdef\gplfronttext{}%
  \makeatother
  \ifGPblacktext
    \def\colorrgb#1{}%
    \def\colorgray#1{}%
  \else
    \ifGPcolor
      \def\colorrgb#1{\color[rgb]{#1}}%
      \def\colorgray#1{\color[gray]{#1}}%
      \expandafter\def\csname LTw\endcsname{\color{white}}%
      \expandafter\def\csname LTb\endcsname{\color{black}}%
      \expandafter\def\csname LTa\endcsname{\color{black}}%
      \expandafter\def\csname LT0\endcsname{\color[rgb]{1,0,0}}%
      \expandafter\def\csname LT1\endcsname{\color[rgb]{0,1,0}}%
      \expandafter\def\csname LT2\endcsname{\color[rgb]{0,0,1}}%
      \expandafter\def\csname LT3\endcsname{\color[rgb]{1,0,1}}%
      \expandafter\def\csname LT4\endcsname{\color[rgb]{0,1,1}}%
      \expandafter\def\csname LT5\endcsname{\color[rgb]{1,1,0}}%
      \expandafter\def\csname LT6\endcsname{\color[rgb]{0,0,0}}%
      \expandafter\def\csname LT7\endcsname{\color[rgb]{1,0.3,0}}%
      \expandafter\def\csname LT8\endcsname{\color[rgb]{0.5,0.5,0.5}}%
    \else
      \def\colorrgb#1{\color{black}}%
      \def\colorgray#1{\color[gray]{#1}}%
      \expandafter\def\csname LTw\endcsname{\color{white}}%
      \expandafter\def\csname LTb\endcsname{\color{black}}%
      \expandafter\def\csname LTa\endcsname{\color{black}}%
      \expandafter\def\csname LT0\endcsname{\color{black}}%
      \expandafter\def\csname LT1\endcsname{\color{black}}%
      \expandafter\def\csname LT2\endcsname{\color{black}}%
      \expandafter\def\csname LT3\endcsname{\color{black}}%
      \expandafter\def\csname LT4\endcsname{\color{black}}%
      \expandafter\def\csname LT5\endcsname{\color{black}}%
      \expandafter\def\csname LT6\endcsname{\color{black}}%
      \expandafter\def\csname LT7\endcsname{\color{black}}%
      \expandafter\def\csname LT8\endcsname{\color{black}}%
    \fi
  \fi
    \setlength{\unitlength}{0.0500bp}%
    \ifx\gptboxheight\undefined%
      \newlength{\gptboxheight}%
      \newlength{\gptboxwidth}%
      \newsavebox{\gptboxtext}%
    \fi%
    \setlength{\fboxrule}{0.5pt}%
    \setlength{\fboxsep}{1pt}%
\begin{picture}(4320.00,2880.00)%
    \gplgaddtomacro\gplbacktext{%
      \csname LTb\endcsname%
      \put(592,512){\makebox(0,0)[r]{\strut{}$0$}}%
      \put(592,784){\makebox(0,0)[r]{\strut{}$20$}}%
      \put(592,1056){\makebox(0,0)[r]{\strut{}$40$}}%
      \put(592,1328){\makebox(0,0)[r]{\strut{}$60$}}%
      \put(592,1600){\makebox(0,0)[r]{\strut{}$80$}}%
      \put(592,1871){\makebox(0,0)[r]{\strut{}$100$}}%
      \put(592,2143){\makebox(0,0)[r]{\strut{}$120$}}%
      \put(592,2415){\makebox(0,0)[r]{\strut{}$140$}}%
      \put(592,2687){\makebox(0,0)[r]{\strut{}$160$}}%
      \put(688,352){\makebox(0,0){\strut{}10}}%
      \put(1166,352){\makebox(0,0){\strut{}20}}%
      \put(2121,352){\makebox(0,0){\strut{}40}}%
      \put(3076,352){\makebox(0,0){\strut{}60}}%
      \put(4031,352){\makebox(0,0){\strut{}80}}%
    }%
    \gplgaddtomacro\gplfronttext{%
      \csname LTb\endcsname%
      \put(128,1599){\rotatebox{-270}{\makebox(0,0){\strut{}\# page accesses}}}%
      \put(2359,112){\makebox(0,0){\strut{}Group Size}}%
      \csname LTb\endcsname%
      \put(3296,2544){\makebox(0,0)[r]{\strut{}GNNK-BB (SUM)}}%
      \csname LTb\endcsname%
      \put(3296,2384){\makebox(0,0)[r]{\strut{}GNNK-BF (SUM)}}%
      \csname LTb\endcsname%
      \put(3296,2224){\makebox(0,0)[r]{\strut{}GNNK-BB (MAX)}}%
      \csname LTb\endcsname%
      \put(3296,2064){\makebox(0,0)[r]{\strut{}GNNK-BF (MAX)}}%
    }%
    \gplbacktext
    \put(0,0){\includegraphics{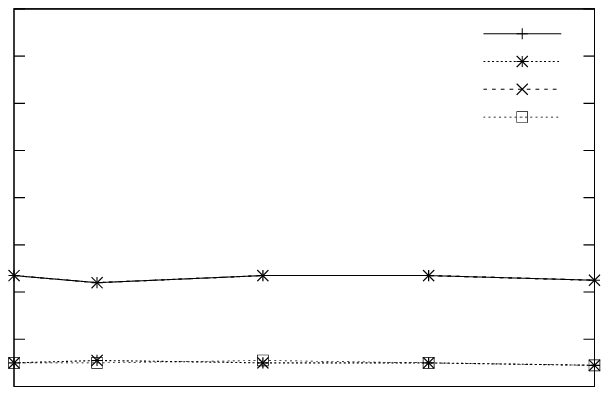}}%
    \gplfronttext
  \end{picture}%
\endgroup

		}
		\caption{}
		\label{graph:group-Y-io}
	\end{subfigure}%

	\begin{subfigure}[b]{0.5\linewidth}
		\centering
		\setlength{\abovecaptionskip}{2pt}
		\setlength{\belowcaptionskip}{2pt}
		\resizebox{\textwidth}{!}{
\begingroup
  \makeatletter
  \providecommand\color[2][]{%
    \GenericError{(gnuplot) \space\space\space\@spaces}{%
      Package color not loaded in conjunction with
      terminal option `colourtext'%
    }{See the gnuplot documentation for explanation.%
    }{Either use 'blacktext' in gnuplot or load the package
      color.sty in LaTeX.}%
    \renewcommand\color[2][]{}%
  }%
  \providecommand\includegraphics[2][]{%
    \GenericError{(gnuplot) \space\space\space\@spaces}{%
      Package graphicx or graphics not loaded%
    }{See the gnuplot documentation for explanation.%
    }{The gnuplot epslatex terminal needs graphicx.sty or graphics.sty.}%
    \renewcommand\includegraphics[2][]{}%
  }%
  \providecommand\rotatebox[2]{#2}%
  \@ifundefined{ifGPcolor}{%
    \newif\ifGPcolor
    \GPcolorfalse
  }{}%
  \@ifundefined{ifGPblacktext}{%
    \newif\ifGPblacktext
    \GPblacktexttrue
  }{}%
  \let\gplgaddtomacro\g@addto@macro
  \gdef\gplbacktext{}%
  \gdef\gplfronttext{}%
  \makeatother
  \ifGPblacktext
    \def\colorrgb#1{}%
    \def\colorgray#1{}%
  \else
    \ifGPcolor
      \def\colorrgb#1{\color[rgb]{#1}}%
      \def\colorgray#1{\color[gray]{#1}}%
      \expandafter\def\csname LTw\endcsname{\color{white}}%
      \expandafter\def\csname LTb\endcsname{\color{black}}%
      \expandafter\def\csname LTa\endcsname{\color{black}}%
      \expandafter\def\csname LT0\endcsname{\color[rgb]{1,0,0}}%
      \expandafter\def\csname LT1\endcsname{\color[rgb]{0,1,0}}%
      \expandafter\def\csname LT2\endcsname{\color[rgb]{0,0,1}}%
      \expandafter\def\csname LT3\endcsname{\color[rgb]{1,0,1}}%
      \expandafter\def\csname LT4\endcsname{\color[rgb]{0,1,1}}%
      \expandafter\def\csname LT5\endcsname{\color[rgb]{1,1,0}}%
      \expandafter\def\csname LT6\endcsname{\color[rgb]{0,0,0}}%
      \expandafter\def\csname LT7\endcsname{\color[rgb]{1,0.3,0}}%
      \expandafter\def\csname LT8\endcsname{\color[rgb]{0.5,0.5,0.5}}%
    \else
      \def\colorrgb#1{\color{black}}%
      \def\colorgray#1{\color[gray]{#1}}%
      \expandafter\def\csname LTw\endcsname{\color{white}}%
      \expandafter\def\csname LTb\endcsname{\color{black}}%
      \expandafter\def\csname LTa\endcsname{\color{black}}%
      \expandafter\def\csname LT0\endcsname{\color{black}}%
      \expandafter\def\csname LT1\endcsname{\color{black}}%
      \expandafter\def\csname LT2\endcsname{\color{black}}%
      \expandafter\def\csname LT3\endcsname{\color{black}}%
      \expandafter\def\csname LT4\endcsname{\color{black}}%
      \expandafter\def\csname LT5\endcsname{\color{black}}%
      \expandafter\def\csname LT6\endcsname{\color{black}}%
      \expandafter\def\csname LT7\endcsname{\color{black}}%
      \expandafter\def\csname LT8\endcsname{\color{black}}%
    \fi
  \fi
    \setlength{\unitlength}{0.0500bp}%
    \ifx\gptboxheight\undefined%
      \newlength{\gptboxheight}%
      \newlength{\gptboxwidth}%
      \newsavebox{\gptboxtext}%
    \fi%
    \setlength{\fboxrule}{0.5pt}%
    \setlength{\fboxsep}{1pt}%
\begin{picture}(4320.00,2880.00)%
    \gplgaddtomacro\gplbacktext{%
      \csname LTb\endcsname%
      \put(592,512){\makebox(0,0)[r]{\strut{}$0$}}%
      \put(592,754){\makebox(0,0)[r]{\strut{}$20$}}%
      \put(592,995){\makebox(0,0)[r]{\strut{}$40$}}%
      \put(592,1237){\makebox(0,0)[r]{\strut{}$60$}}%
      \put(592,1479){\makebox(0,0)[r]{\strut{}$80$}}%
      \put(592,1720){\makebox(0,0)[r]{\strut{}$100$}}%
      \put(592,1962){\makebox(0,0)[r]{\strut{}$120$}}%
      \put(592,2204){\makebox(0,0)[r]{\strut{}$140$}}%
      \put(592,2445){\makebox(0,0)[r]{\strut{}$160$}}%
      \put(592,2687){\makebox(0,0)[r]{\strut{}$180$}}%
      \put(688,352){\makebox(0,0){\strut{}40}}%
      \put(1524,352){\makebox(0,0){\strut{}50}}%
      \put(2360,352){\makebox(0,0){\strut{}60}}%
      \put(3195,352){\makebox(0,0){\strut{}70}}%
      \put(4031,352){\makebox(0,0){\strut{}80}}%
    }%
    \gplgaddtomacro\gplfronttext{%
      \csname LTb\endcsname%
      \put(128,1599){\rotatebox{-270}{\makebox(0,0){\strut{}running time (ms)}}}%
      \put(2359,112){\makebox(0,0){\strut{}Subgroup Size (\%)}}%
      \csname LTb\endcsname%
      \put(3296,2544){\makebox(0,0)[r]{\strut{}MFSNNK-N (SUM)}}%
      \csname LTb\endcsname%
      \put(3296,2384){\makebox(0,0)[r]{\strut{}MFSNNK-BF (SUM)}}%
      \csname LTb\endcsname%
      \put(3296,2224){\makebox(0,0)[r]{\strut{}MFSNNK-N (MAX)}}%
      \csname LTb\endcsname%
      \put(3296,2064){\makebox(0,0)[r]{\strut{}MFSNNK-BF (MAX)}}%
    }%
    \gplbacktext
    \put(0,0){\includegraphics{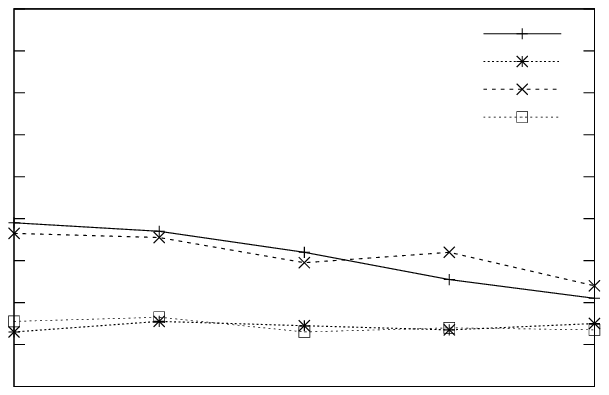}}%
    \gplfronttext
  \end{picture}%
\endgroup

		}
		\caption{}
		\label{graph:subgroup-mfsnnk-Y-run}
	\end{subfigure}%
	\begin{subfigure}[b]{0.5\linewidth}
		\centering
		\setlength{\abovecaptionskip}{2pt}
		\setlength{\belowcaptionskip}{2pt}
		\resizebox{\textwidth}{!}{
\begingroup
  \makeatletter
  \providecommand\color[2][]{%
    \GenericError{(gnuplot) \space\space\space\@spaces}{%
      Package color not loaded in conjunction with
      terminal option `colourtext'%
    }{See the gnuplot documentation for explanation.%
    }{Either use 'blacktext' in gnuplot or load the package
      color.sty in LaTeX.}%
    \renewcommand\color[2][]{}%
  }%
  \providecommand\includegraphics[2][]{%
    \GenericError{(gnuplot) \space\space\space\@spaces}{%
      Package graphicx or graphics not loaded%
    }{See the gnuplot documentation for explanation.%
    }{The gnuplot epslatex terminal needs graphicx.sty or graphics.sty.}%
    \renewcommand\includegraphics[2][]{}%
  }%
  \providecommand\rotatebox[2]{#2}%
  \@ifundefined{ifGPcolor}{%
    \newif\ifGPcolor
    \GPcolorfalse
  }{}%
  \@ifundefined{ifGPblacktext}{%
    \newif\ifGPblacktext
    \GPblacktexttrue
  }{}%
  \let\gplgaddtomacro\g@addto@macro
  \gdef\gplbacktext{}%
  \gdef\gplfronttext{}%
  \makeatother
  \ifGPblacktext
    \def\colorrgb#1{}%
    \def\colorgray#1{}%
  \else
    \ifGPcolor
      \def\colorrgb#1{\color[rgb]{#1}}%
      \def\colorgray#1{\color[gray]{#1}}%
      \expandafter\def\csname LTw\endcsname{\color{white}}%
      \expandafter\def\csname LTb\endcsname{\color{black}}%
      \expandafter\def\csname LTa\endcsname{\color{black}}%
      \expandafter\def\csname LT0\endcsname{\color[rgb]{1,0,0}}%
      \expandafter\def\csname LT1\endcsname{\color[rgb]{0,1,0}}%
      \expandafter\def\csname LT2\endcsname{\color[rgb]{0,0,1}}%
      \expandafter\def\csname LT3\endcsname{\color[rgb]{1,0,1}}%
      \expandafter\def\csname LT4\endcsname{\color[rgb]{0,1,1}}%
      \expandafter\def\csname LT5\endcsname{\color[rgb]{1,1,0}}%
      \expandafter\def\csname LT6\endcsname{\color[rgb]{0,0,0}}%
      \expandafter\def\csname LT7\endcsname{\color[rgb]{1,0.3,0}}%
      \expandafter\def\csname LT8\endcsname{\color[rgb]{0.5,0.5,0.5}}%
    \else
      \def\colorrgb#1{\color{black}}%
      \def\colorgray#1{\color[gray]{#1}}%
      \expandafter\def\csname LTw\endcsname{\color{white}}%
      \expandafter\def\csname LTb\endcsname{\color{black}}%
      \expandafter\def\csname LTa\endcsname{\color{black}}%
      \expandafter\def\csname LT0\endcsname{\color{black}}%
      \expandafter\def\csname LT1\endcsname{\color{black}}%
      \expandafter\def\csname LT2\endcsname{\color{black}}%
      \expandafter\def\csname LT3\endcsname{\color{black}}%
      \expandafter\def\csname LT4\endcsname{\color{black}}%
      \expandafter\def\csname LT5\endcsname{\color{black}}%
      \expandafter\def\csname LT6\endcsname{\color{black}}%
      \expandafter\def\csname LT7\endcsname{\color{black}}%
      \expandafter\def\csname LT8\endcsname{\color{black}}%
    \fi
  \fi
    \setlength{\unitlength}{0.0500bp}%
    \ifx\gptboxheight\undefined%
      \newlength{\gptboxheight}%
      \newlength{\gptboxwidth}%
      \newsavebox{\gptboxtext}%
    \fi%
    \setlength{\fboxrule}{0.5pt}%
    \setlength{\fboxsep}{1pt}%
\begin{picture}(4320.00,2880.00)%
    \gplgaddtomacro\gplbacktext{%
      \csname LTb\endcsname%
      \put(592,512){\makebox(0,0)[r]{\strut{}$0$}}%
      \put(592,754){\makebox(0,0)[r]{\strut{}$20$}}%
      \put(592,995){\makebox(0,0)[r]{\strut{}$40$}}%
      \put(592,1237){\makebox(0,0)[r]{\strut{}$60$}}%
      \put(592,1479){\makebox(0,0)[r]{\strut{}$80$}}%
      \put(592,1720){\makebox(0,0)[r]{\strut{}$100$}}%
      \put(592,1962){\makebox(0,0)[r]{\strut{}$120$}}%
      \put(592,2204){\makebox(0,0)[r]{\strut{}$140$}}%
      \put(592,2445){\makebox(0,0)[r]{\strut{}$160$}}%
      \put(592,2687){\makebox(0,0)[r]{\strut{}$180$}}%
      \put(688,352){\makebox(0,0){\strut{}40}}%
      \put(1524,352){\makebox(0,0){\strut{}50}}%
      \put(2360,352){\makebox(0,0){\strut{}60}}%
      \put(3195,352){\makebox(0,0){\strut{}70}}%
      \put(4031,352){\makebox(0,0){\strut{}80}}%
    }%
    \gplgaddtomacro\gplfronttext{%
      \csname LTb\endcsname%
      \put(128,1599){\rotatebox{-270}{\makebox(0,0){\strut{}\# page accesses}}}%
      \put(2359,112){\makebox(0,0){\strut{}Subgroup Size (\%)}}%
      \csname LTb\endcsname%
      \put(3296,2544){\makebox(0,0)[r]{\strut{}MFSNNK-N (SUM)}}%
      \csname LTb\endcsname%
      \put(3296,2384){\makebox(0,0)[r]{\strut{}MFSNNK-BF (SUM)}}%
      \csname LTb\endcsname%
      \put(3296,2224){\makebox(0,0)[r]{\strut{}MFSNNK-N (MAX)}}%
      \csname LTb\endcsname%
      \put(3296,2064){\makebox(0,0)[r]{\strut{}MFSNNK-BF (MAX)}}%
    }%
    \gplbacktext
    \put(0,0){\includegraphics{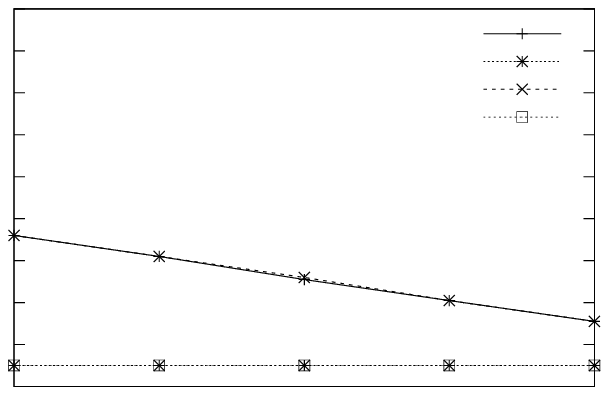}}%
    \gplfronttext
  \end{picture}%
\endgroup

		}
		\caption{}
		\label{graph:subgroup-mfsnnk-Y-io}
	\end{subfigure}%
		\caption{The effect of varying query group size (a-b) and minimum subgroup size (c-d) in running time and I/O}
	\label{graph:yelp}
\end {figure}

\subsection{Experiments on Yelp dataset}
We have run the same set of experiments as mentioned above on the Yelp dataset. All of our experimental results show similar trends in both datasets. Due to page limitations, we only present the experimental results for varying group size for GNNK queries and minimum subgroup size for MFSNNK queries with Yelp dataset in Figure~\ref{graph:yelp} (a-b) and Figure~\ref{graph:yelp} (c-d), respectively.

\section{Conclusion} \label{sec:conclusion}

We presented a new type of group spatial keyword query suitable for a collaborative environment. This query aims to find the best POI that minimizes the aggregate distance and maximizes the text relevancy for a group of users. We have studied three instances of this query, which return (i) the best POI for the whole group, (ii) 
the optimal subgroup with the best POI given a subgroup size $m$, and (iii) the optimal subgroups and the corresponding best POIs of different subgroup sizes in $m,m+1,...,n$. In all these queries, our proposed best-first approach runs approximately 4 times faster (on average) than the branch and bound approach for both real datasets. 

This study brings a number of future studies. 
For example, a study that allows users to set the value of $\alpha$ to reflect their preference of spatial proximity over textual relevance would make the query  more user friendly. 
Also, extending the algorithms to road networks would further improve their practicality.

\bibliographystyle{ACM-Reference-Format}
\bibliography{bibliography} 

\end{document}